\documentclass{article}

\usepackage[margin=1in,letterpaper]{geometry}

\usepackage[charter]{mathdesign}
\usepackage{microtype}
\usepackage{authblk}

\usepackage{enumitem}

\usepackage{amsmath}

\usepackage{amsthm}
\usepackage{cleveref}

\theoremstyle{plain}
\newtheorem{thm}{Theorem}[section]
\newtheorem{lem}[thm]{Lemma}

\newtheorem{obs}[thm]{Observation}
\theoremstyle{definition}

\crefname{section}{Section}{Sections}
\crefname{figure}{Figure}{Figures}
\crefname{lem}{Lemma}{Lemmas}
\crefname{thm}{Theorem}{Theorems}
\crefname{cor}{Corollary}{Corollaries}
\crefname{obs}{Observation}{Observations}
\crefname{prop}{property}{properties}

\newenvironment{problem}[1]{\begin{trivlist}\item\textsc{#1}.}{\end{trivlist}}

\usepackage[numbers,sort]{natbib}

\usepackage{subcaption,tikz}
\usetikzlibrary{arrows.meta,backgrounds,calc,decorations.pathmorphing,intersections,bending,fit}
\tikzset{
  x=7mm,y=7mm,
  >={Triangle[width=0pt 4 0,length=0pt 6 0]},
  vertex/.style={circle,fill=black,inner sep=0pt,minimum size=3pt,transform shape=false},
  edge/.style={draw=black,thick},
  external edge/.style={edge,densely dotted},
  horizontal/.style={draw,ultra thin,decorate,decoration={zigzag,segment length=2pt,amplitude=0.5pt}},
  vertical/.style={edge,->},
  choice/.style={edge,ultra thick},
  label edge/.style={draw=black,very thin,shorten <=1pt,shorten >=1pt},
  region/.style={draw=black!30,fill=black!20,rounded corners=4pt,inner sep=0pt},
  wool/.style={miter limit=1,decorate,decoration={bumps,segment length=4mm,amplitude=1mm}}
}

\tikzset{
  sheep vertices/.pic={
    \path
    let \n1={0.75},
        \n2={\n1/2},
        \n3={\n2/sin(36)} in
    (90:\n3)   node [vertex] (-v-1)  {}
    +(90:0.75) node [vertex] (-root) {}
    (162:\n3)  node [vertex] (-v-2)  {}
    (18:\n3)   node [vertex] (-v-3)  {}
    (234:\n3)  node [vertex] (-in-1) {}
    (306:\n3)  node [vertex] (-in-2) {};
  },
  sheep edges/.pic={
    \path [edge]
    let \n1={0.75},
        \n2={\n1/2},
        \n3={\n2/sin(36)} in
    (0,0) circle [radius=\n3]
    (-v-1) -- (-root)
    (-v-2) to [bend right] (-v-3);
  },
  sheep/.pic={
    \pic {sheep vertices};
    \pic {sheep edges};
  },
  sheep boolean/.pic={
    \pic {sheep vertices};
    \path [vertical] (-root) -- (sheep-v-1);
    \path [vertical,-{>[bend]}]
    let \n1={0.75},
        \n2={\n1/2},
        \n3={\n2/sin(36)} in
    (-v-1) arc [radius=\n3,start angle=90,end angle=154];
    \path [vertical,-{>[bend]}]
    let \n1={0.75},
        \n2={\n1/2},
        \n3={\n2/sin(36)} in
    (-v-1) arc [radius=\n3,start angle=90,end angle=26];
    \path [vertical,-{>[bend]}]
    let \n1={0.75},
        \n2={\n1/2},
        \n3={\n2/sin(36)} in
    (-v-2) arc [radius=\n3,start angle=162,end angle=226];
    \path [vertical,-{>[bend]}]
    let \n1={0.75},
        \n2={\n1/2},
        \n3={\n2/sin(36)} in
    (-v-3) arc [radius=\n3,start angle=378,end angle=314];
    \path [horizontal]
    let \n1={0.75},
        \n2={\n1/2},
        \n3={\n2/sin(36)} in
    (-v-2)  to [bend right] (sheep-v-3)
    (-in-1) arc [radius=\n3,start angle=234,end angle=308];
  }
}

\tikzset{
  real sheep/.pic={
    \path [name path=-body]
    let \n1={0.75},
        \n2={\n1/2},
        \n3={\n2/sin(36)} in
    (1.5+\n3/2,-1) arc [radius=1,start angle=-70,end angle=70]
    to [out=160,in=20,looseness=0.75] ++(180:3+\n3) arc [radius=1,start angle=110,end angle=250]
    to [out=340,in=200,looseness=0.75] cycle;
    \path [name path=-legs-1]
    let \n1={0.75},
        \n2={\n1/2},
        \n3={\n2/sin(36)} in
    (-1-\n3/2,0)  -- +(270:1.5)
    (0-\n3/2,0)   -- +(270:1.5)
    (0.5+\n3/2,0) -- +(270:1.5);
    \path [name intersections={of=-body and -legs-1}]
    (intersection-1) node [vertex] (-rear-leg-2) {}
    (intersection-2) node [vertex] (-p-1)        {}
    (intersection-3) node [vertex] (-in-1)       {};
    \path [name path=-legs-2]
    (-rear-leg-2) ++(150:0.75) -- +(240:0.5)
    (-in-1)       ++(15:0.75) -- +(285:0.5);
    \path [name intersections={of=-body and -legs-2}]
    (intersection-1) node [vertex] (-rear-leg-1) {}
    (intersection-2) node [vertex] (-in-2)       {};
    \path [name path=-head-path-2]
    (1.5,0) -- +(320:1.5);
    \path [name intersections={of=-body and -head-path-2}]
    (intersection-1) node [vertex] (-head-2) {};
    \path [name path=-head-path-1]
    (-head-2) ++(90:0.75) -- +(0:0.5);
    \path [name intersections={of=-body and -head-path-1}]
    (intersection-1) node [vertex] (-head-1) {};
    \path [name path=-head-path-3]
    (-head-1) ++(110:0.375) +(200:0.5) -- +(20:0.5);
    \path [name intersections={of=-body and -head-path-3}]
    (intersection-1) node [vertex] (-head-3) {};
    \path (-head-1) +(335:1.5) coordinate (-head-tip);
    \path [edge]
    (-head-1) .. controls +(0:1) and +(80:0.25) .. node [pos=0.4,vertex] (-head-4) {} (-head-tip)
    .. controls +(260:0.5) and +(340:1) .. node [pos=0.65,vertex] (-beard-1) {} (-head-2);
    \path [edge,wool,decoration={segment length=3mm,amplitude=0.75mm}]
    (-head-3) to [out=0,in=100] (-head-4);
    \path [edge]
    let \n1={0.75},
        \n2={\n1/2},
        \n3={\n2/sin(36)} in
    (1.5+\n3/2,-1) arc [radius=1,start angle=-70,end angle=35]
    decorate [wool,decoration={mirror}] {
      arc [radius=1,start angle=35,end angle=70]
      to [out=160,in=20,looseness=0.75] ++(180:3+\n3)
      arc [radius=1,start angle=110,end angle=210]
    }
    arc [radius=1,start angle=210,end angle=250]
    to [out=340,in=200,looseness=0.75] cycle;
    \path [edge]
    (-rear-leg-1) -- +(240:0.75) node [vertex] (-rear-leg-3) {}
    (-rear-leg-2) -- +(240:0.75) node [vertex] (-rear-leg-4) {}
    (-in-1)       -- +(285:0.75) node [vertex] (-in-3)       {}
    (-in-2)       -- +(285:0.75) node [vertex] (-in-4)       {}
    (-p-1)        -- +(270:1.25) node [vertex] (-p-2)        {}
    (-beard-1)    -- +(260:0.75) node [vertex] (-beard-2)    {};
    \path [edge]
    let \n1={0.75},
        \n2={\n1/2},
        \n3={\n2/sin(36)} in
    (-p-2) ++(270:\n3) circle [radius=\n3]
    +(162:\n3) node [vertex] (-p-3) {}
    +(234:\n3) node [vertex] (-p-4) {}
    +(18:\n3)  node [vertex] (-p-5) {}
    +(306:\n3) node [vertex] (-p-6) {}
    (-p-3) -- +(162:0.75) node [vertex] (-p-7) {}
    (-p-6) -- +(306:0.75) node [vertex] (-p-8) {}
    (-p-4) to [bend left] (-p-5);
  }
}

\tikzset{
  one-way vertices/.pic={
    \path
    let \n1={0.75},
        \n2={\n1/2},
        \n3={\n2/sin(36)} in
    (90:\n3)    node [vertex] (-v-1)   {} 
    ++(90:0.75) node [vertex] (-out-1) {}
    ++(0:0.75)  node [vertex] (-out-2) {}
    +(315:0.75) node [vertex] (-leaf)  {}
    (162:\n3)   node [vertex] (-v-2)   {}
    (18:\n3)    node [vertex] (-v-3)   {}
    (234:\n3)   node [vertex] (-in-1)  {}
    (306:\n3)   node [vertex] (-in-2)  {};
  },
  one-way edges/.pic={
    \path [edge]
    let \n1={0.75},
        \n2={\n1/2},
        \n3={\n2/sin(36)} in
    (0,0) circle [radius=\n3]
    (-v-1) -- (-out-1) -- (-out-2) -- (-leaf)
    (-v-2) to [bend right] (-v-3);
  },
  one-way/.pic={
    \pic {one-way vertices};
    \pic {one-way edges};
  },
  one-way boolean/.pic={
    \pic {one-way vertices};
    \path [vertical] (-out-1) -- (-v-1);
    \path [vertical] (-out-2) -- (-leaf);
    \path [vertical,-{>[bend]}]
    let \n1={0.75},
        \n2={\n1/2},
        \n3={\n2/sin(36)} in
    (-v-1) arc [radius=\n3,start angle=90,end angle=154];
    \path [vertical,-{>[bend]}]
    let \n1={0.75},
        \n2={\n1/2},
        \n3={\n2/sin(36)} in
    (-v-1) arc [radius=\n3,start angle=90,end angle=26];
    \path [vertical,-{>[bend]}]
    let \n1={0.75},
        \n2={\n1/2},
        \n3={\n2/sin(36)} in
    (-v-2) arc [radius=\n3,start angle=162,end angle=226];
    \path [vertical,-{>[bend]}]
    let \n1={0.75},
        \n2={\n1/2},
        \n3={\n2/sin(36)} in
    (-v-3) arc [radius=\n3,start angle=378,end angle=314];
    \path [horizontal]
    let \n1={0.75},
        \n2={\n1/2},
        \n3={\n2/sin(36)} in
    (-out-1) --               (-out-2)
    (-v-2)   to  [bend right] (-v-3)
    (-in-1)  arc [radius=\n3,start angle=234,end angle=308];
  }
}

\tikzset{
    connector vertices/.pic={
      \path
      (0,0)       node [vertex] (-in-1)  {}
      +(0:0.75)   node [vertex] (-in-2)  {}
      ++(90:0.75) node [vertex] (-v-1)   {}
      +(0:0.75)   node [vertex] (-v-2)   {}
      ++(90:0.75) node [vertex] (-out-1) {}
      +(0:0.75)   node [vertex] (-out-2) {};
    },
    connector/.pic={
      \pic {connector vertices};
      \path [edge]
      (-in-1)  -- (-in-2)
      (-v-1)   -- (-v-2)
      (-out-1) -- (-out-2)
      (-in-1)  -- (-v-1) -- (-out-1)
      (-in-2)  -- (-v-2) -- (-out-2);
    },
    connector true/.pic={
      \pic {connector vertices};
      \path [vertical]   (-out-1) -- (-v-1);
      \path [vertical]   (-out-2) -- (-v-2);
      \path [vertical]   (-v-1)   -- (-in-1);
      \path [vertical]   (-v-2)   -- (-in-2);
      \path [horizontal] (-v-1)   -- (-v-2);
    },
    connector false/.pic={
      \pic {connector vertices};
      \path [vertical]   (-in-1) -- (-v-1);
      \path [vertical]   (-in-2) -- (-v-2);
      \path [vertical]   (-v-1)  -- (-out-1);
      \path [vertical]   (-v-2)  -- (-out-2);
      \path [horizontal] (-v-1)   -- (-v-2);
    }
}

\tikzset{
  choice vertices/.pic={
    \path let \n1={0.75},
              \n2={\n1/2},
              \n3={\n2/sin(22.5)},
              \n4={\n2/sin(36)} in
    (180:\n3)                   node [vertex] (-true-start)  {}
    (0:\n3)                     node [vertex] (-false-start) {}
    (135:\n3)                   node [vertex] (-split-1)     {}
    (45:\n3)                    node [vertex] (-split-2)     {}
    (90:\n3)                    node [vertex] (-split-3)     {}
    +(90:0.75)                  node [vertex] (-split-4)     {}
    (247.5:\n3)                 node [vertex] (-in-1)        {}
    (292.5:\n3)                 node [vertex] (-in-2)        {}
    (-true-start)  ++(180:0.75) node [vertex] (-true-pu)     {}
                    +(270:0.75) node [vertex] (-true-u)      {}
                   ++(180:0.75) node [vertex] (-true-pv)     {}
                    +(270:0.75) node [vertex] (-true-v)      {}
                   ++(180:0.75) node [vertex] (-true-1)      {}
                   ++(90:0.75)  node [vertex] (-true-2)      {}
                    +(90:0.75)  node [vertex] (-true-3)      {}
    (-false-start) ++(0:0.75)   node [vertex] (-false-pu)    {}
                    +(270:0.75) node [vertex] (-false-u)     {}
                   ++(0:0.75)   node [vertex] (-false-pv)    {}
                    +(270:0.75) node [vertex] (-false-v)     {}
                   ++(0:0.75)   node [vertex] (-false-1)     {}
                   ++(90:0.75)  node [vertex] (-false-2)     {}
                    +(90:0.75)  node [vertex] (-false-3)     {};
  },
  choice/.pic={
    \pic {choice vertices};
    \path [edge]
    (0,0) circle [radius=1]
    (-split-1)  arc [radius=1,start angle=225,end angle=315]
    (-split-3)     -- (-split-4)
    (-true-start)  -- (-true-pu)  -- (-true-u)
    (-true-v)      -- (-true-pv)  -- (-true-1)  -- (-true-2)  -- (-true-3)
    (-false-start) -- (-false-pu) -- (-false-u)
    (-false-v)     -- (-false-pv) -- (-false-1) -- (-false-2) -- (-false-3);
    \path [choice]
    (-true-pv)  to coordinate [pos=0.5] (-true-e)  (-true-pu)
    (-false-pv) to coordinate [pos=0.5] (-false-e) (-false-pu);
  },
  choice true/.pic={
    \pic {choice vertices};
    \path [vertical] (-true-2)     -- (-true-3);
    \path [vertical] (-true-1)     -- (-true-pv);
    \path [vertical] (-true-pv)    -- (-true-v);
    \path [vertical] (-true-pu)    -- (-true-u);
    \path [vertical] (-true-start) -- (-true-pu);
    \path [vertical] (-false-2)    -- (-false-3);
    \path [vertical] (-false-1)    -- (-false-pv);
    \path [vertical] (-false-pv)   -- (-false-v);
    \path [vertical] (-false-pv)   -- (-false-pu);
    \path [vertical] (-false-pu)   -- (-false-u);
    \path [vertical] (-false-pu)   -- (-false-start);
    \path [vertical] (-split-3)    -- (-split-4);
    \path [vertical,-{>[bend]}]
    (-split-2) to [bend left] (-split-1);
    \path [vertical,-{>[bend]}]
    let \n1={0.75},
        \n2={\n1/2},
        \n3={\n2/sin(22.5)} in
    (-false-start) arc [radius=\n3,start angle=0,end angle=40];
    \path [vertical,-{>[bend]}]
    let \n1={0.75},
        \n2={\n1/2},
        \n3={\n2/sin(22.5)} in
    (-split-2) arc [radius=\n3,start angle=45,end angle=85];
    \path [vertical,-{>[bend]}]
    let \n1={0.75},
        \n2={\n1/2},
        \n3={\n2/sin(22.5)} in
    (-split-1) arc [radius=\n3,start angle=135,end angle=175];
    \path [vertical,-{>[bend]}]
    let \n1={0.75},
        \n2={\n1/2},
        \n3={\n2/sin(22.5)} in
    (-false-start) arc [radius=\n3,start angle=360,end angle=297.5];
    \path [vertical,-{>[bend]}]
    let \n1={0.75},
        \n2={\n1/2},
        \n3={\n2/sin(22.5)} in
    (-true-start) arc [radius=\n3,start angle=180,end angle=242.5];
    \path [horizontal]
    let \n1={0.75},
        \n2={\n1/2},
        \n3={\n2/sin(22.5)},
        \n4={\n2/sin(36)} in
    (-true-1)  --  (-true-2)
    (-false-1) --  (-false-2)
    (-true-pu) --  (-true-pv)
    (-split-3) arc [radius=\n3,start angle=90,end angle=135]
    (-in-1)    arc [radius=\n3,start angle=247.5,end angle=292.5];
  },
  choice false/.pic={
    \pic {choice vertices};
    \path [vertical] (-true-2)      -- (-true-3);
    \path [vertical] (-true-1)      -- (-true-pv);
    \path [vertical] (-true-pv)     -- (-true-v);
    \path [vertical] (-true-pu)     -- (-true-u);
    \path [vertical] (-false-start) -- (-false-pu);
    \path [vertical] (-false-2)     -- (-false-3);
    \path [vertical] (-false-1)     -- (-false-pv);
    \path [vertical] (-false-pv)    -- (-false-v);
    \path [vertical] (-true-pv)     -- (-true-pu);
    \path [vertical] (-false-pu)    -- (-false-u);
    \path [vertical] (-true-pu)     -- (-true-start);
    \path [vertical] (-split-3)     -- (-split-4);
    \path [vertical,-{>[bend]}]
    (-split-1) to [bend right] (-split-2);
    \path [vertical,-{>[bend]}]
    let \n1={0.75},
        \n2={\n1/2},
        \n3={\n2/sin(22.5)} in
    (-true-start) arc [radius=\n3,start angle=180,end angle=140];
    \path [vertical,-{>[bend]}]
    let \n1={0.75},
        \n2={\n1/2},
        \n3={\n2/sin(22.5)} in
    (-split-1) arc [radius=\n3,start angle=135,end angle=95];
    \path [vertical,-{>[bend]}]
    let \n1={0.75},
        \n2={\n1/2},
        \n3={\n2/sin(22.5)} in
    (-split-2) arc [radius=\n3,start angle=45,end angle=5];
    \path [vertical,-{>[bend]}]
    let \n1={0.75},
        \n2={\n1/2},
        \n3={\n2/sin(22.5)} in
    (-false-start) arc [radius=\n3,start angle=360,end angle=297.5];
    \path [vertical,-{>[bend]}]
    let \n1={0.75},
        \n2={\n1/2},
        \n3={\n2/sin(22.5)} in
    (-true-start) arc [radius=\n3,start angle=180,end angle=242.5];
    \path [horizontal]
    let \n1={0.75},
        \n2={\n1/2},
        \n3={\n2/sin(22.5)},
        \n4={\n2/sin(36)} in
    (-true-1)   --               (-true-2)
    (-false-1)  --               (-false-2)
    (-false-pu) --               (-false-pv)
    (-split-3)  arc [radius=\n3,start angle=90,end angle=45]
    (-in-1)     arc [radius=\n3,start angle=247.5,end angle=292.5];
  }
}

\tikzset{
  and vertices/.pic={
    \path let \n1={0.75},
              \n2={\n1/2},
              \n3={\n2/sin(22.5)} in
    (247.5:\n3)  node [vertex] (-in-1-v) {}
    ++(135:0.75) node [vertex] (-in-1-2) {}
    ++(90:0.75)  node [vertex] (-in-1-1) {}
    ++(45:0.75)  node [vertex] (-out-1)  {}
    ++(0:0.75)   node [vertex] (-out-2)  {}
    ++(315:0.75) node [vertex] (-in-2-1) {}
    ++(270:0.75) node [vertex] (-in-2-2) {}
    ++(225:0.75) node [vertex] (-in-2-v) {};
  },
  and/.pic={
    \pic {and vertices};
    \path [edge]
    (-in-1-v) -- (-in-1-2) -- (-in-1-1) -- (-out-1) -- (-out-2) --
    (-in-2-1) -- (-in-2-2) -- (-in-2-v);
  },
  and boolean/.pic={
    \pic {and vertices};
    \path [vertical]   (-out-1)  -- (-in-1-1);
    \path [vertical]   (-in-1-1) -- (-in-1-2);
    \path [vertical]   (-in-1-2) -- (-in-1-v);
    \path [vertical]   (-out-2)  -- (-in-2-1);
    \path [vertical]   (-in-2-1) -- (-in-2-2);
    \path [vertical]   (-in-2-2) -- (-in-2-v);
    \path [horizontal] (-out-1)  -- (-out-2);
  }
}

\tikzset{
  replicator vertices/.pic={
    \path let \n1={0.75},
              \n2={\n1/2},
              \n3={\n2/sin(22.5)} in
    (112.5:\n3)  node [vertex] (-out-1-v) {}
    ++(225:0.75) node [vertex] (-out-1-1) {}
    ++(270:0.75) node [vertex] (-out-1-2) {}
    ++(315:0.75) node [vertex] (-in-1)    {}
    ++(0:0.75)   node [vertex] (-in-2)    {}
    ++(45:0.75)  node [vertex] (-out-2-2) {}
    ++(90:0.75)  node [vertex] (-out-2-1) {}
    ++(135:0.75) node [vertex] (-out-2-v) {};
  },
  replicator/.pic={
    \pic {replicator vertices};
    \path [edge]
    (-out-1-v) -- (-out-1-1) -- (-out-1-2) -- (-in-1) -- (-in-2) --
    (-out-2-2) -- (-out-2-1) -- (-out-2-v);
  },
  replicator true/.pic={
    \pic {replicator vertices};
    \path [vertical]   (-out-1-1) -- (-out-1-v);
    \path [vertical]   (-out-1-2) -- (-in-1);
    \path [vertical]   (-out-2-1) -- (-out-2-v);
    \path [vertical]   (-out-2-2) -- (-in-2);
    \path [horizontal] (-out-1-1) -- (-out-1-2);
    \path [horizontal] (-out-2-1) -- (-out-2-2);
    \path [horizontal] (-in-1)    -- (-in-2);
  },
  replicator false/.pic={
    \pic {replicator vertices};
    \path [vertical]   (-out-1-1) -- (-out-1-v);
    \path [vertical]   (-out-1-2) -- (-out-1-1);
    \path [vertical]   (-in-1)    -- (-out-1-2);
    \path [vertical]   (-out-2-1) -- (-out-2-v);
    \path [vertical]   (-out-2-2) -- (-out-2-1);
    \path [vertical]   (-in-2)    -- (-out-2-2);
    \path [horizontal] (-in-1)    -- (-in-2);
  },
  replicator escape/.pic={
    \pic {replicator vertices};
    \path [vertical]   (-out-1-1) -- (-out-1-v);
    \path [vertical]   (-out-1-2) -- (-in-1);
    \path [vertical]   (-out-2-1) -- (-out-2-v);
    \path [vertical]   (-in-2)    -- (-out-2-2);
    \path [vertical]   (-in-1)    -- (-in-2);
    \path [vertical]   (-out-2-2) -- (-out-2-1);
    \path [horizontal] (-out-1-1) -- (-out-1-2);
  }
}

\tikzset{
  or vertices/.pic={
    \path let \n1={0.75},
              \n2={\n1/2},
              \n3={\n2/sin(18)} in
    (108:\n3)    node [vertex] (-out-1)        {}
    (72:\n3)     node [vertex] (-out-2)        {}
    (144:\n3)    node [vertex] (-left-ear-1)   {}
    +(144:0.75)  node [vertex] (-left-ear-2)   {}
    (180:\n3)    node [vertex] (-left-cross)   {}
    ++(270:0.75) node [vertex] (-in-1-1)       {}
    ++(270:0.75) node [vertex] (-in-1-2)       {}
    ++(270:0.75) node [vertex] (-left-foot-1)  {}
    ++(270:0.75) node [vertex] (-left-foot-2)  {}
    (36:\n3)     node [vertex] (-right-ear-1)  {}
    +(36:0.75)   node [vertex] (-right-ear-2)  {}
    (0:\n3)      node [vertex] (-right-cross)  {}
    ++(270:0.75) node [vertex] (-in-2-1)       {}
    ++(270:0.75) node [vertex] (-in-2-2)       {}
    ++(270:0.75) node [vertex] (-right-foot-1) {}
    ++(270:0.75) node [vertex] (-right-foot-2) {};
  },
  or/.pic={
    \pic {or vertices};
    \path [edge] let \n1={0.75},
                     \n2={\n1/2},
                     \n3={\n2/sin(18)} in
    (-right-cross)  -- (-left-foot-1) -- (-in-1-2) -- (-in-1-1) --
    (-left-cross)   arc [radius=\n3,start angle=180,end angle=0] --
    (-in-2-1)       -- (-in-2-2) -- (-right-foot-1) -- (-left-cross)
    (-left-ear-1)   -- (-left-ear-2)
    (-right-ear-1)  -- (-right-ear-2)
    (-left-foot-1)  -- (-left-foot-2)
    (-right-foot-1) -- (-right-foot-2);
  }
}

\tikzset{
  clause vertices/.pic={
    \path let \n1={0.75},
              \n2={\n1/2},
              \n3={\n2/sin(18)} in
    (108:\n3)       node [vertex] (-out-1)               {}
    (72:\n3)        node [vertex] (-out-2)               {}
    (144:\n3)       node [vertex] (-top-left-ear-1)      {}
    +(144:0.75)     node [vertex] (-top-left-ear-2)      {}
    (180:\n3)       node [vertex] (-top-left-cross)      {}
    ++(270:1.5)     coordinate (-bend-ref)
    ++(270:1.5)     node [vertex] (-inter-1)             {}
    ++(270:0.75)    node [vertex] (-inter-3)             {}
    ++(0:0.75)      node [vertex] (-inter-4)             {}
    ++(90:0.75)     node [vertex] (-inter-2)             {}
    ++(90:0.75)     node [vertex] (-top-left-foot-1)     {}
    +(0:0.75)       node [vertex] (-top-left-foot-2)     {}
    (36:\n3)        node [vertex] (-top-right-ear-1)     {}
    +(36:0.75)      node [vertex] (-top-right-ear-2)     {}
    (0:\n3)         node [vertex] (-top-right-cross)     {}
    ++(270:0.75)    node [vertex] (-in-3-1)              {}
    ++(270:0.75)    node [vertex] (-in-3-2)              {}
    ++(270:0.75)    node [vertex] (-top-right-foot-1)    {}
    ++(270:0.75)    node [vertex] (-top-right-foot-2)    {}
    (-inter-3) ++(288:\n3)
    +(144:\n3)      node [vertex] (-bottom-left-ear-1)   {}
    +(144:\n3+0.75) node [vertex] (-bottom-left-ear-2)   {}
    +(36:\n3)       node [vertex] (-bottom-right-ear-1)  {}
    +(36:\n3+0.75)  node [vertex] (-bottom-right-ear-2)  {}
    +(180:\n3)      node [vertex] (-bottom-left-cross)   {}
    +(0:\n3)        node [vertex] (-bottom-right-cross)  {}
    (-bottom-left-cross)
    ++(270:0.75)    node [vertex] (-in-1-1)              {}
    ++(270:0.75)    node [vertex] (-in-1-2)              {}
    ++(270:0.75)    node [vertex] (-bottom-left-foot-1)  {}
    ++(270:0.75)    node [vertex] (-bottom-left-foot-2)  {}
    (-bottom-right-cross)
    ++(270:0.75)    node [vertex] (-in-2-1)              {}
    ++(270:0.75)    node [vertex] (-in-2-2)              {}
    ++(270:0.75)    node [vertex] (-bottom-right-foot-1) {}
    ++(270:0.75)    node [vertex] (-bottom-right-foot-2) {};
  },
  clause/.pic={
    \pic {clause vertices};
    \path [edge] let \n1={0.75},
                     \n2={\n1/2},
                     \n3={\n2/sin(18)} in
    (-top-left-cross)      arc [radius=\n3,start angle=180,end angle=0]
    (-inter-3)             -- (-inter-1) -- (-top-left-cross) -- (-top-right-foot-1) -- (-in-3-2) -- (-in-3-1) -- (-top-right-cross)
    (-bottom-left-cross)   arc [radius=\n3,start angle=180,end angle=0]
    (-bottom-left-foot-1)  -- (-in-1-2) -- (-in-1-1) -- (-bottom-left-cross) --
    (-bottom-right-foot-1) -- (-in-2-2) -- (-in-2-1) -- (-bottom-right-cross)  -- (-bottom-left-foot-1)
    (-top-left-ear-1)      -- (-top-left-ear-2)
    (-top-right-ear-1)     -- (-top-right-ear-2)
    (-top-right-foot-1)    -- (-top-right-foot-2)
    (-inter-4)             -- (-inter-2) -- (-top-left-foot-1) -- (-top-left-foot-2)
    (-inter-1)             -- (-inter-2)
    (-bottom-left-ear-1)   -- (-bottom-left-ear-2)
    (-bottom-right-ear-1)  -- (-bottom-right-ear-2)
    (-bottom-left-foot-1)  -- (-bottom-left-foot-2)
    (-bottom-right-foot-1) -- (-bottom-right-foot-2);
    \path [name path=bend-ref] (-top-right-cross) -- (-bend-ref);
    \path [name path=inter-up] (-top-left-foot-1) -- +(90:2);
    \path [name intersections={of=bend-ref and inter-up},edge]
    let \p1 = (intersection-1),
        \p2 = (-top-left-foot-1),
        \p3 = (-bend-ref),
        \p4 = (-top-right-cross),
        \n1 = {\y4 - \y3},
        \n2 = {\x4 - \x3},
        \n3 = {atan2(\n1,\n2)},
        \n4 = {45-\n3/2},
        \n5 = {\y1 - \y2},
        \n6 = {\n5/tan(\n4)} in
    (-top-left-foot-1) arc [radius=\n6,start angle=180,end angle={\n3+90}] -- (-top-right-cross);
  },
  clause boolean-1/.pic={
    \pic {clause vertices};
    \path [vertical] (-top-left-ear-1)      -- (-top-left-ear-2);
    \path [vertical] (-top-right-ear-1)     -- (-top-right-ear-2);
    \path [vertical] (-bottom-left-ear-1)   -- (-bottom-left-ear-2);
    \path [vertical] (-bottom-right-ear-1)  -- (-bottom-right-ear-2);
    \path [vertical] (-bottom-left-cross)   -- (-in-1-1);
    \path [vertical] (-bottom-right-cross)  -- (-bottom-left-foot-1);
    \path [vertical] (-bottom-left-foot-1)  -- (-bottom-left-foot-2);
    \path [vertical] (-bottom-left-foot-1)  -- (-in-1-2);
    \path [vertical] (-in-2-1)              -- (-bottom-right-cross);
    \path [vertical] (-in-2-2)              -- (-bottom-right-foot-1);
    \path [vertical] (-bottom-right-foot-1) -- (-bottom-left-cross);
    \path [vertical] (-bottom-right-foot-1) -- (-bottom-right-foot-2);
    \path [vertical] (-top-left-cross)      -- (-inter-1);
    \path [vertical] (-top-right-foot-1)    -- (-top-left-cross);
    \path [vertical] (-in-3-1)              -- (-top-right-cross);
    \path [vertical] (-in-3-2)              -- (-top-right-foot-1);
    \path [vertical] (-top-right-foot-1)    -- (-top-right-foot-2);
    \path [vertical] (-top-left-foot-1)     -- (-top-left-foot-2);
    \path [vertical] (-top-left-foot-1)     -- (-inter-2);
    \path [vertical] (-inter-1)             -- (-inter-3);
    \path [vertical] (-inter-2)             -- (-inter-4);
    \path [vertical]
    let \n1={0.75},
        \n2={\n1/2},
        \n3={\n2/sin(18)} in
    (-inter-3) arc [radius=\n3,start angle=108,end angle=140];
    \path [vertical]
    let \n1={0.75},
        \n2={\n1/2},
        \n3={\n2/sin(18)} in
    (-inter-4) arc [radius=\n3,start angle=72,end angle=40];
    \path [vertical]
    let \n1={0.75},
        \n2={\n1/2},
        \n3={\n2/sin(18)} in
    (-out-1) arc [radius=\n3,start angle=108,end angle=140];
    \path [vertical]
    let \n1={0.75},
        \n2={\n1/2},
        \n3={\n2/sin(18)} in
    (-out-2) arc [radius=\n3,start angle=72,end angle=40];
    \path [name path=bend-ref] (-top-right-cross) -- (-bend-ref);
    \path [name path=inter-up] (-top-left-foot-1) -- +(90:2);
    \path [name intersections={of=bend-ref and inter-up},vertical,-{>[bend]}]
    let \p1 = (intersection-1),
        \p2 = (-top-left-foot-1),
        \p3 = (-bend-ref),
        \p4 = (-top-right-cross),
        \n1 = {\y4 - \y3},
        \n2 = {\x4 - \x3},
        \n3 = {atan2(\n1,\n2)},
        \n4 = {45-\n3/2},
        \n5 = {\y1 - \y2},
        \n6 = {\n5/tan(\n4)} in
    (-top-left-foot-1) ++(0:\n6) ++({\n3+90}:\n6) coordinate (-inter-2-bend)
    (-top-right-cross) -- (-inter-2-bend)
    arc [radius=\n6,start angle={\n3+90},end angle=178.5];
    \path [horizontal]
    let \n1={0.75},
        \n2={\n1/2},
        \n3={\n2/sin(18)} in
    (-in-1-1)             --  (-in-1-2)
    (-in-2-1)             --  (-in-2-2)
    (-in-3-1)             --  (-in-3-2)
    (-inter-1)            --  (-inter-2)
    (-inter-3)            arc [radius=\n3,start angle=108,end angle=72]
    (-bottom-left-ear-1)  arc [radius=\n3,start angle=144,end angle=180]
    (-bottom-right-ear-1) arc [radius=\n3,start angle=36,end angle=0]
    (-out-1)              arc [radius=\n3,start angle=108,end angle=72]
    (-top-left-ear-1)     arc [radius=\n3,start angle=144,end angle=180]
    (-top-right-ear-1)    arc [radius=\n3,start angle=36,end angle=0];
  },
  clause boolean-2/.pic={
    \pic {clause vertices};
    \path [vertical] (-top-left-ear-1)      -- (-top-left-ear-2);
    \path [vertical] (-top-right-ear-1)     -- (-top-right-ear-2);
    \path [vertical] (-bottom-left-ear-1)   -- (-bottom-left-ear-2);
    \path [vertical] (-bottom-right-ear-1)  -- (-bottom-right-ear-2);
    \path [vertical] (-in-1-1)              -- (-bottom-left-cross);
    \path [vertical] (-bottom-left-foot-1)  -- (-bottom-right-cross);
    \path [vertical] (-bottom-left-foot-1)  -- (-bottom-left-foot-2);
    \path [vertical] (-in-1-2)              -- (-bottom-left-foot-1);
    \path [vertical] (-bottom-right-cross)  -- (-in-2-1);
    \path [vertical] (-bottom-right-foot-1) -- (-in-2-2);
    \path [vertical] (-bottom-left-cross)   -- (-bottom-right-foot-1);
    \path [vertical] (-bottom-right-foot-1) -- (-bottom-right-foot-2);
    \path [vertical] (-top-left-cross)      -- (-inter-1);
    \path [vertical] (-top-right-foot-1)    -- (-top-left-cross);
    \path [vertical] (-in-3-1)              -- (-top-right-cross);
    \path [vertical] (-in-3-2)              -- (-top-right-foot-1);
    \path [vertical] (-top-right-foot-1)    -- (-top-right-foot-2);
    \path [vertical] (-top-left-foot-1)     -- (-top-left-foot-2);
    \path [vertical] (-top-left-foot-1)     -- (-inter-2);
    \path [vertical] (-inter-1)             -- (-inter-3);
    \path [vertical] (-inter-2)             -- (-inter-4);
    \path [vertical]
    let \n1={0.75},
        \n2={\n1/2},
        \n3={\n2/sin(18)} in
    (-inter-3) arc [radius=\n3,start angle=108,end angle=140];
    \path [vertical]
    let \n1={0.75},
        \n2={\n1/2},
        \n3={\n2/sin(18)} in
    (-inter-4) arc [radius=\n3,start angle=72,end angle=40];
    \path [vertical]
    let \n1={0.75},
        \n2={\n1/2},
        \n3={\n2/sin(18)} in
    (-out-1) arc [radius=\n3,start angle=108,end angle=140];
    \path [vertical]
    let \n1={0.75},
        \n2={\n1/2},
        \n3={\n2/sin(18)} in
    (-out-2) arc [radius=\n3,start angle=72,end angle=40];
    \path [name path=bend-ref] (-top-right-cross) -- (-bend-ref);
    \path [name path=inter-up] (-top-left-foot-1) -- +(90:2);
    \path [name intersections={of=bend-ref and inter-up},vertical,-{>[bend]}]
    let \p1 = (intersection-1),
        \p2 = (-top-left-foot-1),
        \p3 = (-bend-ref),
        \p4 = (-top-right-cross),
        \n1 = {\y4 - \y3},
        \n2 = {\x4 - \x3},
        \n3 = {atan2(\n1,\n2)},
        \n4 = {45-\n3/2},
        \n5 = {\y1 - \y2},
        \n6 = {\n5/tan(\n4)} in
    (-top-left-foot-1) ++(0:\n6) ++({\n3+90}:\n6) coordinate (-inter-2-bend)
    (-top-right-cross) -- (-inter-2-bend)
    arc [radius=\n6,start angle={\n3+90},end angle=178.5];
    \path [horizontal]
    let \n1={0.75},
        \n2={\n1/2},
        \n3={\n2/sin(18)} in
    (-in-1-1)             --  (-in-1-2)
    (-in-2-1)             --  (-in-2-2)
    (-in-3-1)             --  (-in-3-2)
    (-inter-1)            --  (-inter-2)
    (-inter-3)            arc [radius=\n3,start angle=108,end angle=72]
    (-bottom-left-ear-1)  arc [radius=\n3,start angle=144,end angle=180]
    (-bottom-right-ear-1) arc [radius=\n3,start angle=36,end angle=0]
    (-out-1)              arc [radius=\n3,start angle=108,end angle=72]
    (-top-left-ear-1)     arc [radius=\n3,start angle=144,end angle=180]
    (-top-right-ear-1)    arc [radius=\n3,start angle=36,end angle=0];
  },
  clause boolean-3/.pic={
    \pic {clause vertices};
    \path [vertical] (-top-left-ear-1)      -- (-top-left-ear-2);
    \path [vertical] (-top-right-ear-1)     -- (-top-right-ear-2);
    \path [vertical] (-bottom-left-ear-1)   -- (-bottom-left-ear-2);
    \path [vertical] (-bottom-right-ear-1)  -- (-bottom-right-ear-2);
    \path [vertical] (-bottom-left-foot-1)  -- (-bottom-left-foot-2);
    \path [vertical] (-in-1-1)              -- (-bottom-left-cross);
    \path [vertical] (-in-1-2)              -- (-bottom-left-foot-1);
    \path [vertical] (-in-2-1)              -- (-bottom-right-cross);
    \path [vertical] (-in-2-2)              -- (-bottom-right-foot-1);
    \path [vertical] (-bottom-right-foot-1) -- (-bottom-right-foot-2);
    \path [vertical] (-inter-1)             -- (-top-left-cross);
    \path [vertical] (-top-left-cross)      -- (-top-right-foot-1);
    \path [vertical] (-top-right-cross)     -- (-in-3-1);
    \path [vertical] (-top-right-foot-1)    -- (-in-3-2);
    \path [vertical] (-top-right-foot-1)    -- (-top-right-foot-2);
    \path [vertical] (-top-left-foot-1)     -- (-top-left-foot-2);
    \path [vertical] (-inter-2)             -- (-top-left-foot-1);
    \path [vertical] (-inter-3)             -- (-inter-1);
    \path [vertical] (-inter-4)             -- (-inter-2);
    \path [vertical]
    let \n1={0.75},
        \n2={\n1/2},
        \n3={\n2/sin(18)} in
    (-bottom-left-ear-1) arc [radius=\n3,start angle=144,end angle=112];
    \path [vertical]
    let \n1={0.75},
        \n2={\n1/2},
        \n3={\n2/sin(18)} in
    (-bottom-right-ear-1) arc [radius=\n3,start angle=36,end angle=68];
    \path [vertical]
    let \n1={0.75},
        \n2={\n1/2},
        \n3={\n2/sin(18)} in
    (-out-1) arc [radius=\n3,start angle=108,end angle=140];
    \path [vertical]
    let \n1={0.75},
        \n2={\n1/2},
        \n3={\n2/sin(18)} in
    (-out-2) arc [radius=\n3,start angle=72,end angle=40];
    \path [vertical]
    let \n1={0.75},
        \n2={\n1/2},
        \n3={\n2/sin(18)} in
    (-bottom-left-cross) arc [radius=\n3,start angle=180,end angle=148];
    \path [vertical]
    let \n1={0.75},
        \n2={\n1/2},
        \n3={\n2/sin(18)} in
    (-bottom-right-cross) arc [radius=\n3,start angle=0,end angle=32];
    \path [name path=bend-ref] (-top-right-cross) -- (-bend-ref);
    \path [name path=inter-up] (-top-left-foot-1) -- +(90:2);
    \path [name intersections={of=bend-ref and inter-up},vertical]
    let \p1 = (intersection-1),
        \p2 = (-top-left-foot-1),
        \p3 = (-bend-ref),
        \p4 = (-top-right-cross),
        \n1 = {\y4 - \y3},
        \n2 = {\x4 - \x3},
        \n3 = {atan2(\n1,\n2)},
        \n4 = {45-\n3/2},
        \n5 = {\y1 - \y2},
        \n6 = {\n5/tan(\n4)} in
    (-top-left-foot-1) arc [radius=\n6,start angle=180,end angle={\n3+90}] -- (-top-right-cross);
    \path [horizontal]
    let \n1={0.75},
        \n2={\n1/2},
        \n3={\n2/sin(18)} in
    (-in-1-1)             --  (-in-1-2)
    (-in-2-1)             --  (-in-2-2)
    (-in-3-1)             --  (-in-3-2)
    (-bottom-left-cross)  --  (-bottom-right-foot-1)
    (-bottom-right-cross) --  (-bottom-left-foot-1)
    (-inter-1)            --  (-inter-2)
    (-inter-3)            arc [radius=\n3,start angle=108,end angle=72]
    (-out-1)              arc [radius=\n3,start angle=108,end angle=72]
    (-top-left-ear-1)     arc [radius=\n3,start angle=144,end angle=180]
    (-top-right-ear-1)    arc [radius=\n3,start angle=36,end angle=0];
  },
  clause boolean-1-2/.pic={
    \pic {clause vertices};
    \path [vertical] (-top-left-ear-1)      -- (-top-left-ear-2);
    \path [vertical] (-top-right-ear-1)     -- (-top-right-ear-2);
    \path [vertical] (-bottom-left-ear-1)   -- (-bottom-left-ear-2);
    \path [vertical] (-bottom-right-ear-1)  -- (-bottom-right-ear-2);
    \path [vertical] (-bottom-left-cross)   -- (-in-1-1);
    \path [vertical] (-bottom-right-cross)  -- (-bottom-left-foot-1);
    \path [vertical] (-bottom-left-foot-1)  -- (-bottom-left-foot-2);
    \path [vertical] (-bottom-left-foot-1)  -- (-in-1-2);
    \path [vertical] (-bottom-right-cross)  -- (-in-2-1);
    \path [vertical] (-bottom-right-foot-1) -- (-in-2-2);
    \path [vertical] (-bottom-left-cross)   -- (-bottom-right-foot-1);
    \path [vertical] (-bottom-right-foot-1) -- (-bottom-right-foot-2);
    \path [vertical] (-top-left-cross)      -- (-inter-1);
    \path [vertical] (-top-right-foot-1)    -- (-top-left-cross);
    \path [vertical] (-in-3-1)              -- (-top-right-cross);
    \path [vertical] (-in-3-2)              -- (-top-right-foot-1);
    \path [vertical] (-top-right-foot-1)    -- (-top-right-foot-2);
    \path [vertical] (-top-left-foot-1)     -- (-top-left-foot-2);
    \path [vertical] (-top-left-foot-1)     -- (-inter-2);
    \path [vertical] (-inter-1)             -- (-inter-3);
    \path [vertical] (-inter-2)             -- (-inter-4);
    \path [vertical]
    let \n1={0.75},
        \n2={\n1/2},
        \n3={\n2/sin(18)} in
    (-inter-3) arc [radius=\n3,start angle=108,end angle=140];
    \path [vertical]
    let \n1={0.75},
        \n2={\n1/2},
        \n3={\n2/sin(18)} in
    (-inter-4) arc [radius=\n3,start angle=72,end angle=40];
    \path [vertical]
    let \n1={0.75},
        \n2={\n1/2},
        \n3={\n2/sin(18)} in
    (-out-1) arc [radius=\n3,start angle=108,end angle=140];
    \path [vertical]
    let \n1={0.75},
        \n2={\n1/2},
        \n3={\n2/sin(18)} in
    (-out-2) arc [radius=\n3,start angle=72,end angle=40];
    \path [vertical]
    let \n1={0.75},
        \n2={\n1/2},
        \n3={\n2/sin(18)} in
    (-bottom-left-ear-1) arc [radius=\n3,start angle=144,end angle=176];
    \path [vertical]
    let \n1={0.75},
        \n2={\n1/2},
        \n3={\n2/sin(18)} in
    (-bottom-right-ear-1) arc [radius=\n3,start angle=36,end angle=4];
    \path [name path=bend-ref] (-top-right-cross) -- (-bend-ref);
    \path [name path=inter-up] (-top-left-foot-1) -- +(90:2);
    \path [name intersections={of=bend-ref and inter-up},vertical,-{>[bend]}]
    let \p1 = (intersection-1),
        \p2 = (-top-left-foot-1),
        \p3 = (-bend-ref),
        \p4 = (-top-right-cross),
        \n1 = {\y4 - \y3},
        \n2 = {\x4 - \x3},
        \n3 = {atan2(\n1,\n2)},
        \n4 = {45-\n3/2},
        \n5 = {\y1 - \y2},
        \n6 = {\n5/tan(\n4)} in
    (-top-left-foot-1) ++(0:\n6) ++({\n3+90}:\n6) coordinate (-inter-2-bend)
    (-top-right-cross) -- (-inter-2-bend)
    arc [radius=\n6,start angle={\n3+90},end angle=178.5];
    \path [horizontal]
    let \n1={0.75},
        \n2={\n1/2},
        \n3={\n2/sin(18)} in
    (-in-1-1)          --  (-in-1-2)
    (-in-2-1)          --  (-in-2-2)
    (-in-3-1)          --  (-in-3-2)
    (-inter-1)         --  (-inter-2)
    (-inter-3)         arc [radius=\n3,start angle=108,end angle=72]
    (-out-1)           arc [radius=\n3,start angle=108,end angle=72]
    (-top-left-ear-1)  arc [radius=\n3,start angle=144,end angle=180]
    (-top-right-ear-1) arc [radius=\n3,start angle=36,end angle=0];
  },
  clause boolean-1-3/.pic={
    \pic {clause vertices};
    \path [vertical] (-top-left-ear-1)      -- (-top-left-ear-2);
    \path [vertical] (-top-right-ear-1)     -- (-top-right-ear-2);
    \path [vertical] (-bottom-left-ear-1)   -- (-bottom-left-ear-2);
    \path [vertical] (-bottom-right-ear-1)  -- (-bottom-right-ear-2);
    \path [vertical] (-bottom-left-cross)   -- (-in-1-1);
    \path [vertical] (-bottom-right-cross)  -- (-bottom-left-foot-1);
    \path [vertical] (-bottom-left-foot-1)  -- (-bottom-left-foot-2);
    \path [vertical] (-bottom-left-foot-1)  -- (-in-1-2);
    \path [vertical] (-in-2-1)              -- (-bottom-right-cross);
    \path [vertical] (-in-2-2)              -- (-bottom-right-foot-1);
    \path [vertical] (-bottom-right-foot-1) -- (-bottom-left-cross);
    \path [vertical] (-bottom-right-foot-1) -- (-bottom-right-foot-2);
    \path [vertical] (-top-left-cross)      -- (-inter-1);
    \path [vertical] (-top-left-cross)      -- (-top-right-foot-1);
    \path [vertical] (-top-right-cross)     -- (-in-3-1);
    \path [vertical] (-top-right-foot-1)    -- (-in-3-2);
    \path [vertical] (-top-right-foot-1)    -- (-top-right-foot-2);
    \path [vertical] (-top-left-foot-1)     -- (-top-left-foot-2);
    \path [vertical] (-top-left-foot-1)     -- (-inter-2);
    \path [vertical] (-inter-1)             -- (-inter-3);
    \path [vertical] (-inter-2)             -- (-inter-4);
    \path [vertical]
    let \n1={0.75},
        \n2={\n1/2},
        \n3={\n2/sin(18)} in
    (-inter-3) arc [radius=\n3,start angle=108,end angle=140];
    \path [vertical]
    let \n1={0.75},
        \n2={\n1/2},
        \n3={\n2/sin(18)} in
    (-inter-4) arc [radius=\n3,start angle=72,end angle=40];
    \path [vertical]
    let \n1={0.75},
        \n2={\n1/2},
        \n3={\n2/sin(18)} in
    (-out-1) arc [radius=\n3,start angle=108,end angle=140];
    \path [vertical]
    let \n1={0.75},
        \n2={\n1/2},
        \n3={\n2/sin(18)} in
    (-out-2) arc [radius=\n3,start angle=72,end angle=40];
    \path [vertical]
    let \n1={0.75},
        \n2={\n1/2},
        \n3={\n2/sin(18)} in
    (-top-left-ear-1) arc [radius=\n3,start angle=144,end angle=176];
    \path [vertical]
    let \n1={0.75},
        \n2={\n1/2},
        \n3={\n2/sin(18)} in
    (-top-right-ear-1) arc [radius=\n3,start angle=36,end angle=4];
    \path [name path=bend-ref] (-top-right-cross) -- (-bend-ref);
    \path [name path=inter-up] (-top-left-foot-1) -- +(90:2);
    \path [name intersections={of=bend-ref and inter-up},vertical,-{>[bend]}]
    let \p1 = (intersection-1),
        \p2 = (-top-left-foot-1),
        \p3 = (-bend-ref),
        \p4 = (-top-right-cross),
        \n1 = {\y4 - \y3},
        \n2 = {\x4 - \x3},
        \n3 = {atan2(\n1,\n2)},
        \n4 = {45-\n3/2},
        \n5 = {\y1 - \y2},
        \n6 = {\n5/tan(\n4)} in
    (-top-left-foot-1) ++(0:\n6) ++({\n3+90}:\n6) coordinate (-inter-2-bend)
    (-top-right-cross) -- (-inter-2-bend)
    arc [radius=\n6,start angle={\n3+90},end angle=178.5];
    \path [horizontal]
    let \n1={0.75},
        \n2={\n1/2},
        \n3={\n2/sin(18)} in
    (-in-1-1)             --  (-in-1-2)
    (-in-2-1)             --  (-in-2-2)
    (-in-3-1)             --  (-in-3-2)
    (-inter-1)            --  (-inter-2)
    (-inter-3)            arc [radius=\n3,start angle=108,end angle=72]
    (-out-1)              arc [radius=\n3,start angle=108,end angle=72]
    (-bottom-left-ear-1)  arc [radius=\n3,start angle=144,end angle=180]
    (-bottom-right-ear-1) arc [radius=\n3,start angle=36,end angle=0];
  },
  clause boolean-2-3/.pic={
    \pic {clause vertices};
    \path [vertical] (-top-left-ear-1)      -- (-top-left-ear-2);
    \path [vertical] (-top-right-ear-1)     -- (-top-right-ear-2);
    \path [vertical] (-bottom-left-ear-1)   -- (-bottom-left-ear-2);
    \path [vertical] (-bottom-right-ear-1)  -- (-bottom-right-ear-2);
    \path [vertical] (-in-1-1)              -- (-bottom-left-cross);
    \path [vertical] (-bottom-left-foot-1)  -- (-bottom-right-cross);
    \path [vertical] (-bottom-left-foot-1)  -- (-bottom-left-foot-2);
    \path [vertical] (-in-1-2)              -- (-bottom-left-foot-1);
    \path [vertical] (-bottom-right-cross)  -- (-in-2-1);
    \path [vertical] (-bottom-right-foot-1) -- (-in-2-2);
    \path [vertical] (-bottom-left-cross)   -- (-bottom-right-foot-1);
    \path [vertical] (-bottom-right-foot-1) -- (-bottom-right-foot-2);
    \path [vertical] (-top-left-cross)      -- (-inter-1);
    \path [vertical] (-top-left-cross)      -- (-top-right-foot-1);
    \path [vertical] (-top-right-cross)     -- (-in-3-1);
    \path [vertical] (-top-right-foot-1)    -- (-in-3-2);
    \path [vertical] (-top-right-foot-1)    -- (-top-right-foot-2);
    \path [vertical] (-top-left-foot-1)     -- (-top-left-foot-2);
    \path [vertical] (-top-left-foot-1)     -- (-inter-2);
    \path [vertical] (-inter-1)             -- (-inter-3);
    \path [vertical] (-inter-2)             -- (-inter-4);
    \path [vertical]
    let \n1={0.75},
        \n2={\n1/2},
        \n3={\n2/sin(18)} in
    (-inter-3) arc [radius=\n3,start angle=108,end angle=140];
    \path [vertical]
    let \n1={0.75},
        \n2={\n1/2},
        \n3={\n2/sin(18)} in
    (-inter-4) arc [radius=\n3,start angle=72,end angle=40];
    \path [vertical]
    let \n1={0.75},
        \n2={\n1/2},
        \n3={\n2/sin(18)} in
    (-out-1) arc [radius=\n3,start angle=108,end angle=140];
    \path [vertical]
    let \n1={0.75},
        \n2={\n1/2},
        \n3={\n2/sin(18)} in
    (-out-2) arc [radius=\n3,start angle=72,end angle=40];
    \path [vertical]
    let \n1={0.75},
        \n2={\n1/2},
        \n3={\n2/sin(18)} in
    (-top-left-ear-1) arc [radius=\n3,start angle=144,end angle=176];
    \path [vertical]
    let \n1={0.75},
        \n2={\n1/2},
        \n3={\n2/sin(18)} in
    (-top-right-ear-1) arc [radius=\n3,start angle=36,end angle=4];
    \path [name path=bend-ref] (-top-right-cross) -- (-bend-ref);
    \path [name path=inter-up] (-top-left-foot-1) -- +(90:2);
    \path [name intersections={of=bend-ref and inter-up},vertical,-{>[bend]}]
    let \p1 = (intersection-1),
        \p2 = (-top-left-foot-1),
        \p3 = (-bend-ref),
        \p4 = (-top-right-cross),
        \n1 = {\y4 - \y3},
        \n2 = {\x4 - \x3},
        \n3 = {atan2(\n1,\n2)},
        \n4 = {45-\n3/2},
        \n5 = {\y1 - \y2},
        \n6 = {\n5/tan(\n4)} in
    (-top-left-foot-1) ++(0:\n6) ++({\n3+90}:\n6) coordinate (-inter-2-bend)
    (-top-right-cross) -- (-inter-2-bend)
    arc [radius=\n6,start angle={\n3+90},end angle=178.5];
    \path [horizontal]
    let \n1={0.75},
        \n2={\n1/2},
        \n3={\n2/sin(18)} in
    (-in-1-1)             --  (-in-1-2)
    (-in-2-1)             --  (-in-2-2)
    (-in-3-1)             --  (-in-3-2)
    (-inter-1)            --  (-inter-2)
    (-inter-3)            arc [radius=\n3,start angle=108,end angle=72]
    (-out-1)              arc [radius=\n3,start angle=108,end angle=72]
    (-bottom-left-ear-1)  arc [radius=\n3,start angle=144,end angle=180]
    (-bottom-right-ear-1) arc [radius=\n3,start angle=36,end angle=0];
  },
  clause boolean-1-2-3/.pic={
    \pic {clause vertices};
    \path [vertical] (-top-left-ear-1)      -- (-top-left-ear-2);
    \path [vertical] (-top-right-ear-1)     -- (-top-right-ear-2);
    \path [vertical] (-bottom-left-ear-1)   -- (-bottom-left-ear-2);
    \path [vertical] (-bottom-right-ear-1)  -- (-bottom-right-ear-2);
    \path [vertical] (-bottom-left-cross)   -- (-in-1-1);
    \path [vertical] (-bottom-right-cross)  -- (-bottom-left-foot-1);
    \path [vertical] (-bottom-left-foot-1)  -- (-bottom-left-foot-2);
    \path [vertical] (-bottom-left-foot-1)  -- (-in-1-2);
    \path [vertical] (-bottom-right-cross)  -- (-in-2-1);
    \path [vertical] (-bottom-right-foot-1) -- (-in-2-2);
    \path [vertical] (-bottom-left-cross)   -- (-bottom-right-foot-1);
    \path [vertical] (-bottom-right-foot-1) -- (-bottom-right-foot-2);
    \path [vertical] (-top-left-cross)      -- (-inter-1);
    \path [vertical] (-top-left-cross)      -- (-top-right-foot-1);
    \path [vertical] (-top-right-cross)     -- (-in-3-1);
    \path [vertical] (-top-right-foot-1)    -- (-in-3-2);
    \path [vertical] (-top-right-foot-1)    -- (-top-right-foot-2);
    \path [vertical] (-top-left-foot-1)     -- (-top-left-foot-2);
    \path [vertical] (-top-left-foot-1)     -- (-inter-2);
    \path [vertical] (-inter-1)             -- (-inter-3);
    \path [vertical] (-inter-2)             -- (-inter-4);
    \path [vertical]
    let \n1={0.75},
        \n2={\n1/2},
        \n3={\n2/sin(18)} in
    (-inter-3) arc [radius=\n3,start angle=108,end angle=140];
    \path [vertical]
    let \n1={0.75},
        \n2={\n1/2},
        \n3={\n2/sin(18)} in
    (-inter-4) arc [radius=\n3,start angle=72,end angle=40];
    \path [vertical]
    let \n1={0.75},
        \n2={\n1/2},
        \n3={\n2/sin(18)} in
    (-out-1) arc [radius=\n3,start angle=108,end angle=140];
    \path [vertical]
    let \n1={0.75},
        \n2={\n1/2},
        \n3={\n2/sin(18)} in
    (-out-2) arc [radius=\n3,start angle=72,end angle=40];
    \path [vertical]
    let \n1={0.75},
        \n2={\n1/2},
        \n3={\n2/sin(18)} in
    (-top-left-ear-1) arc [radius=\n3,start angle=144,end angle=176];
    \path [vertical]
    let \n1={0.75},
        \n2={\n1/2},
        \n3={\n2/sin(18)} in
    (-bottom-right-ear-1) arc [radius=\n3,start angle=36,end angle=4];
    \path [vertical]
    let \n1={0.75},
        \n2={\n1/2},
        \n3={\n2/sin(18)} in
    (-bottom-left-ear-1) arc [radius=\n3,start angle=144,end angle=176];
    \path [vertical]
    let \n1={0.75},
        \n2={\n1/2},
        \n3={\n2/sin(18)} in
    (-top-right-ear-1) arc [radius=\n3,start angle=36,end angle=4];
    \path [name path=bend-ref] (-top-right-cross) -- (-bend-ref);
    \path [name path=inter-up] (-top-left-foot-1) -- +(90:2);
    \path [name intersections={of=bend-ref and inter-up},vertical,-{>[bend]}]
    let \p1 = (intersection-1),
        \p2 = (-top-left-foot-1),
        \p3 = (-bend-ref),
        \p4 = (-top-right-cross),
        \n1 = {\y4 - \y3},
        \n2 = {\x4 - \x3},
        \n3 = {atan2(\n1,\n2)},
        \n4 = {45-\n3/2},
        \n5 = {\y1 - \y2},
        \n6 = {\n5/tan(\n4)} in
    (-top-left-foot-1) ++(0:\n6) ++({\n3+90}:\n6) coordinate (-inter-2-bend)
    (-top-right-cross) -- (-inter-2-bend)
    arc [radius=\n6,start angle={\n3+90},end angle=178.5];
    \path [horizontal]
    let \n1={0.75},
        \n2={\n1/2},
        \n3={\n2/sin(18)} in
    (-in-1-1)  --  (-in-1-2)
    (-in-2-1)  --  (-in-2-2)
    (-in-3-1)  --  (-in-3-2)
    (-inter-1) --  (-inter-2)
    (-inter-3) arc [radius=\n3,start angle=108,end angle=72]
    (-out-1)   arc [radius=\n3,start angle=108,end angle=72];
  }
}

\newcommand*{\C}{\mathcal{C}}

\newcommand*{\NN}{\mathbb{N}}
\newcommand*{\dvec}[1]{\vec{\vec{#1}}}

\newcommand*{\dprime}{^{\prime\mkern-2mu\prime}}

\newcommand*{\true}{\textsc{True}}
\newcommand*{\false}{\textsc{False}}

\newcommand*{\pre}[2]{#1[\mathnormal{:}#2]} 
\newcommand*{\suf}[2]{#1[#2\mathnormal{:}]} 
\newcommand*{\e}[1]{\{#1\}}                 
\renewcommand*{\a}[1]{(#1)}                 
\newcommand*{\p}[1]{\langle #1 \rangle}     
\newcommand*{\s}[1]{[#1]}                   
\newcommand*{\red}[2]{#1 - #2}              

\newif\ifnotes
\notestrue
\ifnotes
  \usepackage{todonotes}
  \newcommand{\leo}[1]{{\color{black} #1}}
  \newcommand{\nz}[1]{{\color{black} #1}}
  \newcommand{\nzdel}[1]{{\color{purple} 
  \sout{#1}}}
  \newcommand{\yuki}[1]{{\color{black} #1}}

\else
  \newcommand\todo[2][]{}
  \newcommand{\leo}[1]{#1}
  \newcommand{\nz}[1]{#1}
  \newcommand{\nzdel}[1]{}
  \newcommand{\yuki}[1]{#1}
\fi

\begin{document}


\title{A Wild Sheep Chase Through an Orchard
}
\author[1]{Jordan Dempsey\footnote{Research of Jordan Dempsey and Norbert Zeh was supported by the Natural Sciences and Engineering Research Council of Canada under grant RGPIN/05435-2018.}}
\author[2]{Leo van Iersel\footnote{Research of Leo van Iersel and Mark Jones was partially funded by Netherlands Organization for Scientific Research (NWO) under projects OCENW.M.21.306 and OCENW.KLEIN.125.}}
\author[2]{Mark Jones}
\author[2]{Yukihiro Murakami\footnote{Corresponding author, email: y.murakami@tudelft.nl}}
\author[1]{Norbert Zeh}
\affil[1]{Faculty of Computer Science, Dalhousie University, Halifax, Canada}
\affil[2]{Delft Institute of Applied Mathematics, Delft University of Technology, The Netherlands}
\date\today
\maketitle


\begin{abstract}\noindent
  Orchards are a biologically relevant class of phylogenetic networks as they
  can 
  describe treelike evolutionary histories augmented with horizontal transfer
  events. Moreover, the class has attractive mathematical characterizations that
  can be exploited algorithmically. On the other hand, undirected orchard
  networks have hardly been studied yet. Here, we prove that deciding whether
  an \leo{undirected, binary} phylogenetic network is an orchard---or
  equivalently, whether it has an orientation that makes it a rooted
  orchard---is NP-hard.  \leo{For this, we introduce a new characterization of
  undirected orchards which could be useful for proving positive results.}
\end{abstract}

\section{Introduction}\label{sec:intro}

\emph{Phylogenetic networks} are graphs describing evolutionary
relationships~\cite{huson2010phylogenetic,bapteste2013networks}. Here we focus
on explicit phylogenetic networks, which describe a hypothesis of the
evolutionary history of a given set~$X$ of taxa (taxonomic units). Such networks
can be either directed or undirected. Directed phylogenetic networks are
directed acyclic graphs with a unique source vertex (the \emph{root}) and with
sink vertices (the \emph{leaves}) corresponding to the elements of~$X$. Vertices
with more than one incoming arc are used to model reticulate evolutionary events
such as hybridization or lateral gene transfer (LGT). Such vertices are called
\emph{reticulations}. A phylogenetic network without reticulation vertices is a
\emph{phylogenetic tree}. While phylogenetic trees are too limited to model all
evolutionary histories, the space of all phylogenetic networks is much too
large, for computational reasons as well as practical reasons such as
interpretability and identifiability. Hence, almost all research on phylogenetic
networks considers restricted classes of networks.

The class of directed \emph{orchard} networks \nz{(\emph{orchards} for short)}
was introduced in~\cite{erdHos2019class} as a class of directed phylogenetic
networks that can be uniquely reconstructed from so-called ancestral profiles.
Simultaneously,~\cite{janssen2021cherry} defined exactly the same class under
the name \emph{cherry picking networks}.  Precise definitions will be given in
Section~\ref{sec:prelim}, but the rough idea is that a directed network is
\nz{an} \emph{orchard} if it can be reduced to a graph with a single leaf by a
sequence of operations reducing cherries (two leaves with a common parent) and
reticulated cherries (cherries subdivided by a reticulation), see
Figure~\ref{fig:intro}. The motivation of both papers introducing directed
\nz{orchards} was mostly mathematical. \nz{Orchards} are a natural
generalisation of the intensively studied class of tree-child
networks~\cite{cardona2009comparison}. Many proof techniques that can be used
for tree-child networks can also be applied to \nz{orchards} (see
e.g.~\cite{semple2021trinets,landry2022defining,bernardini2023constructing,huber2024orienting}).
Indeed, the idea behind directed \nz{orchards} is based on \emph{cherry picking
sequences}, which were introduced in~\cite{linz2019attaching} to characterize
the tree-child hybridization number of a set of trees.  Later, a static
characterization in terms of so-called ``acyclic cherry covers" was
found~\cite{van2020unifying}, establishing a clear link to tree-based networks,
which are basically networks that can \nz{be} obtained from a rooted tree by
inserting linking arcs~\cite{francis2015phylogenetic,jetten2018nonbinary}. Not
much later, a second characterization of directed \nz{orchards} was presented,
which further strengthened the link to tree-based networks and furthermore
provided a clear biological motivation for \nz{orchards}~\cite{van2022orchard}.
This characterization says that a directed network is \nz{an} orchard precisely
if it can be obtained from (a drawing of) a rooted tree by adding
\emph{horizontal} linking arcs. This means that \nz{orchards} exactly correspond
to evolutionary histories consisting of a species tree with lateral gene
transfer (LGT) events between branches of the tree (i.e., no LGT from the
dead~\cite{szollHosi2013lateral}).

Undirected phylogenetic networks are especially relevant when the available data
is not sufficient to identify the location of the root and/or to identify which
vertices correspond to reticulation events~\cite{morrison2005networks}. The
relationship between undirected (unrooted) and directed (rooted) phylogenetic
networks has recently started to gain
attention~\cite{garvardt2023finding,huber2024orienting,docker2024existence,bulteau2023turning,urata2024orientability}.

Undirected \nz{orchards} were introduced
in~\cite[Section~8.3.3.]{murakami2021phylogenetic} but not yet studied in
detail. Roughly speaking, an undirected network is \nz{an} \emph{orchard} if it
can be reduced to a network with a single vertex by a sequence of operations
reducing cherries (two leaves with a common neighbour) and 2-chains (two leaves
with their neighbours connected by an edge that is not a cut edge), see
Figure~\ref{fig:intro}. This PhD thesis also gave an example showing that
arbitrarily reducing cherries and 2-chains does not work to decide whether an
undirected network is \nz{an} orchard. \nz{The thesis further} conjectured that
an undirected network is \nz{an} orchard if and only if it can be oriented as a
directed network that is \nz{an} orchard \nz{and provided a proof sketch} for
this equivalence. \nz{We provide a formal proof of this claim in
\cref{sec:prelim}. One of the questions posed was whether the following problem
can be solved in polynomial time:}

\begin{problem}{Orchard Recognition}
  Decide whether an unrooted network is an orchard.
\end{problem}

For directed networks, this is possible because in that case, it does work to
repeatedly reduce an arbitrary cherry or reticulated
cherry~\cite{erdHos2019class,janssen2021cherry}. As mentioned above, a similar
approach does not work for the undirected case.

The problem of orienting an undirected network to a directed orchard network has
been considered in~\cite{huber2024orienting}. \yuki{By orienting, we mean to subdivide a single edge, add a vertex~$\rho$, add an arc from $\rho$ to the subdivision vertex, and then orient all edges}:

\begin{problem}{Orchard Orientation}
  Decide whether an unrooted network has an orientation that is an orchard.
\end{problem}

In~\cite{huber2024orienting}, it was shown that this problem is fixed-parameter tractable \nz{when
parameterized by} the \emph{level} of the network, which is a measure for
\nz{the tree-likeness of a network} and can be defined as the maximum number of
edges that need to be deleted per biconnected component to turn the network into
a tree.
However, 
the computational complexity of the problem was left open. The paper did
show NP-hardness for the related problem of deciding whether an undirected network can
be oriented as a directed tree-based network. However, despite the close link
between \nz{orchards} and tree-based networks, there does not seem to be an easy
reduction between the two problems, or an easy way to adapt their reduction to
\nz{orchards}.

In this paper, we prove that both \textsc{Orchard Recognition} and
\textsc{Orchard Orientation} are NP-hard, solving the open problem
from~\cite[Section~8.3.3.]{murakami2021phylogenetic}.  \nz{This result is
somewhat surprising since \textsc{Orchard Recognition} \emph{can} be solved in
polynomial time for directed networks.}  Our proof focuses on showing that
\textsc{Orchard Orientation} is NP-hard.  By Murakami's conjecture, which we
prove in \cref{sec:prelim}, this immediately implies that that \textsc{Orchard
Recognition} is also NP-hard.  Our NP-hardness proof for \textsc{Orchard
Orientation} utilizes a subnetwork of the example
from~\cite[Section~8.3.3.]{murakami2021phylogenetic} which somewhat resembles a
sheep. We use its property that a network containing this sheep can only be
oriented as an orchard 
if the root is inside the sheep.  Hence,
a sequence of cherry reductions and 2-chain reductions must end in the sheep;
\nz{testing whether a network is an orchard becomes a chase through the network
that ends at the sheep if the network is an orchard, or fails to catch the sheep
if it is not.  Hence, the title of our paper.  To prove the NP-hardness of
deciding whether an unrooted network is an orchard, we provide a reduction from
\textsc{3-SAT} such that a given formula in 3-CNF is satisfiable if and only if the
corresponding network is an orchard.}

\section{Preliminaries}

\label{sec:prelim}


\paragraph{Sets of integers and sequences.}

\nz{Throughout this paper, we use $[n]$ to denote the set $\{1, \ldots, n\}$,
where $n$ is a non-negative integer.  If $n = 0$, then $[n] = \emptyset$.  For a
sequence $S$ of length $n$, we use $S_i$ to denote the $i$th entry in $S$, that
is $S = \s{S_1, \ldots, S_n}$.  We use $\pre{S}{i}$ to refer to the subsequence
containing the first $i$ elements in $S$, $\pre{S}{i} = \s{S_1, \ldots, S_i}$.
We use $\suf{S}{i}$ to denote the subsequence containing all but the first $i -
1$ elements in $S$, $\suf{S}{i} = \s{S_i, \ldots, S_n}$.  $\pre{S}{0}$ denotes
the empty sequence $\s{}$, as does $\suf{S}{i}$, for $i > n$.}

\paragraph{Unrooted phylogenetic networks and trees.}

\leo{A \emph{leaf} of \nz{an undirected} graph is a vertex of degree at
most~$1$.} 
An \emph{unrooted binary phylogenetic network} on~$X$ is a simple,
\leo{connected}, \nz{undirected} graph in which leaves are labelled bijectively
with the elements of~$X$ and all other vertices are of degree-$3$ \nz{and
unlabelled.}  We refer to these \nz{simply} as \emph{networks}, and sometimes as
\emph{undirected networks}. An \emph{\nz{unrooted} tree} is an \nz{unrooted}
network without cycles.

\paragraph{Rooted phylogenetic networks.}

A \emph{directed binary phylogenetic network} is a directed acyclic graph with a
single source (the \emph{root}) of out-degree~1 and whose non-root vertices have
in-degree~1 and out-degree~0 (leaves), in-degree~1 and out-degree~2 (tree
\nz{vertices}) or in-degree~2 and out-degree~1 (reticulations). \nz{We refer to
such networks simply as \emph{directed networks}.  Throughout this paper, we
reserve the term \emph{edge} to refer to an edge in an undirected network, and
the term \emph{arc} to refer to an edge in a directed network.  We use the
notation $\e{u, v}$ to refer to an edge with endpoints $u$ and $v$, and $\a{u,
v}$ to refer to an arc from $u$ to $v$.}

\paragraph{Suppressing vertices and subdividing edges.}

\nz{If $v$ is a degree-$2$ vertex in an undirected graph $G$, with neighbours
$u$ and $w$, then to \emph{suppress} $v$ is to remove $v$ and its incident edges
from $G$ and to reconnect $u$ and $w$ with a new edge $\e{u, w}$ unless this
edge already exists in $G$.  The inverse operation is to \emph{subdivide} an
edge $\e{u, w}$ with a new vertex $v$, which removes the edge $\e{u, w}$ from
$G$ and adds the vertex $v$ and the two edges $\e{u, v}$ and $\e{v, w}$ to $G$.

These operations can also be defined for directed graphs, where subdividing an
arc $\a{u, w}$ with a new vertex $v$ removes $\a{u, w}$ and adds the vertex $v$
and two arcs $\a{u, v}$ and $\a{v, w}$ to $G$.  A vertex $v$ can be suppressed
in a directed graph if it has in-degree $1$ and out-degree $1$.  In this case,
suppressing $v$ removes $v$, its in-arc $\a{u, v}$, and its out-arc $\a{v, w}$
from $G$, and adds a new arc $\a{u, w}$ to $G$, again provided that this arc
does not already exist in $G$.}

\paragraph{Reducing unrooted networks.}

Let~$x$ and~$y$ be two leaves in an \nz{unrooted} network~$N$.  \nz{The ordered
pair $\p{x, y}$ is a \emph{cherry}} if~$x$ and~$y$ share a common neighbour.
\leo{A \emph{cut edge} is an edge whose removal disconnects the network.}
\nz{The pair $\p{x, y}$ is} a \emph{$2$-chain} if the neighbours of~$x$ and~$y$
are connected by an edge that is not a cut edge of~$N$. \nz{A~\emph{reducible
pair} in a network is a pair of leaves 
$\p{x, y}$ that is either  \yuki{an edge $\e{x,y}$}, a cherry, or a $2$-chain.}
\nz{For such a pair}, \emph{reducing $N$ by $\p{x, y}$} is the action of
\begin{itemize}[noitemsep]
  \item Deleting the leaf $x$ from $N$ if \yuki{$\e{x,y}$ is an edge or if} $\p{x, y}$ is a cherry of $N$ or
  \item Deleting the edge between the neighbours of $x$ and $y$ if $\p{x, y}$ is
  a $2$-chain of $N$
\end{itemize}
and suppressing any degree-$2$ vertices to ensure the result is once again a
network.
\yuki{Note that an edge $\e{x,y}$ is a reducible pair of a network if and only if the network contains exactly one edge between the two leaves~$x$ and~$y$.
This reduction can take place only when the network contains exactly one edge.}

We denote the network produced by reducing $N$ by $\p{x, y}$ as $\red{N}{\p{x,
y}}$.  \leo{If~$\p{x,y}$ is not a reducible pair of~$N$ then $\red{N}{\p{x,y}} =
N$.} \nz{Note that $\red{N}{\p{x,y}} \ne \red{N}{\p{y,x}}$ if $\p{x, y}$ is a
cherry, but $\red{N}{\p{x,y}} = \red{N}{\p{y,x}}$ for a 2-chain $\p{x, y}$.}

\nz{Reducing a network $N$ by a \emph{sequence} of ordered pairs $S$ produces
the network $\red{N}{S}$ defined as $\red{N}{S} = N$ if $S$ is empty, and
as $\red{N}{S} = \red{(\red{N}{S_1})}{\suf{S}{2}}$ if $S$ is non-empty.}

\paragraph{Reducing rooted networks.}

A \emph{cherry} in a directed network is an ordered pair of leaves $\p{x,y}$
with a common parent.  A \emph{reticulated cherry} is a pair of leaves $\p{x,y}$
such that the parent $p$ of $x$ is a reticulation and~$p$ and~$y$ have a common
parent.  \nz{A \emph{reducible pair} in a rooted network is a pair $\p{x, y}$
that is either a cherry or a reticulated cherry.  \nz{For such a pair,}
\emph{reducing $N$ by $\p{x, y}$} is the action of
\begin{itemize}[noitemsep]
  \item Deleting the leaf $x$ from $N$ if $\p{x, y}$ is a cherry of $N$ or
  \item Deleting the arc from $y$'s parent to $x$'s parent if $\p{x, y}$ is a
reticulated cherry of $N$
\end{itemize}
and suppressing any degree-$2$ vertices to ensure the result is once again a
rooted network.

As for unrooted networks, $\red{N}{\p{x, y}}$ denotes the network produced by
reducing $N$ by $\p{x, y}$, and $\red{N}{\p{x, y}} = N$ if $\p{x, y}$ is not a
reducible pair of $N$.  The network $\red{N}{S}$ obtained by reducing $N$ by a
sequence of ordered pairs $S$ is also defined as in the unrooted case.}

\paragraph{Orchards.}

\begin{figure}
  \hspace{\stretch{1}}%
  \begin{tikzpicture}
    \path     node [vertex] (1) {}
    +(90:1)   node [vertex] (r) {}
    ++(240:1) node [vertex] (2) {}
    ++(240:2) node [vertex] (3) {}
    +(0:2)    node [vertex] (4) {}
    ++(300:1) node [vertex] (5) {}
    +(270:1)  node [vertex] (b) {}
    ++(0:1)   node [vertex] (6) {}
    +(270:1)  node [vertex] (c) {}
    ++(0:1)   node [vertex] (7) {}
    +(270:1)  node [vertex] (d) {}
    +(0:1)    node [vertex] (8) {};
    \begin{scope}[overlay]
      \path [name path=down] (3) -- +(240:3);
      \path [name path=left] (b) -- +(180:2);
      \path [name intersections={of=down and left}] (intersection-1) coordinate (a);
      \path [name path=down] (8) -- +(300:2);
      \path [name path=right] (d) -- +(0:2);
      \path [name intersections={of=down and right}] (intersection-1) coordinate (e);
    \end{scope}
    \path
    (a) node [vertex] (a) {}
    (e) node [vertex] (e) {};
    \path [vertical] (r) -- (1);
    \path [vertical] (1) -- (2);
    \path [vertical] (1) -- (8);
    \path [vertical] (2) -- (3);
    \path [vertical] (2) -- (4);
    \path [vertical] (3) -- (a);
    \path [vertical] (3) -- (5);
    \path [vertical] (4) -- (6);
    \path [vertical] (4) -- (7);
    \path [vertical] (5) -- (b);
    \path [vertical] (5) -- (6);
    \path [vertical] (6) -- (c);
    \path [vertical] (7) -- (d);
    \path [vertical] (7) -- (8);
    \path [vertical] (8) -- (e);
    \path
    (a) node [anchor=north,yshift=-1pt] {\vphantom{$bd$}$a$}
    (b) node [anchor=north,yshift=-1pt] {\vphantom{$bd$}$b$}
    (c) node [anchor=north,yshift=-1pt] {\vphantom{$bd$}$c$}
    (d) node [anchor=north,yshift=-1pt] {\vphantom{$bd$}$d$}
    (e) node [anchor=north,yshift=-1pt] {\vphantom{$bd$}$e$};
    \path (current bounding box.south) node [anchor=north] {$N_d$};
  \end{tikzpicture}%
  \hspace{\stretch{2}}%
  \begin{tikzpicture}
    \path       node [vertex] (b) {}
    ++(0:1)     node [vertex] (1) {}
    ++(54:1)    node [vertex] (2) {}
    +(108:1)    node [vertex] (a) {}
    ++(342:1)   node [vertex] (3) {}
    +(0:1)      node [vertex] (6) {}
    ++(270:1)   node [vertex] (4) {}
    ++(198:1)   node [vertex] (5) {}
    +(252:1)    node [vertex] (c) {}
    (6) +(45:1) node [vertex] (e) {}
    ++(270:1)   node [vertex] (7) {}
    +(315:1)    node [vertex] (d) {};
    \path [edge]
    (1) -- (2) -- (3) -- (4) -- (5) -- (1) -- (b)
    (2) -- (a)
    (3) -- (6) -- (7) -- (4)
    (5) -- (c)
    (6) -- (e)
    (7) -- (d);
    \path
    (a) node [anchor=south,,xshift=-2.5pt,yshift=1pt] {$a$}
    (b) node [anchor=east,xshift=-1pt]                {$b$}
    (c) node [anchor=north,xshift=-2.5pt,yshift=-1pt] {$c$}
    (d) node [anchor=north west]                      {$d$}
    (e) node [anchor=south west]                      {$e$};
    \path (current bounding box.south) node [anchor=north,yshift=-18pt] {$N_u$};
  \end{tikzpicture}%
  \hspace{\stretch{2}}%
  \begin{tikzpicture}
    \path     node [vertex] (1) {}
    +(90:1)   node [vertex] (r) {}
    ++(240:3) ++(300:1) +(270:1) coordinate (b);
    \begin{scope}[overlay]
      \path [name path=down] (1) -- +(240:6);
      \path [name path=left] (b) -- +(180:2);
      \path [name intersections={of=down and left}] (intersection-1) coordinate (a);
      \path [name path=down] (8) -- +(300:6);
      \path [name path=right] (a) -- +(0:6);
      \path [name intersections={of=down and right}] (intersection-1) coordinate (e);
    \end{scope}
    \path
    (a)                              node [vertex] (a) {}
    (e)                              node [vertex] (e) {}
    (barycentric cs:a=0.75,e=0.25)   node [vertex] (b) {}
    (barycentric cs:a=0.5,e=0.5)     node [vertex] (c) {}
    (barycentric cs:a=0.25,e=0.75)   node [vertex] (d) {}
    (barycentric cs:1=0.75,a=0.25)   node [vertex] (2) {}
    (barycentric cs:1=0.25,a=0.75)   node [vertex] (3) {}
    (barycentric cs:2=0.333,d=0.667) node [vertex] (4) {};
    \path [vertical] (r) -- (1);
    \path [vertical] (1) -- (e);
    \path [vertical] (1) -- (2);
    \path [vertical] (2) -- (3);
    \path [vertical] (2) -- (4);
    \path [vertical] (3) -- (a);
    \path [vertical] (3) -- (b);
    \path [vertical] (4) -- (c);
    \path [vertical] (4) -- (d);
    \path
    (a) node [anchor=north,yshift=-1pt] {\vphantom{$bd$}$a$}
    (b) node [anchor=north,yshift=-1pt] {\vphantom{$bd$}$b$}
    (c) node [anchor=north,yshift=-1pt] {\vphantom{$bd$}$c$}
    (d) node [anchor=north,yshift=-1pt] {\vphantom{$bd$}$d$}
    (e) node [anchor=north,yshift=-1pt] {\vphantom{$bd$}$e$};
    \path (current bounding box.south) node [anchor=north] {\vphantom{$N_d$}$T$};
  \end{tikzpicture}%
  \hspace{\stretch{1}}
  \caption{A directed orchard~$N_d$, an undirected
  orchard~$N_u$, and a rooted tree~$T$. $N_d$ is \nz{an} orchard since it can be obtained from~$T$ by
  inserting horizontal arcs. It then follows that $N_u$ is \nz{also an} orchard since~$N_d$ can be obtained from~$N_u$
  by inserting a root and orienting all edges.  Alternatively, it can be seen
  that~$N_d$ is \nz{an} orchard since we can, for example, reduce reticulated
  cherries on~$\p{c, b}$ and~$\p{e, d}$ \nz{to} obtain \nz{the} tree~$T$, which
  can \nz{then} be fully reduced by reducing cherries.  Similarly,~$N_u$ is
  \nz{an} orchard since we can reduce 2-chains on~$\p{b,c}$ and~$\p{d,e}$ and
  subsequently reduce cherries in the resulting tree.}
  \label{fig:intro}
\end{figure}
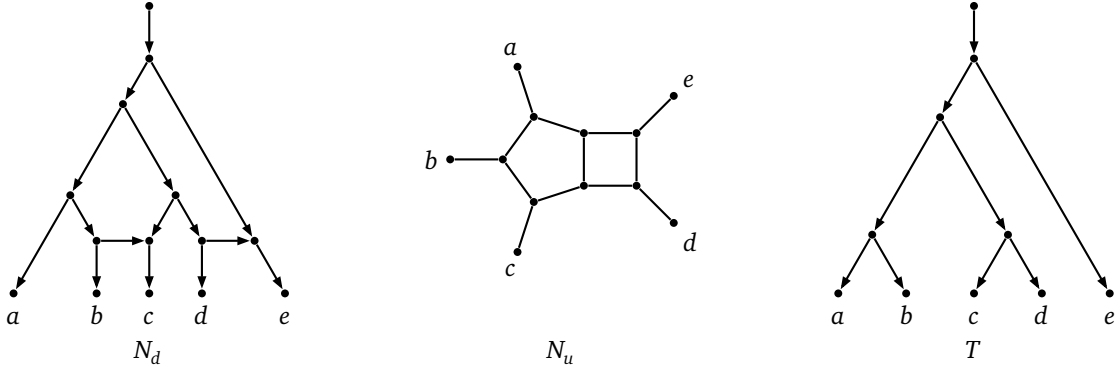

\nz{A sequence of ordered pairs $S$ is said to \emph{reduce} an unrooted network
$N$ if $\red{N}{S}$ has a single \leo{vertex}.  If $N$ is a rooted network, then
$S$ reduces $N$ if $\red{N}{S}$ has a single arc, which connects its root to its
only leaf.} 	\leo{A \nz{rooted or unrooted} network~$N$ is \nz{an}
\emph{orchard} if there exists a sequence of ordered pairs~$S$ that reduces it.}
\nz{A sequence $S$ that reduces an orchard $N$ is said to be \emph{minimal} if
$\red{N}{\pre{S}{i}} \ne \red{N}{\pre{S}{i-1}}$, for all $i \in [|S|]$.}

For example, the network~$N_u$ \nz{in} \cref{fig:intro} is \nz{an} orchard
since it can be reduced by the sequence $\s{\p{b, c}, \p{d, e}, \p{a, b},\break
\p{b, e}, \p{d, e}, \p{c, e}}$ or, alternatively, by the sequence 
$\s{\p{e, d},\p{e, d}, \p{a, e}, \p{b, e}, \p{c, d}, \p{e, d}}$.
\nz{Neither of these two
sequences reduces the rooted network $N_d$ in \cref{fig:intro}.  A sequence that
reduces both $N_d$ and $N_u$ is $\s{\p{c, b}, \p{e, d}, \p{a, b}, \break \p{d,
c}, \p{c, b}, \p{b, e}}$.}

\paragraph{Orientations of undirected networks.}

Given an undirected network~$N_u$ on~$X$, a directed network~$N_d$ on~$X$
is an \emph{orientation} of~$N_u$ if it can be obtained from~$N_u$ as follows
\begin{itemize}[noitemsep]
  \item If~$N_u$ has at least one edge, subdivide an edge of~$N_u$ with a new
  vertex~$v$, add another new vertex~$\rho$ (the root), add an arc~$\a{\rho,
  v}$, and replace each edge with an arc with the same endpoints.
  \item Otherwise, adding a new vertex~$\rho$ (the root) and add an
  arc~$\a{\rho, v}$, with~$v$ the only vertex of~$N_u$.
\end{itemize}
An \emph{orchard orientation} of an undirected network~$N_u$ is an
orientation~$N_d$ that is \nz{a (rooted)} orchard.

\paragraph{Equivalence between orchard recognition and orchard orientation.}

\nz{The following theorem formalizes a proof sketch by Murakami
\cite{murakami2021phylogenetic}.}

\begin{thm}
  \label{thm:orchard-orientation}
  An undirected network is \nz{an} orchard if and only if it has an orchard
  orientation.
\end{thm}

\begin{proof}
  \nz{First assume that $N_u$ is an orchard and that $S$ is a minimal sequence
  that reduces it.  We use induction on $|S|$ to prove that $N_u$ has an orchard
  orientation $N_d$, which is also reduced by $S$ and thus is an orchard
  orientation of $N_u$.  If $|S| = 0$, then $N_u$ has a single leaf, so $S$ also
  reduces its unique orientation $N_d$.  If $|S| = 1$, then $S_1 = \p{x, y}$,
  where $x$ and $y$ are the only two leaves of $N_u$, which are connected by an
  edge.  In the unique orientation $N_d$ of $N_u$, $x$ and $y$ are the only two
  leaves and form a cherry.  Thus, $S$ reduces $N_d$.  Finally, if $|S| > 1$,
  then $S_1 = \p{x, y}$ is a cherry or $2$-path in $N_u$.  By the induction
  hypothesis, $\red{N_u}{S_1}$ has an orientation $N_d'$ that is reduced by
  $\suf{S}{2}$.  If $\p{x, y}$ is a cherry of $N_u$, then we obtain an
  orientation $N_d$ of $N_u$ by subdividing the parent arc of $y$ in $N_d'$ with
  a new vertex $p$ and adding $x$ and the arc $\a{p, x}$ to $N_d'$.  Then $N_d'
  = \red{N_d}{S_1}$.  Thus, since $\suf{S}{2}$ reduces $N_d'$, $S$ reduces
  $N_d$.  If $\p{x, y}$ is a $2$-path of $N_u$, then we obtain an orientation
  $N_d$ of $N_u$ by subdividing the parent arcs of $x$ and $y$ in $N_d'$ with
  two new vertices $p$ and $q$, respectively, and adding the arc $\a{q, p}$ to
  $N_d'$.  This makes $S_1 = \p{x, y}$ a reticulated cherry in~$N_d$, so
  $\red{N_d}{S_1} = N_d'$.  Therefore, since $\suf{S}{2}$ reduces $N_d'$, $S$
  reduces $N_d$.

  Now assume that $N_u$ has an orchard orientation $N_d$, and that $S$ is a
  minimal sequence that reduces $N_d$.  We use induction on $|S|$ to prove that
  $S$ also reduces $N_u$, so $N_u$ is an orchard.  If $|S| = 0$, then $N_d$ has
  a single leaf, so $N_u$ also has a single leaf, and $S$ reduces it.  If $|S| =
  1$, then $N_d$ has two leaves $x$ and $y$, which form a cherry $S_1 = \p{x,
  y}$.  In this case, $N_u$ also has $x$ and $y$ as its only leaves, which are
  connected by an edge $\e{x, y}$.  Thus, once again, $S$ reduces $N_u$.  So
  assume that $|S| > 1$.  Then $S_1 = \p{x, y}$ is a cherry or reticulated
  cherry of $N_d$.  Since $S_1$ does not reduce $N_d$, $N_d$ has more than $3$
  arcs.  Thus, since $N_d$ is an orientation of $N_u$, $N_u$ has more than one
  edge.  Therefore, if $\p{x, y}$ is a cherry of $N_d$, it is also a cherry of
  $N_u$.  If $\p{x, y}$ is a reticulated cherry of $N_d$, then it is a $2$-path
  of $N_u$.  In either case, $\red{N_d}{S_1}$ is an orientation of
  $\red{N_u}{S_1}$.  Since $\suf{S}{2}$ reduces $\red{N_d}{S_1}$, it also
  reduces $\red{N_u}{S_1}$, by the induction hypothesis.  Thus, $S$ reduces
  $N_u$.}
\end{proof}

\section{\nz{Partial Orientations}}

\nz{To prove the hardness of \textsc{Orchard Orientation}, it is helpful to have
an easily checkable characterization of rooted orchards.  HGT-consistent
labellings \cite{van2022orchard} provide such a characterization.  Building on
HGT-consistent labellings, we define \emph{HGT-consistent partial orientations}
(HPOs) of unrooted networks as an essentially equivalent concept in this
section.  We prove that an unrooted network has an HPO if and only if it has an
orientation with an HGT-consistent labelling (i.e., an orchard orientation).  We
find HPOs more convenient to work with, for two reasons:  First, an HPO has the
same vertex set as the given undirected network $N$, and there exists a
bijection between the edges of the undirected network and the corresponding
edges or arcs of the partial orientation.  In contrast, the edge of $N$
subdivided to produce an orientation is not an arc of the orientation, and the
orientation has a few additional vertices and arcs.  Second, and more
importantly, an HGT-consistent labelling is an assignment of integer labels to
the vertices of an orientation that needs to satisfy a number of properties in
the neighbourhood of every vertex.  It is easy to define such a labelling of the
vertices in any subgraph of an orientation, but it takes a non-trivial amount of
care to ensure that the labellings defined for different subgraphs all ``fit
together'' to produce an HGT-consistent labelling (or to argue that this cannot
be done for a given network because it is not an orchard).  In contrast, HPOs are
defined only in terms of conditions on the in-degrees and out-degrees of
vertices, and it is rather easy to verify whether partial orientations of
subgraphs that satisfy these conditions fit together to produce a partial
orientation of the whole network that satisfies the same conditions and, thus,
is an HPO.}

An \emph{HGT-consistent labelling} \cite{van2022orchard} of a \nz{rooted binary}
network $N = (V, E)$ is a labelling $\tau : V \rightarrow \NN$ such that
\begin{enumerate}[label=(R\arabic{*}),widest=4,leftmargin=*,noitemsep]
  \item $\tau(u) \le \tau(v)$ for every \nz{arc} $\a{u, v} \in
  E$,\label[prop]{prop:respect-arcs}
  \item Every internal \nz{vertex} $u$ has \leo{at least one} child $v$ with $\tau(u) <
  \tau(v)$,\label[prop]{prop:vertical-out-arc}
  \item If $v$ is a leaf or tree \nz{vertex} with parent $u$, then $\tau(u) < \tau(v)$,
  and\label[prop]{prop:vertical-in-arc}
  \item If $w$ is a reticulation with parents $u$ and $v$, then exactly one of
  $\tau(u) = \tau(w)$ and $\tau(v) = \tau(w)$
  holds.\label[prop]{prop:horizontal-reticulation-arc}
\end{enumerate}
The labelling $\tau$ can be seen as organizing the \nz{vertices} of $N$ into
levels so that all \nz{arcs} are either directed downwards (if numbering levels
from top to bottom) or horizontally.  Accordingly, we call an \nz{arc} $\a{u,
v}$ with $\tau(u) < \tau(v)$ a \emph{vertical \nz{arc}}, and an \nz{arc} $\a{u,
v}$ with $\tau(u) = \tau(v)$ a \emph{horizontal \nz{arc}}.

\begin{thm}[Van Iersel et al.\ {\cite[Theorem~1]{van2022orchard}}]
  \label{thm:drawing}
  A \nz{rooted} network is \nz{an} orchard if and only if it admits an
  HGT-consistent labelling.
\end{thm}

\nz{%
\begin{lem}
  \label{lem:vertical-path}
  If $N$ is a rooted network with root $\rho$ and $\tau$ is an HGT-consistent
  labelling of $N$, then every vertex $v \ne \rho$ in $N$ has a path $P$ from
  $\rho$ to $v$ in $N$ with $\tau(x) < \tau(v)$, for all $x \ne v$ in $P$.
\end{lem}

\begin{proof}
  By
  \cref{prop:respect-arcs,prop:vertical-in-arc,prop:horizontal-reticulation-arc},
  $v$ has a parent $y$ with $\tau(y) < \tau(v)$.  Since $\rho$ is the only root
  of $\vec{N}$, there exists a path $P'$ from $\rho$ to $y$ in $\vec{N}$.  By
  \cref{prop:respect-arcs}, every vertex $x \in P'$ satisfies $\tau(x) \le
  \tau(y) < \tau(v)$.  $P = P' \circ \s{v}$ is a path from $\rho$ to $v$ and
  every vertex $x \ne v$ in $P$ belongs to $P'$ and, therefore, satisfies
  $\tau(x) < \tau(v)$.
\end{proof}

\begin{lem}
  \label{lem:no-reticulations-near-root}
  Let $\vec{N}$ be an orchard orientation of an unrooted network $N$, let $\rho$
  be its root, let $u$ and $v$ be the two grandchildren of $\rho$ in $\vec{N}$,
  and let $\tau$ be an HGT-consistent labelling of $\vec{N}$.  If $\tau(u) \le
  \tau(v)$, then $u$ has no neighbour $z$ in $\vec{N}$ with $\tau(u) = \tau(z)$.
\end{lem}

\begin{proof}
  Assume the contrary.  Then by \cref{lem:vertical-path}, there exists a path
  $P$ from $\rho$ to $z$ such that every vertex $x \ne z$ in~$P$ satisfies
  $\tau(x) < \tau(z) = \tau(u) \le \tau(v)$.  Thus, $u, v \notin P'$, where $P'$
  is the subpath of $P$ from $\rho$ to $z$'s predecessor $y$ in~$P$.  \yuki{Letting~$p$ be the child of~$\rho$,} this
  implies that $y \in \{\rho, p\}$ because $u$ and $v$ are the only two
  grandchildren of~$\rho$, which implies that every path from $\rho$ to a vertex
  other than $\rho$ or $p$ must contain $u$ or $v$.  If $y = \rho$, then $z = p$
  because $p$ is the only child of $\rho$.  Thus, $\tau(p) = \tau(u)$.  By
  \cref{lem:vertical-path}, there exists a path $P\dprime$ from $\rho$ to $u$
  such that every vertex $x \ne u$ in $P\dprime$ satisfies $\tau(x) < \tau(u)$.
  Any such path $P\dprime$ must include $p$.  Since $p \ne u$, this implies that
  $\tau(p) < \tau(u)$, a contradiction.  Therefore, $y \ne \rho$, so $y = p$
  and, since $z$ is a neighbour of both $u$ and $y$, $z = v$.  This, however, is
  also impossible because the edge $\e{u, v}$ of $N$ is replaced by the two arcs
  $\a{p, u}$ and $\a{p, v}$ in $\vec{N}$, that is, $u$ and $v$ are not
  neighbours in $\vec{N}$.
\end{proof}

We call an HGT-consistent labelling of a rooted network $N$ \emph{asymmetric} if
$\tau(u) \ne \tau(v)$, for the two grandchildren $u$ and $v$ of the root
$\rho$ of $N$.

\begin{lem}
  \label{lem:asymmetric-drawing}
  If $\vec{N}$ is an orientation of an unrooted network $N$, then it has an
  asymmetric HGT-consistent labelling if and only if it has an HGT-consistent
  labelling.
\end{lem}

\begin{proof}
  The ``only if'' direction holds trivially.  To prove the ``if'' direction, let
  $\vec{N}$ be an orientation of an unrooted network $N$, let $\tau$ be an
  HGT-consistent labelling of $\vec{N}$, and let $u$ and $v$ be the two
  grandchildren of the root $\rho$ of $\vec{N}$.  If $\tau(u) \ne \tau(v)$, then
  $\tau$ is asymmetric, and the lemma holds.  So assume that $\tau(u) =
  \tau(v)$.  Then by \cref{lem:no-reticulations-near-root}, neither $u$ nor $v$
  can have a neighbour $z$ with $\tau(z) = \tau(u) = \tau(v)$.  In particular,
  neither $u$ nor $v$ is a reticulation of~$\vec{N}$, by
  \cref{prop:horizontal-reticulation-arc}.  Let $p$ be the child of $\rho$,
  which is the parent of $u$ and $v$.  We define a new labelling $\tau'$ of
  $\vec{N}$ as
  \begin{equation*}
    \tau'(x) = \begin{cases}
      \tau(x)     & \text{if } x \in \{\rho, p, u\}\\
      \tau(x) + 1 & \text{otherwise.}
    \end{cases}
  \end{equation*}
  We prove that $\tau'$ is HGT-consistent.  Since $\tau'(u) \ne \tau'(v)$, it is
  asymmetric, so the lemma holds.
  
  Every arc $\a{y, z}$ of $\vec{N}$ satisfies $y \in \{\rho, p, u\}$ or $z
  \notin \{\rho, p, u\}$ because $\rho$ has no in-arcs, the only in-arc of $p$
  is $\a{\rho, p}$, and the only in-arc of $u$ is $\a{p, u}$ (because $u$ is not
  a reticulation).  Thus, $\tau'(z) - \tau'(y) \ge \tau(z) - \tau(y)$, for every
  arc $\a{y, z}$ of $\vec{N}$.  In particular, $\tau(y) < \tau(z)$ implies that
  $\tau'(y) < \tau'(z)$, and $\tau(y) \le \tau(z)$ implies that $\tau'(y) \le
  \tau'(z)$.  Since $\tau$ satisfies
  \cref{prop:respect-arcs,prop:vertical-out-arc,prop:vertical-in-arc}, this
  shows that so does $\tau'$.
  
  If $z$ is a reticulation with parents $x$ and $y$, then by
  \cref{prop:respect-arcs,prop:horizontal-reticulation-arc}, we have w.l.o.g.\
  that $\tau(x) < \tau(z)$ and $\tau(y) = \tau(z)$.  As just argued, the former
  implies that $\tau'(x) < \tau'(z)$.  Thus, $\tau'$ satisfies
  \cref{prop:horizontal-reticulation-arc} if $\tau'(y) = \tau'(z)$.  Assume the
  contrary.  Then, since $\tau'$ satisfies \cref{prop:respect-arcs}, we have
  $\tau'(y) < \tau'(z)$.  Since $\tau(y) = \tau(z)$, we have $\tau'(y) =
  \tau(y)$ and $\tau'(z) = \tau(z) + 1$.  Therefore, $y \in \{\rho, p, u\}$ and
  $z \notin \{\rho, p, u\}$.  Since $p$ is the only child of $\rho$, $u$ and $v$
  are the only children of $p$, and neither $u$ nor $v$ is a reticulation, we
  must have $y = u$.  We proved, however, that $u$ has no neighbour $z$ with
  $\tau(z) = \tau(u)$.  Thus, we obtain a contradiction, and $\tau'$ satisfies
  \cref{prop:horizontal-reticulation-arc}.
\end{proof}}

A \emph{partial orientation} (PO) $\vec{N}$ of an undirected network $N$ chooses
directions for a subset of the edges of $N$ and leaves the remaining edges
undirected.  \nz{In contrast to an orientation, it does not subdivide any edges
of $N$ nor add any vertices.}  We call a PO \emph{HGT-consistent} (short, an
HPO) if there exists an orientation $\dvec{N}$ of $N$ and an \nz{asymmetric}
HGT-consistent labelling of $\dvec{N}$ such that $\tau(x) < \tau(y)$ for every
\nz{arc} $\a{x, y}$ of $\vec{N}$, and $\tau(x) = \tau(y)$, for every edge $\e{x,
y}$ of $\vec{N}$.  In other words, $\vec{N}$ and $\dvec{N}$ agree on the
directions of vertical \nz{arcs} while \nz{every edge of $\vec{N}$ corresponds
to a horizontal arc of $\dvec{N}$.  The only exception is the edge $\e{u, v}$ of
$N$ that is subdivided to introduce the child $p$ of the root $\rho$ of
$\dvec{N}$.  Since $\tau$ is asymmetric, this edge is an arc in $\vec{N}$ and
corresponds to the two arcs $\a{p, u}$ and $\a{p, v}$ in $\dvec{N}$.}

By definition, any network that has an HPO also \nz{has an orchard orientation}.
Conversely, if $\dvec{N}$ is an orchard orientation of $N$, \nz{then $\dvec{N}$
has an asymmetric HGT-consistent labelling $\tau$, by
\cref{lem:asymmetric-drawing}.}  We obtain an HPO $\vec{N}$ of $N$ by
\nz{replacing each edge $\e{x, y}$ of $N$ with an arc $\a{x, y}$ if $\tau(x) <
\tau(y)$, and with an arc $\a{y, x}$ if $\tau(x) > \tau(y)$; if $\tau(x) =
\tau(y)$, then $\e{x, y}$ is an edge also in $\vec{N}$.} Thus, we have the
following observation.

\begin{obs}
  \label{obs:HGT-consistent}
  An unrooted binary network $N$ \nz{has an orchard orientation} if and only if
  it has an HPO.
\end{obs}

The next lemma characterizes HPOs of unrooted binary networks \nz{without
reference to an orientation or HGT-consistent labelling.}  A \emph{semi-directed
path} in a partially directed network $\vec{N}$ is a sequence of vertices
$\s{x_0, \ldots, x_n}$ such that for all $1 \le i \le n$, either $\a{x_{i-1},
x_i}$ is \nz{an arc} of $\vec{N}$ or $\e{x_{i-1}, x_i}$ is an edge of $\vec{N}$.
A \emph{semi-directed cycle} is a semi-directed path $\s{x_0, \ldots, x_n}$
\nz{with} $x_0 = x_n$.  In a partially directed network $\vec{N}$, the
\emph{in-degree} of a vertex $v$ is the number of \nz{arcs} $\a{u, v}$ \nz{with
head $v$} in $\vec{N}$, and the \emph{out-degree} of $v$ is the number of
\nz{arcs} $\a{v, w}$ \nz{with tail $v$} in $\vec{N}$.  \nz{Edges incident to $v$
are not counted.}

\begin{lem}
  \label{lem:HGT-consistent}
  A PO $\vec{N}$ of an unrooted binary network $N$ is an HPO if and only if
  \begin{enumerate}[label=(O\arabic{*}),widest=3,leftmargin=*,noitemsep]
    \item There exists \nz{exactly one vertex $r$ with only out-arcs and no
    incident edges in $\vec{N}$}, called the \nz{\emph{root}}
    of~$\vec{N}$,\label[prop]{prop:unique-root}
    \item Every vertex $v \ne r$ has in-degree $1$ \nz{in
    $\vec{N}$},\label[prop]{prop:in-degree}
    \item Every non-leaf vertex $v$ of $N$ has out-degree at least $1$ in
    $\vec{N}$, and\label[prop]{prop:out-degree}
    \item $\vec{N}$ contains no semi-directed
    cycle.\label[prop]{prop:cycle-free}
  \end{enumerate}
\end{lem}

\begin{proof}
  First assume that $\vec{N}$ is an HPO.  Then there exist an orientation
  $\dvec{N}$ of $N$ and an \nz{asymmetric} HGT-consistent labelling $\tau$ of
  $\dvec{N}$ such that $\tau(x) < \tau(y)$, for every \nz{arc} $\a{x, y}$ of
  $\vec{N}$, and $\tau(x) = \tau(y)$, for every edge $\e{x, y}$ of $\vec{N}$.

  \nz{Let $\rho$ be the root of $\dvec{N}$, and let $u$ and $v$ be its
  grandchildren.  Then $\e{u, v}$ is an edge of $N$.  Since $\tau$ is
  asymmetric, we have w.l.o.g.\ that $\tau(u) < \tau(v)$.  By
  \cref{prop:respect-arcs,prop:vertical-in-arc,prop:horizontal-reticulation-arc},
  every vertex $x \in \dvec{N}$ other than $\rho$ has a parent $y$ with $\tau(y)
  < \tau(x)$.  If $x$ is a vertex of $\vec{N}$ and $x \notin \{u, v\}$, then
  $\e{y, x}$ is an edge of $N$.  Since $\tau(y) < \tau(x)$, $\a{y, x}$ is an arc
  of $\vec{N}$, so $x$ has in-degree at least $1$ in $\vec{N}$.  Since $\tau(u)
  < \tau(v)$, $\a{u, v}$ is an arc of $\vec{N}$.  Thus, $v$ also has in-degree
  at least $1$ in $\vec{N}$.  This shows that the only vertex of $\vec{N}$ that
  can have in-degree $0$ is $u$.  If $u$ does not have in-degree $0$, then
  $\vec{N}$ contains a directed cycle, which is impossible because $\tau(x) <
  \tau(y)$, for every arc $\a{x, y}$ in $\vec{N}$.  Thus, $u$ does not have any
  in-arcs.  Since $\tau(u) < \tau(v)$, $\a{u, v}$ is an arc of $\vec{N}$.  Any
  other edge of $N$ incident to $u$ is an arc of $\dvec{N}$.  By
  \cref{lem:no-reticulations-near-root}, $u$ has no neighbour $z$ in $\dvec{N}$
  with $\tau(u) = \tau(z)$, so $u$ has no incident edges in $\vec{N}$.  This
  shows that all edges of $N$ incident to $u$ are out-arcs of $u$ in $\vec{N}$,
  that is, $r = u$ satisfies \cref{prop:unique-root}.}

  We just argued that any vertex $x \ne r$ has in-degree at least $1$ \nz{in
  $\vec{N}$.}  If $x$ is a leaf or tree node of $\dvec{N}$, then its parent $y$
  \nz{in $\dvec{N}$} satisfies $\tau(y) < \tau(x)$, \nz{by
  \cref{prop:vertical-in-arc}}, while every child $z$ of $x$ \nz{in $\dvec{N}$}
  (if $x$ is a tree node) satisfies $\tau(x) \le \tau(z)$, by
  \cref{prop:respect-arcs}.  Thus, $x$ has in-degree \nz{exactly $1$ in
  $\vec{N}$.}  Similarly, if $x$ is a reticulation \nz{in $\dvec{N}$}, then its
  child $z$ \nz{in $\dvec{N}$} satisfies $\tau(x) \le \tau(z)$, by
  \cref{prop:respect-arcs}.  For its two parents $y_1$ and $y_2$ \nz{in
  $\dvec{N}$,} we have $\tau(y_1) = \tau(x)$ or $\tau(y_2) = \tau(x)$, by
  \cref{prop:horizontal-reticulation-arc}.  Thus, $x$ has in-degree at most $1$
  \nz{in $\vec{N}$.}  Since $x$ has in-degree at least $1$, its in-degree is
  \nz{exactly $1$ in $\vec{N}$.} This proves that $\vec{N}$ satisfies
  \cref{prop:in-degree}.

  By \cref{prop:vertical-out-arc}, every non-leaf vertex \nz{$x \ne u$ of
  $\dvec{N}$} has a child $z$ with $\tau(x) < \tau(z)$.  This vertex is an
  out-neighbour of $x$ in $\vec{N}$.  Thus, $x$ has out-degree at least $1$
  \nz{in $\vec{N}$.  We already argued that $u$ has only out-arcs, so $u$ also
  has out-degree at least $1$.  This shows that} \cref{prop:out-degree} holds.

  \nz{To prove that $\vec{N}$ contains no semi-directed cycle, we use the
  following observation:

  \begin{obs}
    \label{obs:undirected-matching}
    If $\vec{N}$ is a partial orientation of a binary network that satisfies
    \cref{prop:unique-root,prop:in-degree,prop:out-degree}, then the edges
    of $\vec{N}$ form a matching.
  \end{obs}

  \begin{proof}
    Consider any vertex $x$ of $\vec{N}$.  If $x$ is the vertex $r$ in
    \cref{prop:unique-root}, then $x$ has no incident edges.  If $x$ is a leaf
    of $N$, then it has one incident edge in $N$ and, thus, also at most one
    incident edge in $\vec{N}$.  If $x \ne r$ and $x$ is not a leaf of $N$, then
    by \cref{prop:in-degree,prop:out-degree}, $x$ has an in-arc and an out-arc
    in $\vec{N}$.  Since $N$ is binary, this implies that $x$ has at most one
    incident edge in $\vec{N}$.
  \end{proof}}

  Now assume that $\vec{N}$ has a semi-directed cycle $C = \s{x_0, \ldots,
  x_n}$.  Then $\tau(x_{i-1}) \le \tau(x_i)$, for all $1 \le i \le n$.
  \nz{Since $x_0 = x_n$, this implies that $\tau(x_0) = \cdots = \tau(x_n)$,
  that is, $C$ consists entirely of edges, so the edges in $\vec{N}$ do not form
  a matching, contradicting \cref{obs:undirected-matching}.  Thus,
  \cref{prop:cycle-free} holds.}

  Now the other direction.  Assume that $\vec{N}$ satisfies
  \crefrange{prop:unique-root}{prop:cycle-free}.  \nz{Let $G$ be the directed
  graph obtained by contracting all edges in $\vec{N}$.  Then $G$ is acyclic.
  Indeed, by \cref{obs:undirected-matching}, every vertex $x \in G$ corresponds
  to a pair of vertices $x', x\dprime \in \vec{N}$ such that either $x' =
  x\dprime$ or $\e{x', x\dprime}$ is an edge of $\vec{N}$.  Thus, if $G$
  contains a cycle $C = \s{x_0, \ldots, x_n}$, then $\vec{N}$ contains a
  semi-directed cycle over some subset of the vertex set $\{x_1', x_1\dprime,
  \ldots, x_n', x_n\dprime\}$.  Since $\vec{N}$ contains no such cycle, by
  \cref{prop:cycle-free}, $G$ is acyclic.  By
  \cref{prop:unique-root,prop:in-degree}, $r$ is the only source of $G$.}

  Since $G$ is acyclic, we can topologically sort it and number its vertices in
  order.  Let $\tau'$ be the resulting labelling of the vertices in $G$.  For
  every vertex $x \in G$ and its two corresponding vertices $x', x\dprime \in
  \vec{N}$, we define $\tau(x') = \tau(x\dprime) = \tau'(x)$.  \nz{This ensures
  that $\tau$ satisfies the condition that $\tau(x) < \tau(y)$, for every arc
  $\a{x, y}$ of $\vec{N}$ (because $\a{x, y}$ is also an arc of $G$), and that
  $\tau(x) = \tau(y)$, for every edge $\e{x, y}$ of $\vec{N}$ (because $x$ and
  $y$ correspond to the same vertex of $G$).  We have to prove that there exists
  an orientation $\dvec{N}$ of $N$ and an extension of $\tau$ to the vertex set
  of $\dvec{N}$ that is an asymmetric HGT-consistent labelling of $\dvec{N}$.

  Let $u = r$, and let $v$ be the out-neighbour of $u$ with minimum label
  $\tau(v)$.  Note that $\tau(v) > \tau(u)$ because $u$ is the source of $G$.
  Then we obtain $\dvec{N}$ as follows: We subdivide the edge $\e{u, v}$ of $N$
  with a new vertex $p$ and create a root vertex $\rho$ whose child in
  $\dvec{N}$ is $p$.  In $\dvec{N}$, $u$ and $v$ are children of $p$.  For any
  edge $\e{x, y} \ne \e{u, v}$ of $N$, $\dvec{N}$~contains the arc $\a{x, y}$ if
  $\tau(x) < \tau(y)$, the arc $\a{y, x}$ if $\tau(x) > \tau(y)$, and one of
  these two arcs chosen arbitarily if $\tau(x) = \tau(y)$.  To extend $\tau$ to
  the vertex set of $\dvec{N}$, we set $\tau(p) = \tau(u) - 1$ and $\tau(\rho) =
  \tau(p) - 1$.  We need to verify that $\dvec{N}$ is an orientation of $N$ and
  that $\tau$ is an asymmetric HGT-consistent labelling of $\dvec{N}$.
  
  To see that $\dvec{N}$ is an orientation of $N$, note first that $\dvec{N}$
  is acyclic.  Indeed, if $\dvec{N}$ contains a cycle $C$, then $C$ cannot
  contain $\rho$, $p$ or $u$ because $\rho$ has no in-arcs, $p$'s only
  in-neighbour is $\rho$, and $u$'s only in-neighbour is $p$.  This implies that
  every arc $\a{x, y}$ in $C$ corresponds to the arc $\a{x, y}$ or to the edge
  $\e{x, y}$ in $\vec{N}$.  This makes $C$ a semi-directed cycle in $\vec{N}$,
  contradicting \cref{prop:cycle-free}.
  
  Next note that every vertex $x \ne u$ has in-degree at least $1$ in $\vec{N}$,
  by \cref{prop:in-degree}.  Thus, it also has an in-neighbour in $\dvec{N}$.
  The vertices $u$ and $p$ have $p$ and $\rho$ as in-neighbours, respectively.
  This makes $\rho$ the only root of $\dvec{N}$.  Finally, we already observed
  that $u$ has $p$ as an in-neighbour.  Thus, it is a leaf of $\dvec{N}$ if it
  is a leaf of $N$.  Any other leaf $x \ne u$ of $N$ has in-degree $1$ in
  $\vec{N}$, by \cref{prop:in-degree}, so it also has in-degree $1$ in
  $\dvec{N}$.  In particular, it has no out-neighbours.  This shows that every
  leaf of $N$ is a leaf of $\dvec{N}$, and all internal vertices of $\dvec{N}$
  other than $\rho$ are tree vertices or reticulations.  Thus, $\dvec{N}$ is an
  orientation of $N$.
  
  To verify that $\tau$ is an HGT-consistent labelling of $\dvec{N}$, observe
  that every arc $\a{x, y} \in \{\a{\rho, p}, \a{p, u}, \a{p, v}\}$ satisfies
  $\tau(x) < \tau(y)$ because we explicitly ensure that $\tau(\rho) < \tau(p) <
  \tau(u)$ and, as observed above, $\tau(u) < \tau(v)$.  Any other arc $\a{x,
  y}$ of $\dvec{N}$ is either an arc $\a{x, y}$ or an edge $\e{x, y}$ of
  $\vec{N}$.  Thus, $\tau(x) \le \tau(y)$.  This shows that $\tau$ satisfies
  \cref{prop:respect-arcs}.
  
  The vertices $\rho$, $p$, and $u$ satisfy \cref{prop:vertical-out-arc} because
  $\tau(\rho) < \tau(p) < \tau(u) < \tau(v)$.  Any other internal vertex $x$ of
  $\dvec{N}$ is a vertex of $\vec{N}$ and is not a leaf of $N$.  Thus, by
  \cref{prop:out-degree}, there exists an out-arc $\a{x, y}$ of $x$ in
  $\vec{N}$.  This implies that $\tau(x) < \tau(y)$, that is, $x$ has an
  out-neighbour $y$ in $\dvec{N}$ with $\tau(x) < \tau(y)$.  This shows that
  $\tau$ satisfies \cref{prop:vertical-out-arc}.

  Similarly, if $x$ is a leaf or a tree vertex of $\dvec{N}$, then $x \ne \rho$
  and $x$ has a unique in-arc $\a{y, x}$ in $\dvec{N}$.  If $x = p$, then $y =
  \rho$ and $\tau(\rho) < \tau(p)$.  If $x \in \{u, v\}$, then $p$ is the parent
  of $x$ because $\tau(p) < \tau(u) < \tau(v)$ and $\tau$ satisfies
  \cref{prop:respect-arcs}.  For any other leaf or tree vertex $x$ of
  $\dvec{N}$, every out-arc $\a{x, z}$ of $x$ corresponds to the arc $\a{x, z}$
  or the edge $\e{x, z}$ in $\vec{N}$.  Neither contributes to the in-degree of
  $x$ in $\vec{N}$.  Since $x$ has in-degree $1$ in $\vec{N}$, by
  \cref{prop:in-degree}, $\a{y, x}$ must therefore be an in-arc of $x$ in
  $\vec{N}$, which implies that $\tau(y) < \tau(x)$.  This shows that
  \cref{prop:vertical-in-arc} holds.
  
  Finally, if $x$ is a reticulation of $\dvec{N}$ with parents $y_1$ and $y_2$,
  then $\tau(y_1) \le \tau(x)$ and $\tau(y_2) \le \tau(x)$, by
  \cref{prop:respect-arcs}.  We cannot have both $\tau(y_1) < \tau(x)$ and
  $\tau(y_2) < \tau(x)$ because this would imply that $x$ has in-degree greater
  than $1$ in $\vec{N}$, contradicting \cref{prop:in-degree}.  Thus, $\tau(y_1)
  = \tau(x)$ or $\tau(y_2) = \tau(x)$.  If $\tau(y_1) = \tau(y_2) = \tau(x)$,
  then both $\e{x, y_1}$ and $\e{x, y_2}$ are edges of $\vec{N}$.  This
  contradicts \cref{obs:undirected-matching}.  Thus, $\tau(y_1) = \tau(x)$ or
  $\tau(y_2) = \tau(x)$, but not both, and
  \cref{prop:horizontal-reticulation-arc} holds.}
\end{proof}

\nz{In \cref{sec:reduction}, we will use \cref{lem:HGT-consistent} to prove that
a formula $F$ in 3-CNF is satisfiable if and only if a network $N(F)$ we
construct from it has an HPO.} We will also frequently use the following two
lemmas.

\begin{lem}
  \label{lem:no-semi-directed-paths}
  Let $\e{x, y}$ be an edge of an unrooted binary network $N$, and let $\vec{N}$
  be an HPO of $N$.  Then $\vec{N}$ contains no semi-directed path \nz{of length
  at least $2$}
  from $x$ to $y$
  \nz{and which ends in an arc.}
\end{lem}

\begin{proof}
  Assume that there exists such a path $P$.  \nz{Then the final arc in $P$ is an
  arc $\a{z, y}$ with 
  \yuki{$x\ne z$.}
  Thus, by \cref{prop:in-degree}, $\vec{N}$ does
  not contain the arc $\a{x, y}$, so $\vec{N}$ contains the arc $\a{y, x}$ or
  the edge $\e{x, y}$.}  This \nz{arc or} edge together with $P$ forms a
  semi-directed cycle in $\vec{N}$, a contradiction to \cref{prop:cycle-free}.
\end{proof}

\begin{lem}
  \label{lem:out-edge-pair}
  Let $N$ be an unrooted binary network, let $\vec{N}$ be an HPO of $N$, and let
  $C = \s{v_0, v_1, \ldots, v_n = v_0}$ be a chordless cycle in $N$.  For all $1
  \le i \le n$, let $u_i$ be the \nz{unique} neighbour of $v_i$ \nz{not in $C$.}
  Then there exists an index $i \in [n]$ such that  \nz{$\vec{N}$ contains the
  edge $\e{v_{i-1}, v_i}$ and the two arcs $\a{v_{i-1}, u_{i-1}}$ and $\a{v_i,
  u_i}$.}
\end{lem}

\begin{proof}
  First we prove that there exists at least one index $i \in [n]$ such that
  $\e{v_{i-1}, v_i}$ \nz{is an edge of $\vec{N}$.}  Assume the contrary and
  consider the restriction $\vec{C}$ of $\vec{N}$ to $C$.  Since the total
  in-degree of all vertices in $\vec{C}$ equals their total out-degree, this
  implies that we either have a vertex of in-degree $2$ in $\vec{C}$ or every
  vertex has in-degree $1$ and out-degree $1$ in $\vec{C}$.  The former case
  contradicts \cref{prop:in-degree}; the latter, \cref{prop:cycle-free}.

  Now let $i_1 < \cdots < i_t$ be all indices in $[n]$ such that $\e{v_{i_h -
  1}, v_{i_h}}$ \nz{is an edge of $\vec{N}$.}  Assume that there is no index $h$
  such that \nz{both $\a{v_{i_h - 1}, u_{i_h - 1}}$ and $\a{v_{i_h}, u_{i_h}}$
  are arcs of $\vec{N}$.}  Since $\e{v_{i_1 - 1}, v_{i_1}}$ \nz{is an edge of
  $\vec{N}$,} either \nz{$\a{v_{i_1 - 1}, u_{i_1 - 1}}$ or $\a{u_{i_1 - 1},
  v_{i_1 - 1}}$ is an arc of $\vec{N}$}, and either \nz{$\a{v_{i_1}, u_{i_1}}$
  or $\a{u_{i_1}, v_{i_1}}$ is an arc of $\vec{N}$,} both by
  \cref{obs:undirected-matching}.  Since $\a{v_{i_1 - 1}, u_{i_1 - 1}}$ and
  $\a{v_{i_1}, u_{i_1}}$ aren't both \nz{arcs of} $\vec{N}$, we can assume
  w.l.o.g.\ that $\a{u_{i_1}, v_{i_1}}$ \nz{is an arc of $\vec{N}$}.  By
  \cref{prop:in-degree}, this implies that $\s{v_{i_1}, \ldots, v_{i_2 - 1}}$ is
  a directed path in $\vec{N}$.  By \cref{prop:out-degree}, this implies that
  $\a{v_{i_2 - 1}, u_{i_2 - 1}}$ is \nz{an arc}.  Since $\a{v_{i_2 - 1}, u_{i_2
  - 1}}$ and $\a{v_{i_2}, u_{i_2}}$ aren't both \nz{arcs of} $\vec{N}$, and
  since $\e{v_{i_2 - 1}, v_{i_2}}$ \nz{is an edge of} $\vec{N}$, this implies
  that $\a{u_{i_2}, v_{i_2}}$ \nz{is an arc of} $\vec{N}$.

  By repeating this argument for $2 \le h \le t$, we conclude that for all $1
  \le h \le t$, $\s{v_{i_h}, \ldots, v_{i_{h+1} - 1}}$ is a directed path in
  $\vec{N}$.  Together, with the edges $\e{v_{i_1 - 1}, v_{i_1}}, \ldots,
  \e{v_{i_t - 1},v_{i_t}}$, this makes $C$ a semi-directed cycle in $\vec{N}$,
  a~contradiction to \cref{prop:cycle-free}.  Thus, there must exist an index
  $h$ such that $\a{v_{i_h - 1}, u_{i_h - 1}}$ and $\a{v_{i_h}, u_{i_h}}$
  \nz{are both arcs of} $\vec{N}$.
\end{proof}

\section{NP-Hardness of Unrooted Orchard \nz{Orientation}}

\nz{We are ready to prove that \textsc{Orchard Orientation} is NP-hard even for
binary networks.}  We prove this by constructing a \nz{binary} network $N =
N(F)$ from any Boolean formula $F$ in 3-CNF such that $F$ is satisfiable if and
only \nz{if} $N$ \nz{has an HPO}. Since \textsc{3-SAT} is NP-hard, this proves that so is
\nz{\textsc{Orchard Orientation}}.

\nz{We construct $N$ from certain subgraphs, called \emph{widgets}, that
represent vertices and clauses in $F$, and the interactions between them.  We
describe these widgets in \cref{sec:widgets}.  \cref{sec:reduction} shows how to
combine them to obtain $N$ and proves that $N$ has an HPO if and only if $F$ is
satisfiable.  This proof relies on a number of properties of the widgets, which
are shown in \cref{sec:widgets}.}

\subsection{Widgets}

\label{sec:widgets}

We construct the network $N$ from six types of widgets, \nz{five types of
\emph{gate widgets} that mimic logical gates and a \emph{wire widget} used to
mimic the wires correcting these gates to form a circuit}.  Every widget $W$ has
one or more edges $\e{x_1, x_2}$ such that in $N$, $x_1$ and $x_2$ have
neighbours $x_1'$ and $x_2'$ not in $W$.  We call $x_1'$ and $x_2'$ the
\emph{external neighbours} of $x_1$ and $x_2$, respectively.  \nz{If $W$ mimics
a gate, then} the two edges $\e{x_1, x_1'}$ and $\e{x_2, x_2'}$ together mimic
an input or output wire of this gate.  Accordingly, we call $\e{x_1, x_2}$
\nz{an} \emph{input} or \emph{output} of $W$.  \nz{The network $N$ will be
constructed by identifying each input of a gate widget with the output of some
wire widget, and each output of a gate widget with the input of some wire
widget.}

We can view an HPO $\vec{N}$ of $N$ as assigning states to \nz{the} inputs and
outputs of \nz{gate} widgets.  Specifically, we consider an input $\e{x_1, x_2}$
to be {\true} if $\a{x_1, x_1'}$ and $\a{x_2, x_2'}$ \nz{are arcs of} $\vec{N}$,
{\false} if $\a{x_1', x_1}$ and $\a{x_2', x_2}$ \nz{are arcs of} $\vec{N}$, and
undefined otherwise.  Similarly, we consider an output $\e{x_1, x_2}$ to be
{\true} if $\a{x_1', x_1}$ and $\a{x_2', x_2}$ \nz{are arcs of} $\vec{N}$,
{\false} if $\a{x_1, x_1'}$ and $\a{x_2, x_2'}$ \nz{are arcs of} $\vec{N}$, and
undefined otherwise.  Note the opposite orientation of the edges \nz{incident
to} {\true} and {\false} inputs and outputs.  This ensures that when connecting
the output of one widget $W_1$ to the input of another widget $W_2$, we
interpret the edges that form this connection as the same truth value whether we
look at them as output edges of $W_1$ or as input edges of $W_2$.

The six types of widgets from which we build $N$ are:
\begin{itemize}
  \item A \emph{choice widget} represents a Boolean variable $x$ in $F$ and
  forces us to choose whether $x$ is {\true} or {\false} by choosing to
  \nz{leave} one of two edges in this widget undirected in any HPO of $N$.
  \item A \emph{clause widget} represents a clause.  It has three inputs
  representing the literals in this clause and one output representing their
  disjunction.  If the output is {\true}, then at least one of the inputs must
  be {\true} in any HPO of $N$.
  \item An \emph{and-widget} has a number of inputs and one output representing
  the conjunction of the inputs.  
  If the output is
  {\true}, then all inputs must be {\true} in any HPO of $N$.
  \item A \emph{replicator widget} is an upside-down and-widget.  It has one
  input and a number of outputs.  Intuitively, it is used to replicate the input
  value across multiple outputs.  The replicator \nz{actually} ensures only that
  if one of its outputs is {\true}, its input must also be {\true}, which is
  sufficient for our purposes.
  \item A \emph{root widget}, the sheep, is a widget that must contain the root
  of any HPO of $N$.
  \item A \emph{wire widget} has one input and one output and is used as a
  ``wire'' to connect the output of one widget to the input of another widget.
\end{itemize}

For each widget type, we define a set of POs that we call \emph{Boolean POs}
(BPOs).  We use these BPOs in the proof that $N$ \nz{has an HPO} if $F$ is
satisfiable.  Specifically, we will show that the combination of these BPOs of
the widgets is an HPO of $N$ if $F$ is satisfiable.

\subsubsection{Root Widget (Sheep)}

The root widget is shown in \cref{fig:root-widget}.  The BPO of this widget is
shown in \cref{fig:root-widget-labelling}.  The following observation can be
verified by inspection of \cref{fig:root-widget-labelling}.

\begin{figure}[ht]
  \centering
  \subcaptionbox{\label{fig:root-widget}}{\begin{tikzpicture}
    \pic (sheep) {sheep};
    \path [external edge] 
    (sheep-in-1) -- +(270:1) node [vertex] (i1) {}
    (sheep-in-2) -- +(270:1) node [vertex] (i2) {};
    \path
    (sheep-root) node [anchor=south]                       (r-l)  {$r$}
    (sheep-v-1)  node [anchor=north]                              {$u$}
    (sheep-v-2)  node [anchor=east]                        (v-l)  {$v$}
    (sheep-v-3)  node [anchor=west]                        (w-l)  {$w$}
    (sheep-in-1) node [anchor=east,xshift=1pt,yshift=-3pt] (i1-l) {$i_1$}
    (sheep-in-2) node [anchor=west,xshift=1pt,yshift=-3pt] (i2-l) {$i_2$}
    (barycentric cs:sheep-in-1=0.5,i1=0.5) coordinate (bot)
    (barycentric cs:i1=0.5,i2=0.5) node [anchor=north,yshift=-2pt] {\shortstack{Input from\\other widgets}};
    \begin{scope}[on background layer]
      \node [region,fit=(r-l) (v-l) (w-l) (i1-l) (i2-l) (bot)] {};
    \end{scope}
  \end{tikzpicture}}%
  \hspace{1in}%
  \subcaptionbox{\label{fig:root-widget-labelling}}{\begin{tikzpicture}
    \pic (sheep) {sheep boolean};
    \path
    (sheep-in-1) +(270:1) node [vertex] (i1) {}
    (sheep-in-2) +(270:1) node [vertex] (i2) {};
    \path [external edge,->] (sheep-in-1) -- (i1);
    \path [external edge,->] (sheep-in-2) -- (i2);
    \path
    (sheep-root) node [anchor=south]                       (r-l)  {\phantom{$r$}}
    (sheep-v-1)  node [anchor=north]                              {\phantom{$u$}}
    (sheep-v-2)  node [anchor=east]                        (v-l)  {\phantom{$v$}}
    (sheep-v-3)  node [anchor=west]                        (w-l)  {\phantom{$v$}}
    (sheep-in-1) node [anchor=east,xshift=1pt,yshift=-3pt] (i1-l) {\phantom{$i_1$}}
    (sheep-in-2) node [anchor=west,xshift=1pt,yshift=-3pt] (i2-l) {\phantom{$i_2$}}
    (barycentric cs:sheep-in-1=0.5,i1=0.5) coordinate (bot)
    (barycentric cs:i1=0.5,i2=0.5) node [anchor=north,yshift=-2pt] {\phantom{\shortstack{Input from\\other widgets}}};
    \begin{scope}[on background layer]
      \node [region,fit=(r-l) (v-l) (w-l) (i1-l) (i2-l) (bot)] {};
    \end{scope}
  \end{tikzpicture}}
  \caption{(a) The root widget (the sheep). \nz{The widget consists of the
  vertices and edges in the shaded region.  The dashed edges are connections to
  vertices in other widgets.}  (b) Its BPO.  \nz{Arcs are smooth lines; edges
  are zigzag lines}.}
\end{figure}
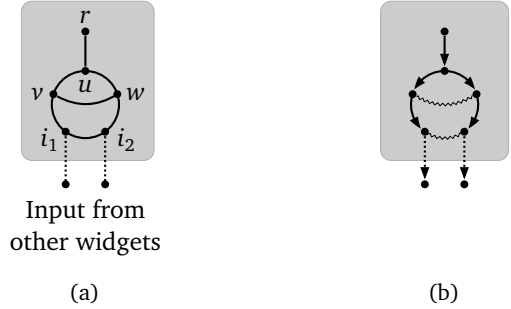

\begin{obs}
  \label{obs:root-widget-boolean}
  Let $N$ be an unrooted binary network that contains a root widget $R$, and let
  $\vec{N}$ be a PO of $N$ whose restriction to $R$ is the BPO of $R$ and such
  that the input $\e{i_1, i_2}$ of $R$ is {\true}.  Then all vertices in $R$
  satisfy \cref{prop:in-degree,prop:out-degree}, any semi-directed cycle in
  $\vec{N}$ must include a vertex not in $R$, \nz{and $\rho$ is a root of
  $\vec{N}$}.
\end{obs}

\begin{lem}
  \label{lem:root-in-sheep}
  Let $N$ be an unrooted binary network that contains a root widget $R$, and let
  $\vec{N}$ be an HPO of $N$. \nz{Then the root of $\vec{N}$ is one of the
  vertices $r$, $u$, $v$, and $w$ in $R$.}
\end{lem}

\begin{proof}
  Consider the triangle $T = \s{u, v, w}$ in $R$.  By
  \cref{obs:undirected-matching,lem:no-semi-directed-paths}, there must exist a
  vertex $x \in T$ such that \nz{$\vec{N}$ contains the arcs $\a{x, y}$ and
  $\a{x, z}$, and the edge $\e{y, z}$,} where $y$ and $z$ are the other two
  vertices in $T$.
  
  If 
  \yuki{$x=u$,}
  then \nz{$\a{r,u}$ or $\a{u, r}$ is an arc of $\vec{N}$ because
  otherwise, both $r$ and $u$ would have in-degree $0$ (and an incident edge),
  contradicting \cref{prop:unique-root,prop:in-degree} of an HPO.} This implies
  that $r$ \nz{or $u$} has in-degree $0$ and no incident edges \nz{in
  $\vec{N}$.}  Thus, it is the root of $\vec{N}$.

  If 
  \yuki{$x\ne u$,}
  we can assume w.l.o.g.\ that 
  \yuki{$x=v$}
  because the root widget is
  symmetric.  In this case, \nz{$\vec{N}$ contains the arcs $\a{v, u}$ and
  $\a{v, w}$, and the edge $\e{u,w}$.}  By \cref{prop:out-degree}, this implies
  that $\a{w, i_2}$ \nz{is also an arc of} $\vec{N}$.  \nz{If $\vec{N}$ contains
  the arc $\a{i_1, v}$ or the edge $\e{i_1, v}$, then $\s{i_1, v, w, i_2}$ is a
  semi-directed path that ends in an arc, contradicting
  \cref{lem:no-semi-directed-paths}.  Thus, $\vec{N}$ contains the arc $\a{v,
  i_1}$.  This shows that $v$ has only out-arcs and thus must be the root
  of~$\vec{N}$.}
\end{proof}

\subsubsection{\nz{Wire} Widget}

The \nz{wire} widget is shown in \cref{fig:connector-widget}.  Its purpose is to
ensure that either its input is {\true} or its output is {\false}.

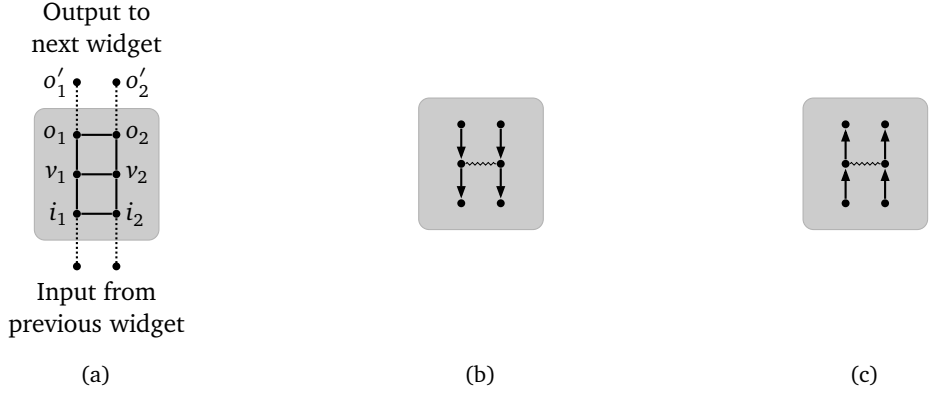
\begin{figure}[ht]
  \centering
  \subcaptionbox{\label{fig:connector-widget}}{\begin{tikzpicture}
    \pic (connector) {connector};
    \path [external edge] 
    (connector-in-1)  -- +(270:1) node [vertex] (i1) {}
    (connector-in-2)  -- +(270:1) node [vertex] (i2) {}
    (connector-out-1) -- +(90:1)  node [vertex] (o1) {}
    (connector-out-2) -- +(90:1)  node [vertex] (o2) {}
    (barycentric cs:connector-in-1=0.5,i1=0.5)  coordinate (bot)
    (barycentric cs:connector-out-1=0.5,o1=0.5) coordinate (top);
    \path
    (connector-v-1)   node [anchor=east] (v1-l) {$v_1$}
    (connector-v-2)   node [anchor=west] (v2-l) {$v_2$}
    (connector-in-1)  node [anchor=east] (i1-l) {$i_1$}
    (connector-in-2)  node [anchor=west] (i2-l) {$i_2$}
    (connector-out-1) node [anchor=east] (o1-l) {$o_1$}
    (connector-out-2) node [anchor=west] (o2-l) {$o_2$}
    (o1)              node [anchor=east]        {$o_1'$}
    (o2)              node [anchor=west]        {$o_2'$}
    (barycentric cs:i1=0.5,i2=0.5) node [anchor=north,yshift=-2pt] {\shortstack{Input from\\previous widget}}
    (barycentric cs:o1=0.5,o2=0.5) node [anchor=south,yshift=6pt]  {\shortstack{Output to\\next widget}};
    \begin{scope}[on background layer]
      \node [region,fit=(i1-l) (i2-l) (o1-l) (o2-l) (v1-l) (v2-l) (bot) (top)] {};
    \end{scope}
  \end{tikzpicture}}%
  \hspace{1in}%
  \subcaptionbox{\label{fig:connector-true-labelling}}{\begin{tikzpicture}
    \pic (connector) {connector true};
    \path
    (connector-in-1) +(270:1) coordinate (i1)
    (connector-in-2) +(270:1) coordinate (i2)
    (barycentric cs:connector-in-1=0.5,i1=0.5)  coordinate (bot)
    (barycentric cs:connector-out-1=0.5,o1=0.5) coordinate (top);
    \path
    (connector-v-1)   node [anchor=east] (v1-l) {\phantom{$v_1$}}
    (connector-v-2)   node [anchor=west] (v2-l) {\phantom{$v_1$}}
    (connector-in-1)  node [anchor=east] (i1-l) {\phantom{$i_1$}}
    (connector-in-2)  node [anchor=west] (i2-l) {\phantom{$i_1$}}
    (connector-out-1) node [anchor=east] (o1-l) {\phantom{$o_1$}}
    (connector-out-2) node [anchor=west] (o2-l) {\phantom{$o_1$}}
    (barycentric cs:i1=0.5,i2=0.5) node [anchor=north,yshift=-6pt] {\phantom{\shortstack{Input from\\previous widget}}};
    \begin{scope}[on background layer]
      \node [region,fit=(i1-l) (i2-l) (o1-l) (o2-l) (v1-l) (v2-l) (bot) (top)] {};
    \end{scope}
  \end{tikzpicture}}%
  \hspace{1in}%
  \subcaptionbox{\label{fig:connector-false-labelling}}{\begin{tikzpicture}
    \pic (connector) {connector false};
    \path
    (connector-in-1) +(270:1) coordinate (i1)
    (connector-in-2) +(270:1) coordinate (i2)
    (barycentric cs:connector-in-1=0.5,i1=0.5)  coordinate (bot)
    (barycentric cs:connector-out-1=0.5,o1=0.5) coordinate (top);
    \path
    (connector-v-1)   node [anchor=east] (v1-l) {\phantom{$v_1$}}
    (connector-v-2)   node [anchor=west] (v2-l) {\phantom{$v_1$}}
    (connector-in-1)  node [anchor=east] (i1-l) {\phantom{$i_1$}}
    (connector-in-2)  node [anchor=west] (i2-l) {\phantom{$i_1$}}
    (connector-out-1) node [anchor=east] (o1-l) {\phantom{$o_1$}}
    (connector-out-2) node [anchor=west] (o2-l) {\phantom{$o_1$}}
    (barycentric cs:i1=0.5,i2=0.5) node [anchor=north,yshift=-6pt] {\phantom{\shortstack{Input from\\previous widget}}};
    \begin{scope}[on background layer]
      \node [region,fit=(i1-l) (i2-l) (o1-l) (o2-l) (v1-l) (v2-l) (bot) (top)] {};
    \end{scope}
  \end{tikzpicture}}  
  \caption{(a) A \nz{wire} widget.  \nz{The widget consists of the vertices and
  edges in the shaded region.  The dashed edges are connections to vertices in
  other widgets.}  (b) its \true-BPO.  (c) Its \false-BPO.  In figures (b) and
  (c), \nz{arcs are smooth lines; edges are zigzag lines.} The orientations of
  the edges $\e{o_1, o_2}$ and $\e{i_1, i_2}$ are not shown \nz{in figures (b)
  and (c)}, as they are determined by the other widgets containing these edges.}
\end{figure}

\begin{lem}
  \label{lem:connector-widget}
  Let $N$ be an unrooted binary network that contains \nz{a root widget $R$ and
  a wire} widget $L$, \nz{let $\e{i_1, i_2}$ be the input of $L$, let $\e{o_1,
  o_2}$ be the output of $L$, let $v_1$ and $v_2$ be the two interior vertices
  of $L$,} let $o_1'$ and $o_2'$ be the external neighbours of $o_1$ and $o_2$,
  respectively \nz{(shown in \cref{fig:connector-widget}),} and let $\vec{N}$ be
  an HPO of $N$.  If $\a{o_i', o_i}$ \nz{is an arc of $\vec{N}$, for some $i \in
  \{1, 2\}$,} then $\a{v_1, i_1}$ and $\a{v_2, i_2}$ \nz{are arcs of} $\vec{N}$
  and there exist semi-directed paths from $o_i$ to both $i_1$ and $i_2$ in
  $\vec{N}$.
\end{lem}

\begin{proof}
  \nz{Note that all vertices of $L$ must satisfy
  \cref{prop:in-degree,prop:out-degree} because the root of $\vec{N}$ is a
  vertex of $R$ not contained in any other widget, by \cref{lem:root-in-sheep}.}

  Assume w.l.o.g.\ that $\a{o_1', o_1}$ \nz{is an arc of} $\vec{N}$.  By
  \cref{lem:out-edge-pair} applied to the cycle $\s{o_1, o_2, v_2, v_1}$,
  \nz{$\vec{N}$ contains the edge $\e{v_1, v_2}$ and the arcs $\a{v_1, i_1}$ and
  $\a{v_2, i_2}$, or the edge $\e{v_2, o_2}$ and the arcs $\a{o_2, o_2'}$ and
  $\a{v_2, i_2}$.}
  
  In the former case, $\a{o_1, v_1}$ \nz{is also an arc of} $\vec{N}$, by
  \cref{prop:in-degree}.  Thus, we have semi-directed paths $\s{o_1, v_1, i_1}$
  and $\s{o_1, v_1, v_2, i_2}$ in $\vec{N}$, and the lemma holds in this case.

  In the latter case, $\a{o_1, o_2}$ and $\a{v_1, v_2}$ \nz{are also arcs of}
  $\vec{N}$, and $\a{v_1, o_1}$ \nz{is not an arc of} $\vec{N}$, by
  \cref{prop:in-degree}.  The latter implies that either $\e{o_1, v_1}$ \nz{is
  an edge of} $\vec{N}$ or $\a{o_1, v_1}$ \nz{is an arc of} $\vec{N}$.

  Since $\a{v_1, v_2}$ and $\a{v_2, i_2}$ \nz{are arcs of} $\vec{N}$,
  \nz{neither $\a{i_1, v_1}$ nor $\a{i_2, i_1}$ can be an arc of} $\vec{N}$, by
  \cref{lem:no-semi-directed-paths}.  By \cref{prop:in-degree}, $\a{i_1, i_2}$
  \nz{is not an arc of} $\vec{N}$.  Thus, $\e{i_1, i_2}$ \nz{is an edge of}
  $\vec{N}$ and, by \cref{obs:undirected-matching}, $\e{v_1, i_1}$ \nz{is not an
  edge of} $\vec{N}$.  Thus, $\a{v_1, i_1}$ \nz{is an arc of} $\vec{N}$ and once
  again, $\s{o_1, v_1, i_1}$ and $\s{o_1, v_1, v_2, i_2}$ are semi-directed
  paths in~$\vec{N}$.  This shows that the lemma holds also in this case.
\end{proof}

From here on, when we say that we \emph{join} an output of some widget $W_1$ to
the input of another widget $W_2$, then we mean that we identify the output of
$W_1$ with the input of a \nz{wire} widget $L$, and the input of $W_2$ with the
output of $L$.

Since we view \nz{wire} widgets as the wires connecting logical gates
represented by the other widgets, we have two BPOs of the \nz{wire} widget
representing a {\true} wire or a {\false} wire.  We call these BPOs the
\true-BPO and the \false-BPO of the \nz{wire} widget, respectively.  They are
shown in \cref{fig:connector-true-labelling,fig:connector-false-labelling}.  The
following observation is easily verified by inspection of
\cref{fig:connector-true-labelling,fig:connector-false-labelling}.

\begin{obs}
  \label{obs:connector-boolean}
  Let $N$ be an unrooted binary network that contains a \nz{wire} widget $L$,
  and let $\vec{N}$ be a PO of $N$ whose restriction to $L$ is one of the two
  BPOs of $L$.  Then \nz{the two interior vertices} $v_1$ and $v_2$ \nz{of $L$}
  satisfy \cref{prop:in-degree,prop:out-degree}, any semi-directed cycle in
  $\vec{N}$ includes a vertex contained in a different widget, and \nz{neither
  $v_1$ nor $v_2$ is} the root of $\vec{N}$.
\end{obs}

\subsubsection{Choice Widget}

The \emph{choice widget} corresponding to a variable $x$ is shown in
\cref{fig:choice-widget}.  The two bold edges $\e{d_x, e_x}$ and $\e{d_{\bar x},
e_{\bar x}}$ represent the literals $x$ and $\bar x$.  As the following lemma
shows, only one of these two edges can be undirected in any HPO of $N$.
Choosing $\e{d_x, e_x}$ to be undirected corresponds to setting $x$ to {\true},
while choosing $\e{d_{\bar x}, e_{\bar x}}$ to be undirected corresponds to setting
$x$ to {\false}.

\begin{figure}[t]
  \centering
  \subcaptionbox{\label{fig:choice-widget}}{\begin{tikzpicture}
    \pic (choice) {choice};
    \path [external edge]
    (choice-true-1)  -- ++(180:1) node [vertex] (t1) {} -- +(180:1) node [vertex] (it1) {}
    (choice-true-2)  -- ++(180:1) node [vertex] (t2) {} -- +(180:1) node [vertex] (it2) {}
    (choice-false-1) -- ++(0:1)   node [vertex] (f1) {} -- +(0:1)   node [vertex] (if1) {}
    (choice-false-2) -- ++(0:1)   node [vertex] (f2) {} -- +(0:1)   node [vertex] (if2) {}
    (choice-in-1)    -- +(270:1)  node [vertex] (v1) {}
    (choice-in-2)    -- +(270:1)  node [vertex] (v2) {}
    (v1)             -- +(270:1)  node [vertex] (i1) {}
    (v2)             -- +(270:1)  node [vertex] (i2) {}
    (v1)             -- (v2)
    (i1)             -- (i2)
    (t1)             -- (t2)
    (it1)            -- (it2)
    (f1)             -- (f2)
    (if1)            -- (if2)
    (barycentric cs:choice-true-1=0.5,t1=0.5)  coordinate (left)
    (barycentric cs:choice-false-1=0.5,f1=0.5) coordinate (right)
    (barycentric cs:choice-in-1=0.5,v1=0.5)    coordinate (bot);
    \path
    (barycentric cs:it1=0.5,it2=0.5) node [anchor=east,xshift=-2pt]                    {Output $x$}
    (barycentric cs:if1=0.5,if2=0.5) node [anchor=west,xshift=2pt]                     {Output $\bar x$}
    (barycentric cs:i1=0.5,i2=0.5)   node [anchor=north,yshift=-6pt]                   {Feedback input}
    (choice-true-3)                  node [anchor=south]                       (ax-l)  {$a_x$}
    (choice-true-2)                  node [anchor=south west]                          {$b_x$}
    (choice-true-1)                  node [anchor=north]                       (cx-l)  {$c_x$}
    (choice-true-pv)                 node [anchor=south]                               {$d_x$}
    (choice-true-pu)                 node [anchor=south]                               {$e_x$}
    (choice-true-v)                  node [anchor=north]                       (fx-l)  {$f_x$}
    (choice-true-u)                  node [anchor=north]                       (gx-l)  {\vphantom{$f_x$}$g_x$}
    (choice-false-3)                 node [anchor=south]                       (nax-l) {$a_{\bar x}$}
    (choice-false-2)                 node [anchor=south east]                          {$b_{\bar x}$}
    (choice-false-1)                 node [anchor=north]                       (ncx-l) {$c_{\bar x}$}
    (choice-false-pv)                node [anchor=south]                               {$d_{\bar x}$}
    (choice-false-pu)                node [anchor=south]                               {$e_{\bar x}$}
    (choice-false-v)                 node [anchor=north]                       (nfx-l) {$f_{\bar x}$}
    (choice-false-u)                 node [anchor=north]                       (ngx-l) {\vphantom{$f_{\bar x}$}$g_{\bar x}$}
    (choice-true-start)              node [anchor=west]                                {$h_x$}
    (choice-false-start)             node [anchor=east]                                {$h_{\bar x}$}
    (choice-in-1)                    node [anchor=east,xshift=1pt,yshift=-5pt] (i1-l)  {$i_1$}
    (choice-in-2)                    node [anchor=west,xshift=1pt,yshift=-5pt] (i2-l)  {$i_2$}
    (choice-split-1)                 node [anchor=south east]                          {$w$}
    (choice-split-2)                 node [anchor=south west]                          {$y$}
    (choice-split-3)                 node [anchor=south east]                          {$v$}
    (choice-split-4)                 node [anchor=south]                       (u-l)   {$u$}
    (v1)                             node [anchor=east,yshift=-2pt]                    {$v_1'$}
    (v2)                             node [anchor=west,yshift=-2pt]                    {$v_2'$}
    (i1)                             node [anchor=east]                                {$i_1'$}
    (i2)                             node [anchor=west]                                {$i_2'$};
    \begin{scope}[on background layer]
      \node [region,fit=(ax-l) (cx-l) (fx-l) (gx-l) (nax-l) (ncx-l) (nfx-l) (ngx-l) (i1-l) (i2-l) (u-l) (left) (right) (bot)] {};
    \end{scope}
  \end{tikzpicture}}\\[\bigskipamount]
  \subcaptionbox{\label{fig:choice-widget-true}}{\begin{tikzpicture}
    \pic (choice) {choice true};
    \path (choice-true-1)  +(180:1) node [vertex] (t1) {};
    \path (choice-true-2)  +(180:1) node [vertex] (t2) {};
    \path (choice-false-1) +(0:1)   node [vertex] (f1) {};
    \path (choice-false-2) +(0:1)   node [vertex] (f2) {};
    \path (choice-in-1)    +(270:1) node [vertex] (i1) {};
    \path (choice-in-2)    +(270:1) node [vertex] (i2) {};
    \path [external edge,<-] (choice-true-1)  -- (t1);
    \path [external edge,<-] (choice-true-2)  -- (t2);
    \path [external edge,<-] (choice-false-1) -- (f1);
    \path [external edge,<-] (choice-false-2) -- (f2);
    \path [external edge,->] (choice-in-1)    -- (i1);
    \path [external edge,->] (choice-in-2)    -- (i2);
    \path
    (choice-split-4) +(90:0.5)  coordinate (top)
    (choice-true-1)  +(180:0.5) coordinate (left)
    (choice-false-1) +(0:0.5)   coordinate (right)
    (choice-in-1)    +(270:0.5) coordinate (bot);
    \begin{scope}[on background layer]
      \node [region,fit=(top) (bot) (left) (right)] {};
    \end{scope}
  \end{tikzpicture}}\\[\bigskipamount]
  \subcaptionbox{\label{fig:choice-widget-false}}{\begin{tikzpicture}
    \pic (choice) {choice false};
    \path (choice-true-1)  +(180:1) node [vertex] (t1) {};
    \path (choice-true-2)  +(180:1) node [vertex] (t2) {};
    \path (choice-false-1) +(0:1)   node [vertex] (f1) {};
    \path (choice-false-2) +(0:1)   node [vertex] (f2) {};
    \path (choice-in-1)    +(270:1) node [vertex] (i1) {};
    \path (choice-in-2)    +(270:1) node [vertex] (i2) {};
    \path [external edge,<-] (choice-true-1)  -- (t1);
    \path [external edge,<-] (choice-true-2)  -- (t2);
    \path [external edge,<-] (choice-false-1) -- (f1);
    \path [external edge,<-] (choice-false-2) -- (f2);
    \path [external edge,->] (choice-in-1)    -- (i1);
    \path [external edge,->] (choice-in-2)    -- (i2);
    \path
    (choice-split-4) +(90:0.5)  coordinate (top)
    (choice-true-1)  +(180:0.5) coordinate (left)
    (choice-false-1) +(0:0.5)   coordinate (right)
    (choice-in-1)    +(270:0.5) coordinate (bot);
    \begin{scope}[on background layer]
      \node [region,fit=(top) (bot) (left) (right)] {};
    \end{scope}
  \end{tikzpicture}}
  \caption{(a) A choice widget.  \nz{The widget consists of the vertices and
  edges in the shaded region.  The dashed edges represent wire widgets joining
  its inputs and outputs to other widgets.}  (b) Its \true-BPO.  (c) Its
  \false-BPO.  In figures (b) and (c), \nz{arcs are smooth lines; edges are
  zigzag lines.}}
\end{figure}
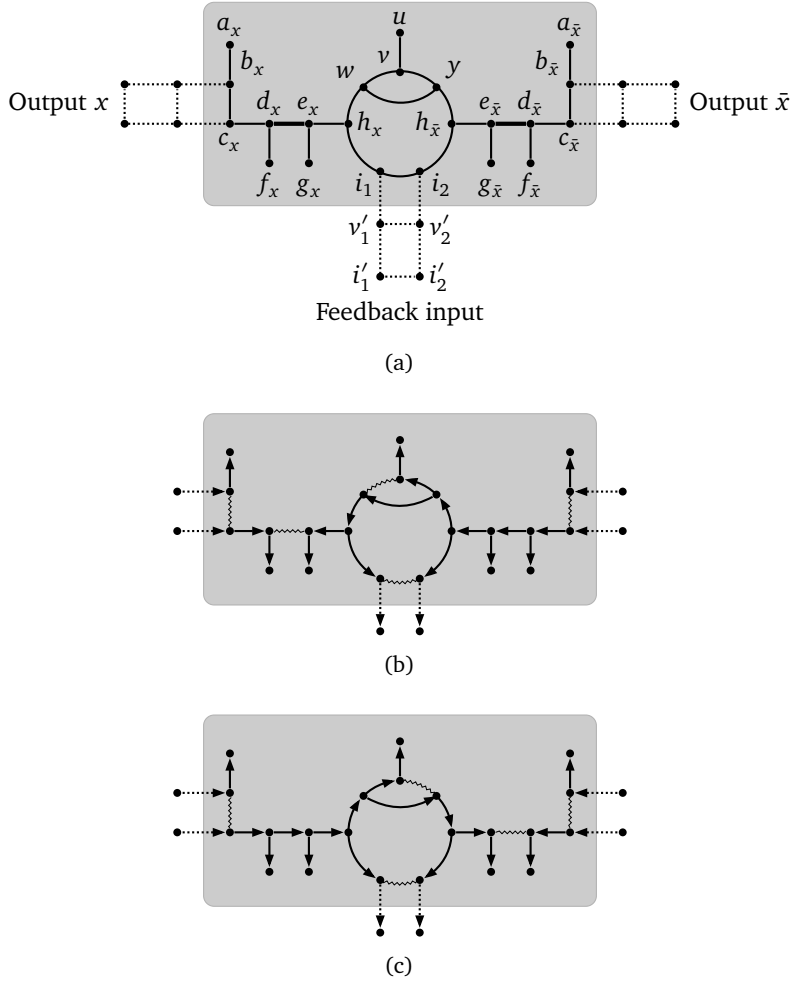

\begin{lem}
  \label{lem:choice-widget}
  Let $N$ be an unrooted binary network that contains a root widget $R$ and a
  choice widget $V$ whose input is joined to the output of another widget using
  some \nz{wire} widget $L$.  Let $v_1'$ and $v_2'$ be the two interior vertices
  of $L$, and let $\e{i_1', i_2'}$ be the input of $L$ \nz{(shown in
  \cref{fig:choice-widget})}.  Then any HPO $\vec{N}$ of $N$ contains at least
  one of the two \nz{arcs} $\a{d_x, e_x}$ and $\a{d_{\bar x}, e_{\bar x}}$ and
  must contain the two \nz{arcs} $\a{v_1', i_1'}$ and $\a{v_2', i_2'}$.
  Moreover, there exist semi-directed paths from $c_x$ to $i_1'$ and $i_2'$ or
  from $c_{\bar x}$ to $i_1'$ and $i_2'$ in $\vec{N}$.
\end{lem}

\begin{proof}
  \nz{Note that all vertices of $L$ must satisfy
  \cref{prop:in-degree,prop:out-degree} because the root of $\vec{N}$ is a
  vertex of $R$ not contained in any other widget, by \cref{lem:root-in-sheep}.}

  Consider the triangle $T = \s{v, w, y}$ in $V$.  By
  \cref{obs:undirected-matching,lem:no-semi-directed-paths}, there must exist a
  vertex $x \in T$ such that $\vec{N}$ \nz{contains the arcs $\a{x, y}$ and $\a{x,
  z}$, and the edge $\e{y, z}$,} where $y$ and $z$ are the other two vertices in
  $T$.  If $x = v$, \nz{then no matter the orientation of the edge $\e{u, v}$,
  at least one of $u$ and $v$ has no in-arcs in $\vec{N}$, contradicting
  \cref{prop:in-degree}.}  Thus, $x \ne v$.  Since the choice widget is
  symmetric, we can assume w.l.o.g.\ that \nz{$x = w$}, that is,
  $\vec{N}$~\nz{contains the arcs $\a{w, v}$ and $\a{w, y}$, and the edge $\e{v,
  y}$.}

  This implies that $\a{h_x, w}$ \nz{is an arc of} $\vec{N}$, by
  \cref{prop:in-degree}, and that $\a{y, h_{\bar x}}$ \nz{is an arc of}
  $\vec{N}$, by \cref{prop:out-degree}.  \nz{The latter implies that $\a{i_2,
  h_{\bar x}}$ is not an arc of} $\vec{N}$, by \cref{prop:in-degree}.

  If $\e{i_2, h_{\bar x}}$ \nz{is an edge of} $\vec{N}$, then by
  \cref{prop:in-degree,prop:out-degree}, $\vec{N}$ \nz{either contains the arcs
  $\a{v_2', i_2}$ and $\a{i_2, i_1}$, or the arcs $\a{i_1, i_2}$ and $\a{i_2,
  v_2'}$.}

  In the former case, $\vec{N}$ contains the semi-directed path $\s{h_x, w, y,
  h_{\bar x}, i_2, i_1}$, and $\a{i_2, i_1}$ is \nz{an arc}.  By
  \cref{lem:no-semi-directed-paths}, no such path can exist.  Thus, $\a{i_1,
  i_2}$ and $\a{i_2, v_2'}$ \nz{are arcs of} $\vec{N}$.

  This implies that $\vec{N}$ \nz{does not contain the arc $\a{v_1', i_1}$ nor
  the edge $\e{v_1', i_1}$,} because this would make \nz{$\s{v_1', i_1, i_2,
  v_2'}$} a semi-directed path in $\vec{N}$ \nz{that ends in an arc,
  contradicting \cref{lem:no-semi-directed-paths} again.}  Thus, $\a{i_1, v_1'}$
  \nz{is an arc of} $\vec{N}$.

  By \cref{prop:in-degree}, this implies that $\e{v_1', v_2'}$ \nz{is an edge
  of} $\vec{N}$ and, therefore, by \cref{prop:out-degree}, $\a{v_1', i_1'}$ and
  $\a{v_2', i_2'}$ \nz{are arcs of} $\vec{N}$.  \nz{By \cref{prop:in-degree},
  $\a{h_x, i_1}$ is an arc of $\vec{N}$.  Thus, we have semi-directed paths}
  $\s{h_x, i_1, v_1', i_1'}$ and $\s{h_x, i_1, v_1', v_2', i_2'}$ from $h_x$ to
  $i_1'$ and $i_2'$.

  This followed from the assumption that $\e{i_2, h_{\bar x}}$ \nz{is an edge
  of} $\vec{N}$.  Now consider the second case, when $\a{h_{\bar x}, i_2}$
  \nz{is an arc of} $\vec{N}$.  \nz{This implies that $\vec{N}$ contains neither
  the edge $\e{i_1, h_x}$ nor the arc $\a{i_1, h_x}$} because in either case, we
  would obtain a semi-directed path $\s{i_1, h_x, w, y, h_{\bar x}, i_2}$ in
  $\vec{N}$ \nz{that ends in an arc, contradicting}
  \cref{lem:no-semi-directed-paths}.  Thus, $\a{h_x, i_1}$ \nz{is an arc of}
  $\vec{N}$.

  Since $\vec{N}$ \nz{contains both} $\a{h_x, i_1}$ and $\a{h_{\bar x}, i_2}$,
  \nz{it must also contain the edge} $\e{i_1, i_2}$, by \cref{prop:in-degree},
  and therefore, \nz{the arcs} $\a{i_1, v_1'}$ and $\a{i_2, v_2'}$, by
  \cref{prop:out-degree}.  
  By \cref{prop:in-degree}, this implies
  that $\vec{N}$ \nz{contains the edge} $\e{v_1', v_2'}$ and therefore, by
  \cref{prop:out-degree}, \nz{the arcs} $\a{v_1', i_1'}$ and $\a{v_2', i_2'}$.
  Moreover, $\s{h_x, i_1, v_1', i_1'}$ and $\s{h_x, i_1, v_1', v_2', i_2'}$ are
  semi-directed paths from $h_x$ to $i_1'$ and $i_2'$ also in this case.

  We have shown that $\a{h_x, w}$, $\a{h_x, i_1}$, $\a{v_1', i_1'}$, and
  $\a{v_2', i_2'}$ \nz{are arcs of} $\vec{N}$ and that $\vec{N}$ contains the two
  semi-directed paths $\s{h_x, i_1, v_1', i_1'}$ and $\s{h_x, i_1, v_1', v_2',
  i_2'}$ no matter whether $\vec{N}$ \nz{contains the edge} $\e{h_{\bar x},
  i_2}$ or \nz{the arc} $\a{h_{\bar x}, i_2}$.

  \nz{By \cref{prop:in-degree}, $\a{e_x, g_x}$ and $\a{d_x, f_x}$ must be arcs
  of $\vec{N}$.}  Thus, since $\a{h_x, w}$ and $\a{h_x, i_1}$ \nz{are arcs of}
  $\vec{N}$, \nz{the same property implies that} $\a{e_x, h_x}$, $\a{d_x, e_x}$,
  and $\a{c_x, d_x}$ \nz{are also arcs of} $\vec{N}$.  \nz{This makes} $\s{c_x,
  d_x, e_x, h_x}$ a directed path in $\vec{N}$.  Together with the two
  semi-directed paths $\s{h_x, i_1, v_1', i_1'}$ and $\s{h_x, i_1, v_1', v_2',
  i_2'}$, we obtain semi-directed paths from $c_x$ to both $i_1'$ and $i_2'$
  in~$\vec{N}$.
  \yuki{At the start of the proof, we set~$x=w$ when setting up the orientation of the triangle. 
  A symmetric analysis using~$x=y$ will yield the result that there is a semi-directed path from~$c_{\bar x}$ to $i_1'$ and $i_2'$ in $\vec{N}$.}
\end{proof}

There are two BPOs of the choice widget that represent whether we set the
variable $x$ to {\true} or {\false}.  Accordingly, we call these partial
orientations the \emph{\true-BPO} and the \emph{\false-BPO} of the choice
widget.  They are shown in
\cref{fig:choice-widget-true,fig:choice-widget-false}.  The following
observation is easily verified by inspection of
\nz{these two figures.}

\begin{obs}
  \label{obs:choice-widget-boolean}
  Let $N$ be an unrooted binary network that contains a choice widget $V$, and
  let $\vec{N}$ be a PO of $N$ whose restriction to $V$ is one of the two BPOs
  of $V$ and such that the input and the two outputs of $V$ are {\true}.  Then
  all vertices in $V$ satisfy \cref{prop:in-degree,prop:out-degree}, any
  semi-directed cycle in $\vec{N}$ must include a vertex not in $V$, and the
  root of $\vec{N}$ does not belong to $V$.
\end{obs}

\subsubsection{And-Widget}

An and-widget is shown in \cref{fig:basic-and-widget}.  This widget has two
inputs representing two literals $\alpha$ and $\beta$, and an output
representing their conjunction $\alpha \wedge \beta$.  We call such a 2-input
and-widget a \emph{basic and-widget} to distinguish it from the following more
general type of and-widget with $k \ge 2$ inputs: If $k = 2$, an and-widget with
$k$ inputs is just a basic and-widget.  If $k > 2$, we first construct an
and-widget with $k - 1$ inputs and then add a basic and-widget whose output is
joined to one of the inputs of the $(k-1)$-input and-widget using a \nz{wire}
widget.  An and-widget with $5$ inputs is shown in \cref{fig:and-widget}.

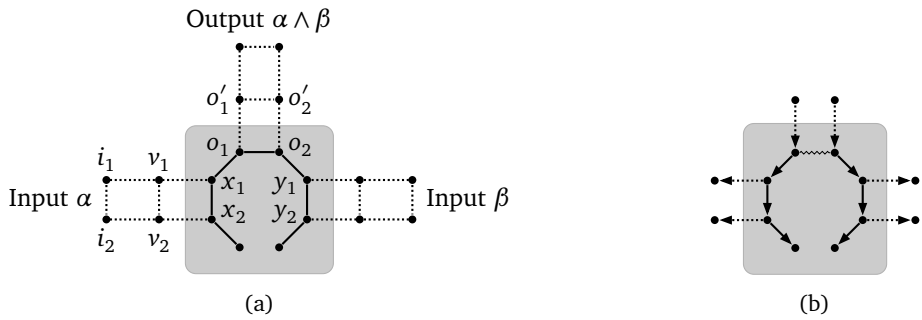
\begin{figure}[b]
  \centering
  \subcaptionbox{\label{fig:basic-and-widget}}{\begin{tikzpicture}
    \pic (and) {and};
    \path [external edge]
    (and-in-1-1) -- ++(180:1) node [vertex] (l1) {} -- +(180:1) node [vertex] (il1) {}
    (and-in-1-2) -- ++(180:1) node [vertex] (l2) {} -- +(180:1) node [vertex] (il2) {}
    (and-in-2-1) -- ++(0:1)   node [vertex] (r1) {} -- +(0:1)   node [vertex] (ir1) {}
    (and-in-2-2) -- ++(0:1)   node [vertex] (r2) {} -- +(0:1)   node [vertex] (ir2) {}
    (and-out-1)  -- ++(90:1)  node [vertex] (o1) {} -- +(90:1)  node [vertex] (oo1) {}
    (and-out-2)  -- ++(90:1)  node [vertex] (o2) {} -- +(90:1)  node [vertex] (oo2) {}
    (l1)  -- (l2)
    (il1) -- (il2)
    (r1)  -- (r2)
    (ir1) -- (ir2)
    (o1)  -- (o2)
    (oo1) -- (oo2)
    (barycentric cs:and-in-1-1=0.5,l1=0.5) coordinate (left)
    (barycentric cs:and-in-2-1=0.5,r1=0.5) coordinate (right)
    (barycentric cs:and-out-1=0.5,o1=0.5) coordinate (top)
    (and-in-1-v) +(270:0.5) coordinate (bot);
    \path
    (barycentric cs:il1=0.5,il2=0.5) node [anchor=east,xshift=-2pt] {Input $\alpha$}
    (barycentric cs:ir1=0.5,ir2=0.5) node [anchor=west,xshift=2pt]  {Input $\beta$}
    (barycentric cs:oo1=0.5,oo2=0.5) node [anchor=south,yshift=2pt] {Output $\alpha \wedge \beta$}
    (and-in-1-1)                     node [anchor=west,yshift=-2pt] {$x_1$}
    (and-in-1-2)                     node [anchor=west,yshift=2pt]  {$x_2$}
    (and-in-2-1)                     node [anchor=east,yshift=-2pt] {$y_1$}
    (and-in-2-2)                     node [anchor=east,yshift=2pt]  {$y_2$}
    (and-out-1)                      node [anchor=east,yshift=2pt]  {$o_1$}
    (and-out-2)                      node [anchor=west,yshift=2pt]  {$o_2$}
    (o1)                             node [anchor=east]             {$o_1'$}
    (o2)                             node [anchor=west]             {$o_2'$}
    (l1)                             node [anchor=south]            {$v_1$}
    (l2)                             node [anchor=north]            {\vphantom{$i_2$}$v_2$}
    (il1)                            node [anchor=south]            {$i_1$}
    (il2)                            node [anchor=north]            {$i_2$};
    \begin{scope}[on background layer]
      \node [region,fit=(left) (right) (top) (bot)] {};
    \end{scope}
  \end{tikzpicture}}%
  \hspace{1in}%
  \subcaptionbox{\label{fig:basic-and-widget-labelling}}{\begin{tikzpicture}
    \pic (and) {and boolean};
    \path
    (and-in-1-1) +(180:1)                  node [vertex] (l1) {}
    (and-in-1-2) +(180:1)                  node [vertex] (l2) {}
    (and-in-2-1) +(0:1)                    node [vertex] (r1) {}
    (and-in-2-2) +(0:1)                    node [vertex] (r2) {}
    (and-out-1)  +(90:1)                   node [vertex] (o1) {}
    (and-out-2)  +(90:1)                   node [vertex] (o2) {}
    (barycentric cs:and-in-1-1=0.6,l1=0.5) coordinate (left)
    (barycentric cs:and-in-2-1=0.6,r1=0.5) coordinate (right)
    (barycentric cs:and-out-1=0.4,o1=0.5)  coordinate (top)
    (and-in-1-v) +(270:0.5)                coordinate (bot);
    \path [external edge,->] (and-in-1-1) -- (l1);
    \path [external edge,->] (and-in-1-2) -- (l2);
    \path [external edge,->] (and-in-2-1) -- (r1);
    \path [external edge,->] (and-in-2-2) -- (r2);
    \path [external edge,<-] (and-out-1)  -- (o1);
    \path [external edge,<-] (and-out-2)  -- (o2);
    \begin{scope}[on background layer]
      \node [region,fit=(left) (right) (top) (bot)] {};
    \end{scope}
  \end{tikzpicture}}
  \caption{(a) A basic and-widget.  \nz{The widget consists of the vertices and
  edges in the shaded region.  The dashed edges represent wire widgets joining
  its inputs and outputs to other widgets.}  (b) Its BPO.  \nz{Arcs are smooth
  lines; edges are zigzag lines.}}
\end{figure}
\begin{figure}[b]
  \centering
  \subcaptionbox{\label{fig:and-widget}}{\begin{tikzpicture}[scale=0.5,transform shape]
    \path let \n1={0.75},
              \n2={\n1/2},
              \n3={\n2/sin(22.5)},
              \n4={\n3*cos(22.5)} in
    (0,0)                       pic (and-1) {and}
    ++(45:2*\n4+0.75)  ++(2,1)  pic (and-2) {and}
    ++(315:2*\n4+0.75) ++(2,-1) pic (and-3) {and}
    ++(315:2*\n4+0.75) ++(2,-1) pic (and-4) {and}
    (and-2-in-1-1 -| {barycentric cs:and-2-in-1-1=0.5,and-1-out-2=0.5}) node [vertex] (conn-1-2-1) {}
    (and-2-in-1-2 -| {barycentric cs:and-2-in-1-1=0.5,and-1-out-2=0.5}) node [vertex] (conn-1-2-2) {}
    (and-2-in-2-1 -| {barycentric cs:and-2-in-2-1=0.5,and-3-out-1=0.5}) node [vertex] (conn-2-3-1) {}
    (and-2-in-2-2 -| {barycentric cs:and-2-in-2-1=0.5,and-3-out-1=0.5}) node [vertex] (conn-2-3-2) {}
    (and-3-in-2-1 -| {barycentric cs:and-3-in-2-1=0.5,and-4-out-1=0.5}) node [vertex] (conn-3-4-1) {}
    (and-3-in-2-2 -| {barycentric cs:and-3-in-2-1=0.5,and-4-out-1=0.5}) node [vertex] (conn-3-4-2) {};
    \path
    (and-1-out-1 |- and-2-in-1-1) +(270:1.5)  coordinate (corner-1-2-1)
    (and-1-out-2 |- and-2-in-1-2) +(270:0.75) coordinate (corner-1-2-2)
    (and-3-out-2 |- and-2-in-2-1) +(270:1.5)  coordinate (corner-2-3-1)
    (and-3-out-1 |- and-2-in-2-2) +(270:0.75) coordinate (corner-2-3-2)
    (and-4-out-2 |- and-3-in-2-1) +(270:1.5)  coordinate (corner-3-4-1)
    (and-4-out-1 |- and-3-in-2-2) +(270:0.75) coordinate (corner-3-4-2)
    (and-4-in-1-v) +(270:0.5)  coordinate (bot)
    (and-2-out-1)  +(90:0.5)   coordinate (top)
    (and-1-in-1-1) +(180:2.75) coordinate (left)
    (and-4-in-2-1) +(0:2.75)   coordinate (right);
    \path [edge]
    (and-1-out-1)  -- (corner-1-2-1) arc [radius=1.5,start angle=180,end angle=90]  -- (and-2-in-1-1)
    (and-1-out-2)  -- (corner-1-2-2) arc [radius=0.75,start angle=180,end angle=90] -- (and-2-in-1-2)
    (and-3-out-1)  -- (corner-2-3-2) arc [radius=0.75,start angle=0,end angle=90]   -- (and-2-in-2-2)
    (and-3-out-2)  -- (corner-2-3-1) arc [radius=1.5,start angle=0,end angle=90]    -- (and-2-in-2-1)
    (and-4-out-1)  -- (corner-3-4-2) arc [radius=0.75,start angle=0,end angle=90]   -- (and-3-in-2-2)
    (and-4-out-2)  -- (corner-3-4-1) arc [radius=1.5,start angle=0,end angle=90]    -- (and-3-in-2-1)
    (conn-1-2-1)   -- (conn-1-2-2)
    (conn-2-3-1)   -- (conn-2-3-2)
    (conn-3-4-1)   -- (conn-3-4-2);
    \path [external edge]
    (and-2-out-1)  -- ++(90:1) node [vertex] (o1) {} -- +(90:1) node [vertex] (oo1) {}
    (and-2-out-2)  -- ++(90:1) node [vertex] (o2) {} -- +(90:1) node [vertex] (oo2) {}
    (and-4-in-1-1) -- +(180:0.75) arc [radius=1.5,start angle=90,end angle=180]  to (\tikztostart |- bot) -- ++(270:0.5) node [vertex] (d1) {} -- +(270:1) node [vertex] (id1) {}
    (and-4-in-1-2) -- +(180:0.75) arc [radius=0.75,start angle=90,end angle=180] to (\tikztostart |- bot) -- ++(270:0.5) node [vertex] (d2) {} -- +(270:1) node [vertex] (id2) {}
    (and-4-in-2-1) -- +(0:0.75)   arc [radius=1.5,start angle=90,end angle=0]    to (\tikztostart |- bot) -- ++(270:0.5) node [vertex] (e1) {} -- +(270:1) node [vertex] (ie1) {}
    (and-4-in-2-2) -- +(0:0.75)   arc [radius=0.75,start angle=90,end angle=0]   to (\tikztostart |- bot) -- ++(270:0.5) node [vertex] (e2) {} -- +(270:1) node [vertex] (ie2) {}
    (and-1-in-1-1) -- +(180:0.75) arc [radius=1.5,start angle=90,end angle=180]  to (\tikztostart |- bot) -- ++(270:0.5) node [vertex] (a1) {} -- +(270:1) node [vertex] (ia1) {}
    (and-1-in-1-2) -- +(180:0.75) arc [radius=0.75,start angle=90,end angle=180] to (\tikztostart |- bot) -- ++(270:0.5) node [vertex] (a2) {} -- +(270:1) node [vertex] (ia2) {}
    (and-1-in-2-1) -- +(0:0.75)   arc [radius=1.5,start angle=90,end angle=0]    to (\tikztostart |- bot) -- ++(270:0.5) node [vertex] (b1) {} -- +(270:1) node [vertex] (ib1) {}
    (and-1-in-2-2) -- +(0:0.75)   arc [radius=0.75,start angle=90,end angle=0]   to (\tikztostart |- bot) -- ++(270:0.5) node [vertex] (b2) {} -- +(270:1) node [vertex] (ib2) {}
    (and-3-in-1-1) -- +(180:0.75) arc [radius=1.5,start angle=90,end angle=180]  to (\tikztostart |- bot) -- ++(270:0.5) node [vertex] (c1) {} -- +(270:1) node [vertex] (ic1) {}
    (and-3-in-1-2) -- +(180:0.75) arc [radius=0.75,start angle=90,end angle=180] to (\tikztostart |- bot) -- ++(270:0.5) node [vertex] (c2) {} -- +(270:1) node [vertex] (ic2) {}
    foreach \v in {a,b,c,d,e} {
      (\v1)  -- (\v2)
      (i\v1) -- (i\v2)
    }
    (o1)  -- (o2)
    (oo1) -- (oo2);
    \path
    (barycentric cs:ia1=0.5,ia2=0.5) node [transform shape=false,anchor=north,yshift=-2pt] {\vphantom{$\beta$}$\alpha$}
    (barycentric cs:ib1=0.5,ib2=0.5) node [transform shape=false,anchor=north,yshift=-2pt] {\vphantom{$\beta$}$\beta$}
    (barycentric cs:ic1=0.5,ic2=0.5) node [transform shape=false,anchor=north,yshift=-2pt] {\vphantom{$\beta$}$\gamma$}
    (barycentric cs:id1=0.5,id2=0.5) node [transform shape=false,anchor=north,yshift=-2pt] {\vphantom{$\beta$}$\delta$}
    (barycentric cs:ie1=0.5,ie2=0.5) node [transform shape=false,anchor=north,yshift=-2pt] {\vphantom{$\beta$}$\varepsilon$}
    (barycentric cs:oo1=0.5,oo2=0.5) node [transform shape=false,anchor=south,yshift=2pt]  {\vphantom{$\beta$}$\alpha \wedge \beta \wedge \gamma \wedge \delta \wedge \varepsilon$};
    \begin{scope}[on background layer]
      \node [region,fit=(top) (left) (right) (bot)] {};
    \end{scope}
  \end{tikzpicture}}%
  \hspace{1in}%
  \subcaptionbox{\label{fig:and-widget-labelling}}{\begin{tikzpicture}[scale=0.5,transform shape,>={Triangle[width=0pt 3 0,length=0pt 3 0]}]
    \path let \n1={0.75},
              \n2={\n1/2},
              \n3={\n2/sin(22.5)},
              \n4={\n3*cos(22.5)} in
    (0,0)                       pic (and-1) {and boolean}
    ++(45:2*\n4+0.75)  ++(2,1)  pic (and-2) {and boolean}
    ++(315:2*\n4+0.75) ++(2,-1) pic (and-3) {and boolean}
    ++(315:2*\n4+0.75) ++(2,-1) pic (and-4) {and boolean}
    (and-2-in-1-1 -| {barycentric cs:and-2-in-1-1=0.5,and-1-out-2=0.5}) node [vertex] (conn-1-2-1) {}
    (and-2-in-1-2 -| {barycentric cs:and-2-in-1-1=0.5,and-1-out-2=0.5}) node [vertex] (conn-1-2-2) {}
    (and-2-in-2-1 -| {barycentric cs:and-2-in-2-1=0.5,and-3-out-1=0.5}) node [vertex] (conn-2-3-1) {}
    (and-2-in-2-2 -| {barycentric cs:and-2-in-2-1=0.5,and-3-out-1=0.5}) node [vertex] (conn-2-3-2) {}
    (and-3-in-2-1 -| {barycentric cs:and-3-in-2-1=0.5,and-4-out-1=0.5}) node [vertex] (conn-3-4-1) {}
    (and-3-in-2-2 -| {barycentric cs:and-3-in-2-1=0.5,and-4-out-1=0.5}) node [vertex] (conn-3-4-2) {};
    \path
    (and-1-out-1 |- and-2-in-1-1) +(270:1.5)  coordinate (corner-1-2-1)
    (and-1-out-2 |- and-2-in-1-2) +(270:0.75) coordinate (corner-1-2-2)
    (and-3-out-2 |- and-2-in-2-1) +(270:1.5)  coordinate (corner-2-3-1)
    (and-3-out-1 |- and-2-in-2-2) +(270:0.75) coordinate (corner-2-3-2)
    (and-4-out-2 |- and-3-in-2-1) +(270:1.5)  coordinate (corner-3-4-1)
    (and-4-out-1 |- and-3-in-2-2) +(270:0.75) coordinate (corner-3-4-2)
    (and-4-in-1-v)                +(270:0.5)  coordinate (bot) +(270:1) coordinate (end)
    (and-2-out-1)  +(90:0.5)                  coordinate (top)
    (and-1-in-1-1) +(180:2.75)                coordinate (left)
    (and-4-in-2-1) +(0:2.75)                  coordinate (right);
    \path [edge,<-]          (and-4-out-1) -- (corner-3-4-2) arc [radius=0.75,start angle=0,end angle=90]    -- (conn-3-4-2);
    \path [edge,<-]          (and-4-out-2) -- (corner-3-4-1) arc [radius=1.5,start angle=0,end angle=90]     -- (conn-3-4-1);
    \path [edge,<-]          (and-1-out-1)  -- (corner-1-2-1) arc [radius=1.5,start angle=180,end angle=90]  -- (conn-1-2-1);
    \path [edge,<-]          (and-1-out-2)  -- (corner-1-2-2) arc [radius=0.75,start angle=180,end angle=90] -- (conn-1-2-2);
    \path [edge,<-]          (and-3-out-1)  -- (corner-2-3-2) arc [radius=0.75,start angle=0,end angle=90]   -- (conn-2-3-2);
    \path [edge,<-]          (and-3-out-2)  -- (corner-2-3-1) arc [radius=1.5,start angle=0,end angle=90]    -- (conn-2-3-1);
    \path [vertical]         (and-2-in-1-1) -- (conn-1-2-1);
    \path [vertical]         (and-2-in-1-2) -- (conn-1-2-2);
    \path [vertical]         (and-2-in-2-1) -- (conn-2-3-1);
    \path [vertical]         (and-2-in-2-2) -- (conn-2-3-2);
    \path [vertical]         (and-3-in-2-1) -- (conn-3-4-1);
    \path [vertical]         (and-3-in-2-2) -- (conn-3-4-2);
    \path [external edge,<-] (and-2-out-1) -- +(90:1) coordinate (o1);
    \path [external edge,<-] (and-2-out-2) -- +(90:1) coordinate (o2);
    \path [external edge,->] (and-4-in-1-1) -- +(180:0.75) arc [radius=1.5,start angle=90,end angle=180]  to (\tikztostart |- end) coordinate (d1);
    \path [external edge,->] (and-4-in-1-2) -- +(180:0.75) arc [radius=0.75,start angle=90,end angle=180] to (\tikztostart |- end) coordinate (d2);
    \path [external edge,->] (and-4-in-2-1) -- +(0:0.75)   arc [radius=1.5,start angle=90,end angle=0]    to (\tikztostart |- end) coordinate (e1);
    \path [external edge,->] (and-4-in-2-2) -- +(0:0.75)   arc [radius=0.75,start angle=90,end angle=0]   to (\tikztostart |- end) coordinate (e2);
    \path [external edge,->] (and-1-in-1-1) -- +(180:0.75) arc [radius=1.5,start angle=90,end angle=180]  to (\tikztostart |- end) coordinate (a1);
    \path [external edge,->] (and-1-in-1-2) -- +(180:0.75) arc [radius=0.75,start angle=90,end angle=180] to (\tikztostart |- end) coordinate (a2);
    \path [external edge,->] (and-1-in-2-1) -- +(0:0.75)   arc [radius=1.5,start angle=90,end angle=0]    to (\tikztostart |- end) coordinate (b1);
    \path [external edge,->] (and-1-in-2-2) -- +(0:0.75)   arc [radius=0.75,start angle=90,end angle=0]   to (\tikztostart |- end) coordinate (b2);
    \path [external edge,->] (and-3-in-1-1) -- +(180:0.75) arc [radius=1.5,start angle=90,end angle=180]  to (\tikztostart |- end) coordinate (c1);
    \path [external edge,->] (and-3-in-1-2) -- +(180:0.75) arc [radius=0.75,start angle=90,end angle=180] to (\tikztostart |- end) coordinate (c2);
    \path [horizontal]
    (conn-1-2-1) -- (conn-1-2-2)
    (conn-2-3-1) -- (conn-2-3-2)
    (conn-3-4-1) -- (conn-3-4-2);
    \path
    (barycentric cs:a1=0.5,a2=0.5) +(270:1) node [transform shape=false,anchor=north,yshift=-2pt] {\phantom{$\beta$}}
    (barycentric cs:b1=0.5,b2=0.5) +(270:1) node [transform shape=false,anchor=north,yshift=-2pt] {\phantom{$\beta$}}
    (barycentric cs:c1=0.5,c2=0.5) +(270:1) node [transform shape=false,anchor=north,yshift=-2pt] {\phantom{$\beta$}}
    (barycentric cs:d1=0.5,d2=0.5) +(270:1) node [transform shape=false,anchor=north,yshift=-2pt] {\phantom{$\beta$}}
    (barycentric cs:e1=0.5,e2=0.5) +(270:1) node [transform shape=false,anchor=north,yshift=-2pt] {\phantom{$\beta$}};
    \begin{scope}[on background layer]
      \node [region,fit=(top) (left) (right) (bot)] {};
    \end{scope}
  \end{tikzpicture}}
  \caption{(a) An and-widget with 5 inputs. \nz{The widget consists of the
  vertices and edges in the shaded region.  The dashed edges represent wire
  widgets joining its inputs and outputs to other widgets.} (b) Its BPO.
  \nz{Arcs are smooth lines; edges are zigzag lines.}}
\end{figure}
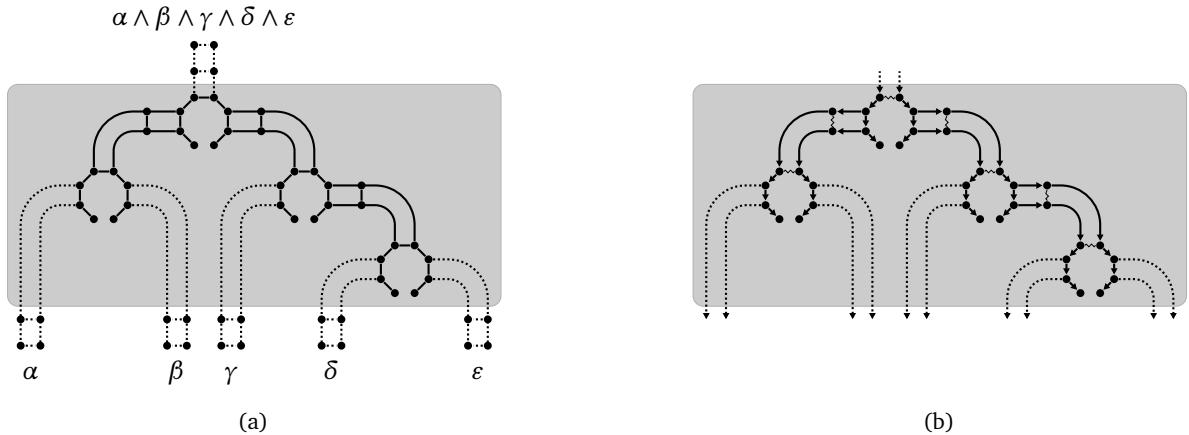

As the following lemma shows, the output of an and-widget can be {\true} only if
it all its inputs are {\true}.

\begin{lem}
  \label{lem:and-widget}
  Let $N$ be an unrooted binary network that contains \nz{a root widget $R$ and}
  an and-widget $A$, assume that every input of $A$ is joined to the output of
  another widget using a \nz{wire} widget, and let $\vec{N}$ be an HPO of $N$
  that makes the output \nz{$\e{o_1, o_2}$} of $A$ {\true}.  Consider the
  \nz{wire} widget $L$ used to join any input of $A$ to the output of another
  widget, let $v_1$ and $v_2$ be the two interior vertices of $L$, and let
  $\e{i_1, i_2}$ be the input of $L$.  Then $\a{v_1, i_1}$ and $\a{v_2, i_2}$
  \nz{are arcs of} $\vec{N}$.  Moreover, there exist semi-directed paths from
  $o_1$ and $o_2$ to $i_1$ and $i_2$ in $\vec{N}$.
\end{lem}

\begin{proof}
  \nz{Note that all vertices of $O$ must satisfy
  \cref{prop:in-degree,prop:out-degree} because the root of $\vec{N}$ is a
  vertex of $R$ not contained in any other widget, by \cref{lem:root-in-sheep}.}
  The proof is by induction on the number $k$ of inputs of $A$.

  If $k = 2$, then $A$ is a basic and-widget.  Assume w.l.o.g.\ that $L$ is the
  \nz{wire} widget used to join some widget to the input $\e{x_1, x_2}$ of $A$
  (see \cref{fig:basic-and-widget}).
  
  Let $o_1'$ and $o_2'$ be the two external neighbours of $o_1$ and $o_2$,
  respectively.  Since $\e{o_1, o_2}$ is {\true}, $\a{o_1', o_1}$ and $\a{o_2',
  o_2}$ \nz{are arcs of} $\vec{N}$.  \nz{Therefore, by \cref{prop:in-degree},
  $\e{o_1, o_2}$ is an edge of $\vec{N}$ and, by \cref{prop:out-degree},
  $\a{o_1, x_1}$ is an arc of $\vec{N}$.}  Now note that $\e{o_1, x_1}$ is one
  of the \nz{external edges connected to the} output of $L$.  Thus, by
  \cref{lem:connector-widget}, $\a{v_1, i_1}$ and $\a{v_2, i_2}$ \nz{must be
  arcs of} $\vec{N}$, and there must exist semi-directed paths from $x_1$ to
  $i_1$ and $i_2$ in $\vec{N}$.  Together with the semi-directed paths $\s{o_1,
  x_1}$ and $\s{o_2, o_1, x_1}$, we obtain semi-directed paths from $o_1$ and
  $o_2$ to $i_1$ and $i_2$ in~$\vec{N}$.

  If $k > 2$, then $A$ is composed of an and-widget $A_1$ with $k - 1$ inputs, a
  basic and-widget $A_2$, and a \nz{wire} widget $L'$ joining the output of
  $A_2$ to one of the inputs of $A_1$.  The edge $\e{o_1, o_2}$ is the output of
  $A_1$.  If \nz{$L \ne L'$ and} $L$ is a \nz{wire} widget used to join one of
  the inputs of $A_1$ to the output of another widget, then by the induction
  hypothesis, $\a{v_1, i_1}$ and $\a{v_2, i_2}$ \nz{are arcs of} $\vec{N}$ and
  there exist semi-directed paths from $o_1$ and $o_2$ to $i_1$ and $i_2$ in
  $\vec{N}$.

  \nz{Now let $v_1'$ and $v_2'$ be the two interior vertices of $L'$, and let
  $\e{i_1', i_2'}$ be the input of~$L'$. Then,} since $L'$ is itself a \nz{wire}
  widget used to join the output of $A_2$ to an input of $A_1$, $\a{v_1', i_1'}$
  and $\a{v_2', i_2'}$ \nz{are arcs of} $\vec{N}$, and there exist semi-directed
  paths from $o_1$ and $o_2$ to $i_1'$ and $i_2'$ in $\vec{N}$.
  
  Since \nz{$\e{i_1', i_2'}$ is the output of $A_2$ and} $\e{v_1', i_1'}$ and
  $\e{v_2', i_2'}$ are \nz{the two external edges connected to the output
  of}~$A_2$, the base case applied to $A_2$ now shows that $\a{v_1, i_1}$ and
  $\a{v_2, i_2}$ \nz{are arcs of} $\vec{N}$ and that $\vec{N}$ contains
  semi-directed paths from $i_1'$ and $i_2'$ to $i_1$ and $i_2$ also for any
  \nz{wire} widget $L$ used to join one of the inputs of $A_2$ to the output of
  another widget.  Together with the semi-directed paths from $o_1$ and $o_2$ to
  $i_1'$ and $i_2'$, we obtain semi-directed paths from $o_1$ and $o_2$ to $i_1$
  and $i_2$ in $\vec{N}$.
\end{proof}

The BPO of a basic and-widget is shown in \cref{fig:basic-and-widget-labelling}.
To obtain the BPO of a $k$-input and-widget, we combine the BPOs of the basic
and-widgets and the \true-BPOs of the \nz{wire} widgets of which the and-widget
is composed.  The BPO of the 5-input and-widget in \cref{fig:and-widget} is
shown in \cref{fig:and-widget-labelling}.  The following observation is easily
verified by inspection of \cref{fig:and-widget-labelling}.

\begin{obs}
  \label{obs:and-widget-boolean}
  Let $N$ be a Boolean network that contains an and-widget $A$, let $\vec{N}$ be
  a PO of $N$ whose restriction to $A$ is the BPO of $A$ and which makes all
  inputs and outputs of $A$ {\true}.  Then all vertices in $A$ satisfy
  \cref{prop:in-degree,prop:out-degree}, any semi-directed cycle in $\vec{N}$
  must include a vertex not in $A$, and the root of $\vec{N}$ does not belong to
  $A$.
\end{obs}

\subsubsection{Replicator Widget}

A replicator widget is shown in \cref{fig:basic-replicator-widget}.  This widget
has one input representing a Boolean value $\lambda$, and two outputs
representing copies of $\lambda$.  We call such a 2-output replicator widget a
\emph{basic replicator widget} to distinguish it from the following more general
type of replicator widget with $k \ge 2$ outputs: If $k = 2$, a replicator
widget with $k$ outputs is just a basic replicator widget.  If $k > 2$, we first
construct a replicator widget with $k - 1$ outputs and then add a basic
replicator widget whose input is joined to one of the outputs of the
$(k-1)$-output replicator widget using a \nz{wire} widget.  A replicator widget
with $5$ outputs is shown in \cref{fig:replicator-widget}.

\begin{figure}[p]
  \centering
  \subcaptionbox{\label{fig:basic-replicator-widget}}{
}\\[\bigskipamount]
  \subcaptionbox{\label{fig:basic-replicator-widget-true-labelling}}{%
}%
  \hspace{1in}%
  \subcaptionbox{\label{fig:basic-replicator-widget-false-labelling}}{%
}%
  \hspace{1in}%
  \subcaptionbox{\label{fig:basic-replicator-widget-escape-labelling}}{%
}
  \caption{(a) A basic replicator widget.  \nz{The widget consists of the
  vertices and edges in the shaded region.  The dashed edges represent wire
  widgets joining its inputs and outputs to other widgets.}  (b) Its \true-BPO.
  (c) Its \false-BPO.  (d) Its escape BPO.  In figures (b)--(d), \nz{arcs are
  smooth lines; edges are zigzag lines}.}
\end{figure}
\begin{figure}[p]
  \centering
  \subcaptionbox{\label{fig:replicator-widget}}{%
}\\[\bigskipamount]
  \subcaptionbox{\label{fig:replicator-widget-true-labelling}}{%
}%
  \hspace{1in}%
  \subcaptionbox{\label{fig:replicator-widget-false-labelling}}{%
}
  \caption{(a) A replicator widget with 5 outputs. \nz{The widget consists of
  the vertices and edges in the shaded region.  The dashed edges represent wire
  widgets joining its inputs and outputs to other widgets.}  (b) Its \true-BPO.
  (c) Its \false-BPO.  \nz{In figures (b) and (c), arcs are smooth lines; edges
  are zigzag lines.}}
\end{figure}

As the following lemma shows, an output of a replicator widget can be {\true}
only if the input of the replicator widget is {\true}.

\begin{lem}
  \label{lem:replicator-widget}
  Let $N$ be an unrooted binary network that contains \nz{a root widget $R$ and}
  a replicator widget $M$, assume that the input of $M$ is joined to the output
  of another widget using a \nz{wire} widget $L$, let $v_1$ and $v_2$ be the two
  interior vertices of $L$, let \nz{$\e{i_1', i_2'}$ be the input} of $L$, and
  let $\vec{N}$ be an HPO of $N$.  Assume that there exists an output $\{o_1,
  o_2\}$ of $M$ that is made {\true} by~$\vec{N}$.  Then $\a{v_1, i_1'}$ and
  $\a{v_2, i_2'}$ \nz{are arcs of} $\vec{N}$ and there exist semi-directed paths
  from $o_1$ and $o_2$ to $i_1'$ and $i_2'$ in $\vec{N}$.
\end{lem}

\begin{proof}
  \nz{Note that all vertices of $L$ must satisfy
  \cref{prop:in-degree,prop:out-degree} because the root of $\vec{N}$ is a
  vertex of $R$ not contained in any other widget, by \cref{lem:root-in-sheep}.}
  The proof is by induction on the number $k$ of outputs of $M$.

  If $k = 2$, then $M$ is a basic replicator widget.  Assume w.l.o.g.\ that
  $\e{o_1, o_2} = \e{x_1, x_2}$.  Let $x_1'$ and $x_2'$ be the two external
  neighbours of $x_1$ and $x_2$, respectively \nz{(see
  \cref{fig:basic-replicator-widget}).}
  
  Since $\e{x_1, x_2}$ is {\true}, $\a{x_1',x_1}$ and $\a{x_2', x_2}$ \nz{are
  arcs of} $\vec{N}$.  \nz{Therefore, by \cref{prop:in-degree}, $\e{x_1, x_2}$
  is an edge of $\vec{N}$ and, by \cref{prop:out-degree}, $\a{x_2, i_1}$ is an
  arc of $\vec{N}$.}  Since \nz{$\e{i_1, i_2}$ is the output of $L$ and}
  $\e{x_2, i_1}$ is one of the \nz{external edges connected to the output of}
  $L$, \nz{\cref{lem:connector-widget} shows that} $\a{v_1, i_1'}$ and $\a{v_2,
  i_2'}$ \nz{are arcs of} $\vec{N}$ and that $\vec{N}$ contains semi-directed
  paths from $i_1$ to $i_1'$ and $i_2'$.  Together with the the semi-directed
  paths $\s{x_2, i_1}$ and $\s{x_1, x_2, i_1}$, we obtain semi-directed paths
  from $x_1$ and $x_2$ to $i_1'$ and $i_2'$.

  If $k > 2$, then $M$ is composed of a replicator widget $M_1$ with $k - 1$
  outputs, a basic replicator widget $M_2$, and a \nz{wire} widget $L'$ joining
  one of the outputs of $M_1$ to the input of $M_2$.  The input of $M$ is the
  input of $M_1$.

  If $\e{o_1, o_2}$ is an output of $M_1$, then by the induction hypothesis,
  $\a{v_1, i_1'}$ and $\a{v_2, i_2'}$ \nz{are arcs of} $\vec{N}$ and there exist
  semi-directed paths from $o_1$ and $o_2$ to $i_1'$ and $i_2'$.
  
  If $\e{o_1, o_2}$ is an output of $M_2$, then \nz{let $v\dprime_1$ and
  $v\dprime_2$ be the two interior vertices of $L'$, and let $\e{i\dprime_1,
  i_2\dprime}$ be the input of $L'$.}  By applying the base case to $M_2$, we
  conclude that $\a{v\dprime_1, i\dprime_1}$ and $\a{v\dprime_2, i\dprime_2}$
  \nz{are arcs of} $\vec{N}$ and that there exist semi-directed paths from $o_1$
  and $o_2$ to $i\dprime_1$ and $i\dprime_2$, 

  Since $\e{i_1\dprime, i_2\dprime}$ is an output of $M_1$ that is made {\true}
  by the two \nz{arcs} $\a{v_1\dprime, i_1\dprime}$ and $\a{v_2\dprime,
  i_2\dprime}$, the induction hypothesis applied to $M_2$ shows that $\a{v_1,
  i_1'}$ and $\a{v_2, i_2'}$ \nz{are arcs of} $\vec{N}$ and that there exist
  semi-directed paths from $i_1\dprime$ and $i_2\dprime$ to $i_1'$ and $i_2'$
  in~$\vec{N}$.  Together with the semi-directed paths from $o_1$ and $o_2$ to
  $i_1\dprime$ and $i_2\dprime$, we obtain semi-directed paths from $o_1$ and
  $o_2$ to $i_1'$ and $i_2'$ in $\vec{N}$.
\end{proof}

For the basic replicator widget, we need three different BPOs.  First, when the
input value $\lambda$ is {\true}, we want both outputs to be {\true} as well.
We call this the \emph{\true-BPO} of the basic replicator widget, shown in
\cref{fig:basic-replicator-widget-true-labelling}.  The second BPO, the
\emph{\false-BPO} shown in \cref{fig:basic-replicator-widget-false-labelling},
is used to replicate a {\false} input value $\lambda$.  The final BPO is used to
output a {\true} and a {\false} value when the input is {\true}.  We call this
the \emph{escape BPO} of the basic replicator widget, as it will be used to
escape from the subgraph of $N(F)$ formed by choice \nz{widgets}, clause
widgets, and the \nz{wire and replicator widgets connecting} them, and to link
up with the root in the root widget.  This BPO is shown in
\cref{fig:basic-replicator-widget-escape-labelling}.

For the escape BPO, we have the following observation, which is easily verified
by inspection of \cref{fig:basic-replicator-widget-escape-labelling}.

\begin{obs}
  \label{obs:replicator-widget-escape}
  Let $N$ be an unrooted binary network that contains a basic replicator widget
  $M$, and let $\vec{N}$ be a PO of $N$ whose restriction to $M$ is the escape
  BPO of $M$.  Assume further that the input of $M$ is made {\true} by
  $\vec{N}$, that the output $\e{o_1, o_2}$ \nz{that is an edge of} $\vec{N}$ is
  made {\true} by $\vec{N}$, and that the output $\e{o_3, o_4}$ \nz{that is an
  arc of} $\vec{N}$ is made {\false} by $\vec{N}$.  Then all vertices in $M$
  satisfy \cref{prop:in-degree,prop:out-degree}, any semi-directed cycle in
  $\vec{N}$ must include a vertex not in $M$, and the root of $\vec{N}$ does not
  belong to $M$.
\end{obs}

We use the {\true}- and {\false}-BPOs of the basic replicator widget to
construct {\true}- and {\false}-BPOs of larger replicator widgets.  In
particular, we obtain the {\true}-BPO of a replicator widget $M$ by combining
the {\true}-BPOs of all basic replicator widgets and \nz{wire} widgets of which
$M$ is composed.  We obtain the {\false}-BPO of $M$ by combining the
{\false}-BPOs of all basic replicator widgets and \nz{wire} widgets of which $M$
is composed.  These BPOs of the $5$-output replicator widget in
\cref{fig:replicator-widget} are shown in
\cref{fig:replicator-widget-true-labelling,fig:replicator-widget-false-labelling}.
The following two observations are easily verified by inspection of
\cref{fig:replicator-widget-true-labelling,fig:replicator-widget-false-labelling}.

\begin{obs}
  \label{obs:replicator-widget-true}
  Let $N$ be \nz{an unrooted binary} network that contains a replicator widget
  $M$, and let $\vec{N}$ be a PO whose restriction to $M$ is the {\true}-BPO of
  $M$.  Assume further that the input of $M$ and all outputs of $M$ are made
  {\true} by $\vec{N}$.  Then all vertices in $M$ satisfy
  \cref{prop:in-degree,prop:out-degree}, any semi-directed cycle in $\vec{N}$
  must include a vertex not in $M$, and the root of $\vec{N}$ does not belong to
  $M$.
\end{obs}

\begin{obs}
  \label{obs:replicator-widget-false}
  Let $N$ be \nz{an unrooted binary} network that contains a replicator widget
  $M$, and let $\vec{N}$ be a PO whose restriction to $M$ is the {\false}-BPO of
  $M$.  Assume further that the input of $M$ and all outputs of $M$ are made
  {\false} by $\vec{N}$.  Then all vertices in $M$ satisfy
  \cref{prop:in-degree,prop:out-degree}, any semi-directed cycle in $\vec{N}$
  must include a vertex not in $M$, and the root of $\vec{N}$ does not belong to
  $M$.
\end{obs}

\subsubsection{Or-Widget}

We build a clause widget by stacking two or-widgets.  \Cref{fig:or-widget} shows
an or-widget.  As the following lemma shows, the output of an or-widget is
{\true} only if at least one of its inputs is {\true}.  Thus, the output of an
or-widget is essentially the disjunction of the two inputs.

\begin{figure}[ht]
  \centering
  \begin{tikzpicture}
    \pic (or) {or};
    \path
    (or-in-1-1)                          node [anchor=south east]              (x1-l) {$x_1$}
    (or-in-1-2)                          node [anchor=north east]              (x2-l) {$x_2$}
    (or-in-2-1)                          node [anchor=south west]              (y1-l) {$y_1$}
    (or-in-2-2)                          node [anchor=north west]              (y2-l) {$y_2$}
    (or-left-cross)                      node [anchor=south east]              (u2-l) {$u_2$}
    (or-right-cross)                     node [anchor=south west]              (v2-l) {$v_2$}
    (or-out-1)                           node [anchor=south east,xshift=0.5mm] (o1-l) {$o_1$}
    (or-out-2)                           node [anchor=south west]              (o2-l) {$o_2$}
    (or-left-ear-1)                      node [anchor=north west]                     {$u_1$}
    (or-right-ear-1)                     node [anchor=north east,xshift=1mm]          {$v_1$}
    (or-left-foot-1)                     node [anchor=north west]                     {$w_1$}
    (or-right-foot-1)                    node [anchor=north east]                     {$z_1$}
    (or-left-foot-2)                     node [anchor=north]                   (w2-l) {$w_2$}
    (or-right-foot-2)                    node [anchor=north]                   (z2-l) {$z_2$}
    (or-in-1-1 -| x2-l.west) ++(180:0.5) node [vertex]                         (l1)   {} +(180:1) node [vertex] (il1) {}
    (or-in-1-2 -| x2-l.west) ++(180:0.5) node [vertex]                         (l2)   {} +(180:1) node [vertex] (il2) {}
    (or-in-2-1 -| y2-l.east) ++(0:0.5)   node [vertex]                         (r1)   {} +(0:1)   node [vertex] (ir1) {}
    (or-in-2-2 -| y2-l.east) ++(0:0.5)   node [vertex]                         (r2)   {} +(0:1)   node [vertex] (ir2) {}
    (or-out-1 |- o2-l.north) ++(90:0.5)  node [vertex]                         (o1)   {} +(90:1)  node [vertex] (oo1) {}
    (or-out-2 |- o2-l.north) ++(90:0.5)  node [vertex]                         (o2)   {} +(90:1)  node [vertex] (oo2) {}
    (barycentric cs:il1=0.5,il2=0.5)     node [anchor=east,xshift=-4pt]               {Input $\alpha$}
    (barycentric cs:ir1=0.5,ir2=0.5)     node [anchor=west,xshift=4pt]                {Input $\beta$}
    (barycentric cs:oo1=0.5,oo2=0.5)     node [anchor=south,yshift=6pt]               {Output $\alpha \vee \beta$}
    (o1)                                 node [anchor=east]                           {$o_1'$}
    (o2)                                 node [anchor=west]                           {$o_2'$}
    (oo1)                                node [anchor=east]                           {$o_1\dprime$}
    (oo2)                                node [anchor=west]                           {$o_2\dprime$}
    (l1)                                 node [anchor=south]                          {$x_1'$}
    (l2)                                 node [anchor=north]                          {$x_2'$}
    (il1)                                node [anchor=south]                          {$x_1\dprime$}
    (il2)                                node [anchor=north]                          {$x_2\dprime$}
    (r1)                                 node [anchor=south]                          {$y_1'$}
    (r2)                                 node [anchor=north]                          {$y_2'$}
    (ir1)                                node [anchor=south]                          {$y_1\dprime$}
    (ir2)                                node [anchor=north]                          {$y_2\dprime$};
    \path [external edge]
    (or-in-1-1) -- (l1) -- (il1)
    (or-in-1-2) -- (l2) -- (il2)
    (or-in-2-1) -- (r1) -- (ir1)
    (or-in-2-2) -- (r2) -- (ir2)
    (or-out-1)  -- (o1) -- (oo1)
    (or-out-2)  -- (o2) -- (oo2)
    (l1)        -- (l2)
    (r1)        -- (r2)
    (o1)        -- (o2);
    \begin{scope}[on background layer]
      \node [region,fit=(x1-l) (x2-l) (y1-l) (y2-l) (u2-l) (v2-l) (o1-l) (o2-l) (w2-l) (z2-l)] {};
    \end{scope}
  \end{tikzpicture}
  \caption{An or-widget.  \nz{The widget consists of the vertices and edges in
  the shaded region.  Dashed edges are connections to other widgets.  The
  endpoints of these edges may be part of a wire widget or of another or-widget.
  In particular, the edges $\e{x_1\dprime, x_2\dprime}$, $\e{y_1\dprime,
  y_2\dprime}$, and $\e{o_1\dprime, o_2\dprime}$ are not drawn because they may
  not exist.}}
  \label{fig:or-widget}
\end{figure}
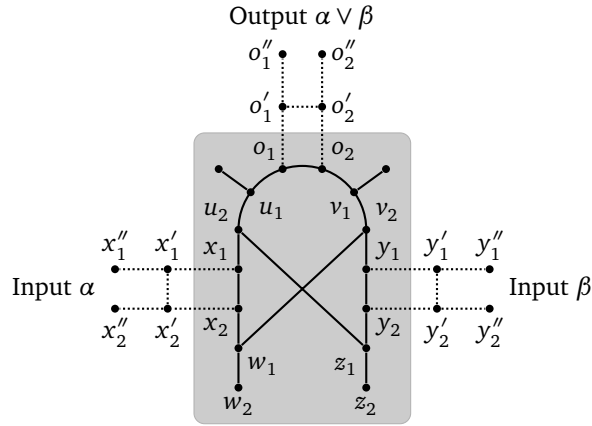

\begin{lem}
  \label{lem:or-widget}
  \nz{Let $N$ be an unrooted binary network $N$ that contains \nz{a root widget
  $R$ and} an or-widget $O$, let $\e{x_1, x_2}$ and $\e{y_1, y_2}$ be the two
  inputs of $O$, let $x_1'$, $x_2'$, $y_1'$, and $y_2'$ be the external
  neighbours of $x_1$, $x_2$, $y_1$, and $y_2$, respectively, assume that $N$
  contains the edges $\e{x_1', x_2'}$ and $\e{y_1', y_2'}$, and let
  $x_1\dprime$, $x_2\dprime$, $y_1\dprime$, and $y_2\dprime$ be the neighbours
  of $x_1'$, $x_2'$, $y_1'$, and $y_2'$, respectively, that are not in $\{x_1,
  x_2, y_1, y_2, x_1', x_2', y_1', y_2'\}$ (see \cref{fig:or-widget}).} If
  $\vec{N}$ is an HPO that makes the \nz{output} $\e{o_1, o_2}$ of $O$ {\true},
  then $\a{x_1', x_1\dprime}$ and $\a{x_2', x_2\dprime}$ \nz{are arcs of}
  $\vec{N}$, or $\a{y_1', y_1\dprime}$ and $\a{y_2', y_2\dprime}$ \nz{are arcs
  of} $\vec{N}$.  In the former case, there exist semi-directed paths from $o_1$
  and $o_2$ to $x_1\dprime$ and $x_2\dprime$ in $\vec{N}$.  In the latter case,
  there exist semi-directed paths from $o_1$ and $o_2$ to $y_1\dprime$ and
  $y_2\dprime$ in $\vec{N}$.
\end{lem}

\begin{proof}
  \nz{Note that all vertices of $L$ must satisfy
  \cref{prop:in-degree,prop:out-degree} because the root of $\vec{N}$ is a
  vertex of $R$ not contained in any other widget, by \cref{lem:root-in-sheep}.}

  Since the output $\e{o_1, o_2}$ of $O$ is {\true}, $\a{o_1', o_1}$ and
  $\a{o_2', o_2}$ \nz{are arcs of} $\vec{N}$, where $o_1'$ and $o_2'$ are the
  two external neighbours of $o_1$ and $o_2$, respectively.  By
  \cref{prop:in-degree}, this implies that $\e{o_1, o_2}$ \nz{is an edge of}
  $\vec{N}$ and, thus, by \cref{prop:out-degree}, that $\a{o_1, u_1}$ and
  $\a{o_2, v_1}$ \nz{are arcs of} $\vec{N}$.  By \cref{prop:in-degree}, this
  implies that neither $\a{u_2, u_1}$ nor $\a{v_2, v_1}$ \nz{is an arc of}
  $\vec{N}$.  Therefore, $\vec{N}$ \nz{contains the arc $\a{u_1, u_2}$ or the
  edge $\e{u_1, u_2}$, and the arc $\a{v_1, v_2}$ or the edge $\e{v_1, v_2}$.}
  Whichever is the case, we have semi-directed paths from $o_1$ and $o_2$ to
  $u_2$ and $v_2$.

  Now consider the chordless cycle $C = \s{u_2, x_1, x_2, w_1, v_2, y_1, y_2,
  z_1}$.  By \cref{lem:out-edge-pair}, there must exist two consecutive vertices
  $a$ and $b$ in $C$ such that \nz{$\vec{N}$ contains the edge $\e{a, b}$ and
  the arcs $\a{a, a'}$ and $\a{b, b'}$,} where $a'$ and $b'$ are the two
  neighbours of $a$ and $b$, \nz{respectively,} that do not belong to $C$.
  Since \nz{$\vec{N}$ contains neither $\a{u_2, u_1}$ nor $\a{v_2, v_1}$,} we
  must have $\{a, b\} = \{x_1, x_2\}$, $\{a, b\} = \{y_1, y_2\}$, $\{a, b\} =
  \{x_2, w_1\}$ or $\{a, b\} = \{y_2, z_1\}$.

  If $\{a, b\} = \{x_1, x_2\}$,  then $\a{x_1, x_1'}$ and $\a{x_2, x_2'}$
  \nz{are arcs of} $\vec{N}$.

  If $\{a, b\} = \{x_2, w_1\}$, then $\e{x_2, w_1}$ \nz{is an edge and} $\a{x_2,
  x_2'}$ \nz{is an arc of} $\vec{N}$.  By \cref{prop:in-degree}, this implies
  that $\a{x_1, x_2}$ \nz{is also an arc of} $\vec{N}$.  If $\e{x_1', x_1}$
  \nz{is an edge or} $\a{x_1', x_1}$ \nz{is an arc of} $\vec{N}$, \nz{then} we
  obtain a semi-directed path $\s{x_1', x_1, x_2, x_2'}$ in $\vec{N}$ \nz{that
  ends in an arc, contradicting} \cref{lem:no-semi-directed-paths}.  Thus,
  $\a{x_1, x_1'}$ \nz{is an arc of} $\vec{N}$, that is, once again, $\a{x_1,
  x_1'}$ and $\a{x_2, x_2'}$ \nz{are arcs of} $\vec{N}$.

  \nz{Since} $\a{x_1, x_1'}$ and $\a{x_2, x_2'}$ \nz{are arcs of} $\vec{N}$,
  $\e{x_1', x_2'}$ \nz{is an edge of} $\vec{N}$, \nz{by \cref{prop:in-degree},}
  and therefore, by \cref{prop:out-degree}, $\a{x_1', x_1\dprime}$ and $\a{x_2',
  x_2\dprime}$ \nz{are arcs of} $\vec{N}$.
  
  This proves that $\a{x_1', x_1\dprime}$ and $\a{x_2', x_2\dprime}$ \nz{are
  arcs of} $\vec{N}$ if $\{a, b\} = \{x_1, x_2\}$ or $\{a, b\} = \{x_2, w_1\}$.
  An analogous argument shows that $\a{y_1', y_1\dprime}$ and $\a{y_2',
  y_2\dprime}$ \nz{are arcs of} $\vec{N}$ if $\{a, b\} = \{y_1, y_2\}$ or $\{a,
  b\} = \{y_2, z_1\}$.
  
  Now assume that $\a{o_1', o_1}$, $\a{o_2', o_2}$, $\a{x_1', x_1\dprime}$, and
  $\a{x_2', x_2\dprime}$ \nz{are arcs of} $\vec{N}$.  We already proved that we
  have semi-directed paths from $o_1$ and $o_2$ to $u_2$ and $v_2$ in this case.

  \nz{$\vec{N}$ must contain at least one of the arcs} $\a{u_2, x_1}$ and
  $\a{w_1, x_2}$.  Otherwise, the subgraph $H$ induced by $x_1$, $x_2$, $x_1'$,
  and $x_2'$ \nz{would have no in-arcs.  Since every vertex in $H$ has in-degree
  $1$, by \cref{prop:in-degree}, this would imply that $H$ contains a directed
  cycle, contradicting \cref{prop:cycle-free}.}

  By an analogous argument \nz{applied to the subgraph induced by $x_1'$ and
  $x_2'$}, at least one of $\a{x_1, x_1'}$ and $\a{x_2, x_2'}$ \nz{must be an
  arc of} $\vec{N}$.

  If $\a{u_2, x_1}$ \nz{is an arc of} $\vec{N}$, then $\a{x_2, x_1}$ \nz{is not
  an arc of} $\vec{N}$\nz{, by \cref{prop:in-degree}.}  Thus, $\a{x_1, x_2}$
  \nz{is an arc or} $\e{x_1, x_2}$ \nz{is an edge of} $\vec{N}$.  In either
  case, we have semi-directed paths from $u_2$ to $x_1$ and $x_2$ and, thus,
  semi-directed paths from $o_1$ and $o_2$ to $x_1$ and $x_2$ in $\vec{N}$.

  If $\a{w_1, x_2}$ \nz{is an arc of} $\vec{N}$, an analogous argument gives us
  semi-directed paths from $w_1$ to $x_1$ and $x_2$.  Also note that by
  \cref{prop:in-degree}, $\a{w_1, w_2}$ \nz{must be an arc of} $\vec{N}$.  Thus,
  by \cref{prop:in-degree} again, $\a{v_2, w_1}$ \nz{is an arc of} $\vec{N}$.
  Since we have semi-directed paths from $o_1$ and $o_2$ to $v_2$, we once again
  obtain semi-directed paths from $o_1$ and $o_2$ to $x_1$ and $x_2$ in
  $\vec{N}$.

  If $\a{x_1, x_1'}$ \nz{is an arc of} $\vec{N}$, then $\a{x_2', x_1'}$ \nz{is
  not}, by \cref{prop:in-degree}.  Thus, $\a{x_1', x_2'}$ \nz{is an arc} or
  $\e{x_1', x_2'}$ \nz{is an edge of} $\vec{N}$.  In either case, we have
  semi-directed paths from $x_1$ to $x_1'$ and $x_2'$ and, thus, since $\a{x_1',
  x_1\dprime}$ and $\a{x_2', x_2\dprime}$ \nz{are arcs of} $\vec{N}$, also
  semi-directed paths from $x_1$ to $x_1\dprime$ and $x_2\dprime$ in $\vec{N}$.

  If $\a{x_2, x_2'}$ \nz{is an arc of} $\vec{N}$, then an analogous argument
  shows that we have semi-directed paths from $x_2$ to $x_1\dprime$ and
  $x_2\dprime$ in~$\vec{N}$.

  Since we have already shown that there exist semi-directed paths from $o_1$
  and $o_2$ to both $x_1$ and $x_2$, we conclude that whether $\a{x_1, x_1'}$ or
  $\a{x_2, x_2'}$ \nz{is an arc of} $\vec{N}$, there exist semi-directed paths
  from $o_1$ and $o_2$ to $x_1\dprime$ and $x_2\dprime$ in $\vec{N}$.

  Analogous arguments prove that there exist semi-directed paths from $o_1$ and
  $o_2$ to $y_1\dprime$ and $y_2\dprime$ if $\a{o_1', o_1}$, $\a{o_2', o_2}$,
  $\a{y_1', y_1\dprime}$, and $\a{y_2', y_2\dprime}$ \nz{are arcs of} $\vec{N}$.
\end{proof}

\subsubsection{Clause Widget}

We obtain a clause widget from two or-widgets by adding edges $\e{u_1, v_1}$ and
$\e{u_2, v_2}$, where $\e{u_1, u_2}$ is one of the inputs of the first
or-widget, and $\e{v_1, v_2}$ is the output of the second or-widget.  This is
shown in \cref{fig:clause-widget}.  The purpose of a clause widget is to
generalize the or-widget so it takes three inputs $\alpha$, $\beta$, and
$\gamma$, and its output represents their disjunction $\alpha \vee \beta \vee
\gamma$.  As the following lemma shows, the clause widget has the desired
property that for its output to be {\true}, at least one of the inputs must be
{\true}.

\begin{figure}[b]
  \centering
  \begin{tikzpicture}
    \pic (clause) {clause};
    \path
    (clause-out-1)                 node [anchor=north,xshift=2pt]                   {$o_1$}
    (clause-out-2)                 node [anchor=north,xshift=-2pt]                  {$o_2$}
    (clause-in-1-1)                node [anchor=south east]                         {$x_1$}
    (clause-in-1-2)                node [anchor=north east]                         {$x_2$}
    (clause-in-2-1)                node [anchor=south west]                         {$y_1$}
    (clause-in-2-2)                node [anchor=north west]                         {$y_2$}
    (clause-in-3-1)                node [anchor=south west]                         {$z_1$}
    (clause-in-3-2)                node [anchor=north west]                         {$z_2$}
    (clause-inter-1)               node [anchor=east]                               {$u_1$}
    (clause-inter-2)               node [anchor=west]                               {$u_2$}
    (clause-inter-3)               node [anchor=north,xshift=2pt]                   {$v_1$}
    (clause-inter-4)               node [anchor=north,xshift=-2pt]                  {$v_2$}
    (clause-bottom-left-ear-1)     node [xshift=-3pt,yshift=-1pt,anchor=north west] {$w_1$}
    (clause-bottom-right-ear-1)    node [xshift=4pt,yshift=-1pt,anchor=north east]  {$w_2$}
    (clause-bottom-left-cross)     node [anchor=south east,yshift=-3pt] (w1-l)      {$w_1'$}
    (clause-bottom-right-cross)    node [anchor=south west,yshift=-3pt]             {$w_2'$};
    \path
    (clause-top-left-ear-2)     +(90:0.5)  coordinate (top)
    (clause-top-right-ear-2)    +(0:0.5)   coordinate (right)
    (clause-bottom-left-foot-2) +(270:0.5) coordinate (bot)
    (clause-bottom-left-ear-2)  +(180:0.5) coordinate (left);
    \path
    node [rectangle,inner sep=0pt,fit=(left) (w1-l)] (find-left) {}
    (find-left.west) coordinate (left);
    \path
    (clause-in-1-1 -| left)  ++(180:0.5) node [vertex] (x1) {} +(180:1) node [vertex] (ix1) {}
    (clause-in-1-2 -| left)  ++(180:0.5) node [vertex] (x2) {} +(180:1) node [vertex] (ix2) {}
    (clause-in-2-1 -| right) ++(0:0.5)   node [vertex] (y1) {} +(0:1)   node [vertex] (iy1) {}
    (clause-in-2-2 -| right) ++(0:0.5)   node [vertex] (y2) {} +(0:1)   node [vertex] (iy2) {}
    (clause-in-3-1 -| right) ++(0:0.5)   node [vertex] (z1) {} +(0:1)   node [vertex] (iz1) {}
    (clause-in-3-2 -| right) ++(0:0.5)   node [vertex] (z2) {} +(0:1)   node [vertex] (iz2) {}
    (clause-out-1  |- top)   ++(90:0.5)  node [vertex] (o1) {} +(90:1)  node [vertex] (oo1) {}
    (clause-out-2  |- top)   ++(90:0.5)  node [vertex] (o2) {} +(90:1)  node [vertex] (oo2) {};
    \path [external edge]
    (clause-in-1-1) -- (x1) -- (ix1)
    (clause-in-1-2) -- (x2) -- (ix2)
    (clause-in-2-1) -- (y1) -- (iy1)
    (clause-in-2-2) -- (y2) -- (iy2)
    (clause-in-3-1) -- (z1) -- (iz1)
    (clause-in-3-2) -- (z2) -- (iz2)
    (clause-out-1)  -- (o1) -- (oo1)
    (clause-out-2)  -- (o2) -- (oo2)
    (x1) -- (x2) (ix1) -- (ix2)
    (y1) -- (y2) (iy1) -- (iy2)
    (z1) -- (z2) (iz1) -- (iz2)
    (o1) -- (o2) (oo1) -- (oo2);
    \path
    (x1)                             node [anchor=south]            {$v_1'$}
    (x2)                             node [anchor=north]            {$v_1\dprime$}
    (ix1)                            node [anchor=south]            {$i_1'$}
    (ix2)                            node [anchor=north]            {$i_1\dprime$}
    (y1)                             node [anchor=south]            {$v_2'$}
    (y2)                             node [anchor=north]            {$v_2\dprime$}
    (iy1)                            node [anchor=south]            {$i_2'$}
    (iy2)                            node [anchor=north]            {$i_2\dprime$}
    (z1)                             node [anchor=south]            {$v_3'$}
    (z2)                             node [anchor=north]            {$v_3\dprime$}
    (iz1)                            node [anchor=south]            {$i_3'$}
    (iz2)                            node [anchor=north]            {$i_3\dprime$}
    (o1)                             node [anchor=east]             {$o_1'$}
    (o2)                             node [anchor=west]             {$o_2'$}
    (barycentric cs:ix1=0.5,ix2=0.5) node [anchor=east,xshift=-3pt] {Input $\alpha$}
    (barycentric cs:iy1=0.5,iy2=0.5) node [anchor=west,xshift=3pt]  {Input $\beta$}
    (barycentric cs:iz1=0.5,iz2=0.5) node [anchor=west,xshift=3pt]  {Input $\gamma$}
    (barycentric cs:oo1=0.5,oo2=0.5) node [anchor=south,yshift=2pt] {Output $\alpha \vee \beta \vee \gamma$};
    \begin{scope}[on background layer]
      \node [region,fit=(top) (bot) (left) (right)] {};
    \end{scope}
  \end{tikzpicture}
  \caption{A clause widget.  \nz{The widget consists of the vertices and edges
   in the shaded region.  Dashed edges represent wire widgets joining its inputs
   and outputs to other widgets.}}
  \label{fig:clause-widget}
\end{figure}
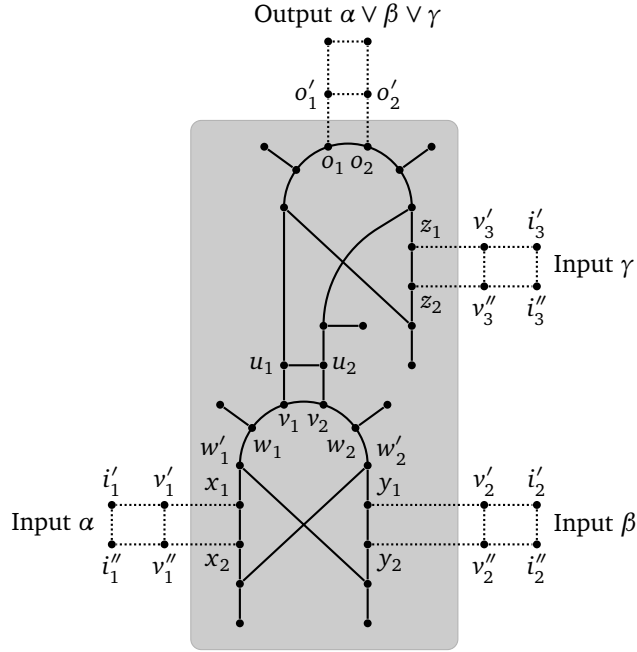
\begin{figure}[p]
  \hspace{\stretch{1}}%
  \subcaptionbox{\label{fig:clause-widget-1-labelling}}{\begin{tikzpicture}
    \pic (clause) {clause boolean-1};
    \path
    (clause-top-left-ear-2)     +(90:0.5)  coordinate (top)   +(90:1)  coordinate (top-end)
    (clause-top-right-ear-2)    +(0:0.5)   coordinate (right) +(0:1)   coordinate (right-end)
    (clause-bottom-left-foot-2) +(270:0.5) coordinate (bot)
    (clause-bottom-left-ear-2)  +(180:0.5) coordinate (left)  +(180:1) coordinate (left-end);
    \path [external edge,<-] (clause-out-1)  -- (clause-out-1  |- top-end);
    \path [external edge,<-] (clause-out-2)  -- (clause-out-2  |- top-end);
    \path [external edge,->] (clause-in-1-1) -- (clause-in-1-1 -| left-end);
    \path [external edge,->] (clause-in-1-2) -- (clause-in-1-2 -| left-end);
    \path [external edge,<-] (clause-in-2-1) -- (clause-in-2-1 -| right-end);
    \path [external edge,<-] (clause-in-2-2) -- (clause-in-2-2 -| right-end);
    \path [external edge,<-] (clause-in-3-1) -- (clause-in-3-1 -| right-end);
    \path [external edge,<-] (clause-in-3-2) -- (clause-in-3-2 -| right-end);
    \begin{scope}[on background layer]
      \node [region,fit=(top) (bot) (left) (right)] {};
    \end{scope}
  \end{tikzpicture}}%
  \hspace*{\stretch{1}}%
  \subcaptionbox{\label{fig:clause-widget-3-labelling}}{\begin{tikzpicture}
    \pic (clause) {clause boolean-3};
    \path
    (clause-top-left-ear-2)     +(90:0.5)  coordinate (top)   +(90:1)  coordinate (top-end)
    (clause-top-right-ear-2)    +(0:0.5)   coordinate (right) +(0:1)   coordinate (right-end)
    (clause-bottom-left-foot-2) +(270:0.5) coordinate (bot)
    (clause-bottom-left-ear-2)  +(180:0.5) coordinate (left)  +(180:1) coordinate (left-end);
    \path [external edge,<-] (clause-out-1)  -- (clause-out-1  |- top-end);
    \path [external edge,<-] (clause-out-2)  -- (clause-out-2  |- top-end);
    \path [external edge,<-] (clause-in-1-1) -- (clause-in-1-1 -| left-end);
    \path [external edge,<-] (clause-in-1-2) -- (clause-in-1-2 -| left-end);
    \path [external edge,<-] (clause-in-2-1) -- (clause-in-2-1 -| right-end);
    \path [external edge,<-] (clause-in-2-2) -- (clause-in-2-2 -| right-end);
    \path [external edge,->] (clause-in-3-1) -- (clause-in-3-1 -| right-end);
    \path [external edge,->] (clause-in-3-2) -- (clause-in-3-2 -| right-end);
    \begin{scope}[on background layer]
      \node [region,fit=(top) (bot) (left) (right)] {};
    \end{scope}
  \end{tikzpicture}}%
  \hspace*{\stretch{1}}%
  \\[\bigskipamount]
  \subcaptionbox{\label{fig:clause-widget-1-2-labelling}}{\begin{tikzpicture}
    \pic (clause) {clause boolean-1-2};
    \path
    (clause-top-left-ear-2)     +(90:0.5)  coordinate (top)   +(90:1)  coordinate (top-end)
    (clause-top-right-ear-2)    +(0:0.5)   coordinate (right) +(0:1)   coordinate (right-end)
    (clause-bottom-left-foot-2) +(270:0.5) coordinate (bot)
    (clause-bottom-left-ear-2)  +(180:0.5) coordinate (left)  +(180:1) coordinate (left-end);
    \path [external edge,<-] (clause-out-1)  -- (clause-out-1  |- top-end);
    \path [external edge,<-] (clause-out-2)  -- (clause-out-2  |- top-end);
    \path [external edge,->] (clause-in-1-1) -- (clause-in-1-1 -| left-end);
    \path [external edge,->] (clause-in-1-2) -- (clause-in-1-2 -| left-end);
    \path [external edge,->] (clause-in-2-1) -- (clause-in-2-1 -| right-end);
    \path [external edge,->] (clause-in-2-2) -- (clause-in-2-2 -| right-end);
    \path [external edge,<-] (clause-in-3-1) -- (clause-in-3-1 -| right-end);
    \path [external edge,<-] (clause-in-3-2) -- (clause-in-3-2 -| right-end);
    \begin{scope}[on background layer]
      \node [region,fit=(top) (bot) (left) (right)] {};
    \end{scope}
  \end{tikzpicture}}%
  \hspace{\stretch{1}}%
  \subcaptionbox{\label{fig:clause-widget-1-3-labelling}}{\begin{tikzpicture}
    \pic (clause) {clause boolean-1-3};
    \path
    (clause-top-left-ear-2)     +(90:0.5)  coordinate (top)   +(90:1)  coordinate (top-end)
    (clause-top-right-ear-2)    +(0:0.5)   coordinate (right) +(0:1)   coordinate (right-end)
    (clause-bottom-left-foot-2) +(270:0.5) coordinate (bot)
    (clause-bottom-left-ear-2)  +(180:0.5) coordinate (left)  +(180:1) coordinate (left-end);
    \path [external edge,<-] (clause-out-1)  -- (clause-out-1  |- top-end);
    \path [external edge,<-] (clause-out-2)  -- (clause-out-2  |- top-end);
    \path [external edge,->] (clause-in-1-1) -- (clause-in-1-1 -| left-end);
    \path [external edge,->] (clause-in-1-2) -- (clause-in-1-2 -| left-end);
    \path [external edge,<-] (clause-in-2-1) -- (clause-in-2-1 -| right-end);
    \path [external edge,<-] (clause-in-2-2) -- (clause-in-2-2 -| right-end);
    \path [external edge,->] (clause-in-3-1) -- (clause-in-3-1 -| right-end);
    \path [external edge,->] (clause-in-3-2) -- (clause-in-3-2 -| right-end);
    \begin{scope}[on background layer]
      \node [region,fit=(top) (bot) (left) (right)] {};
    \end{scope}
  \end{tikzpicture}}%
  \hspace{\stretch{1}}%
  \subcaptionbox{\label{fig:clause-widget-1-2-3-labelling}}{\begin{tikzpicture}
    \pic (clause) {clause boolean-1-2-3};
    \path
    (clause-top-left-ear-2)     +(90:0.5)  coordinate (top)   +(90:1)  coordinate (top-end)
    (clause-top-right-ear-2)    +(0:0.5)   coordinate (right) +(0:1)   coordinate (right-end)
    (clause-bottom-left-foot-2) +(270:0.5) coordinate (bot)
    (clause-bottom-left-ear-2)  +(180:0.5) coordinate (left)  +(180:1) coordinate (left-end);
    \path [external edge,<-] (clause-out-1)  -- (clause-out-1  |- top-end);
    \path [external edge,<-] (clause-out-2)  -- (clause-out-2  |- top-end);
    \path [external edge,->] (clause-in-1-1) -- (clause-in-1-1 -| left-end);
    \path [external edge,->] (clause-in-1-2) -- (clause-in-1-2 -| left-end);
    \path [external edge,->] (clause-in-2-1) -- (clause-in-2-1 -| right-end);
    \path [external edge,->] (clause-in-2-2) -- (clause-in-2-2 -| right-end);
    \path [external edge,->] (clause-in-3-1) -- (clause-in-3-1 -| right-end);
    \path [external edge,->] (clause-in-3-2) -- (clause-in-3-2 -| right-end);
    \begin{scope}[on background layer]
      \node [region,fit=(top) (bot) (left) (right)] {};
    \end{scope}
  \end{tikzpicture}}
  \caption{\nz{A sample of BPOs of a clause widget.  Arcs are smooth lines;
   edges are zigzag lines.} (a) 1-BPO. (b)~3-BPO. (c) 1-2-BPO. (d) 1-3-BPO.
   (e) 1-2-3-BPO.}
   \label{fig:clause-bpos}
\end{figure}
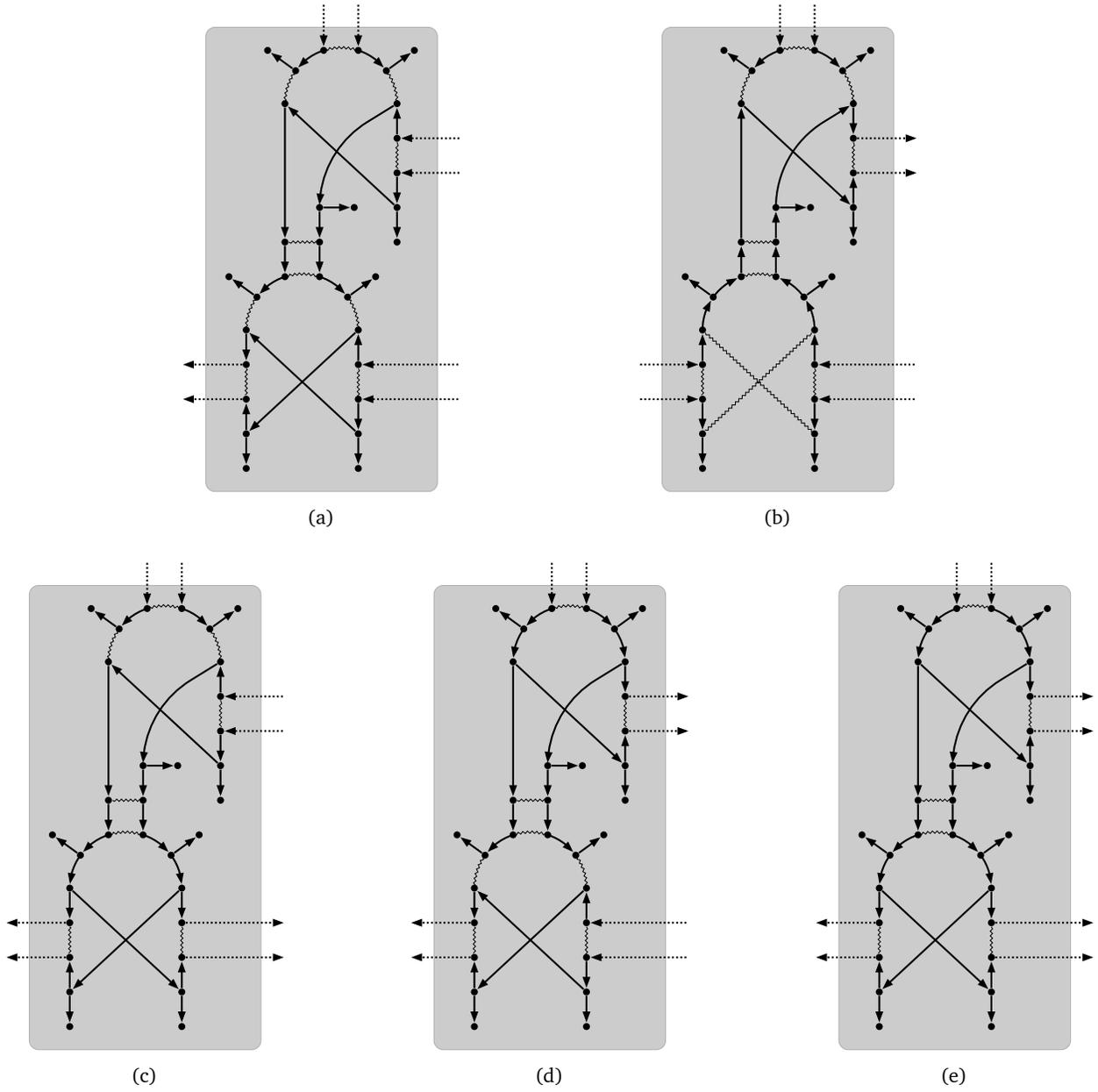

\begin{lem}
  \label{lem:clause-widget}
  Let $N$ be an unrooted binary network that contains a \nz{root widget and a}
  clause widget $C$ whose inputs $\e{x_1, x_2}$, $\e{y_1, y_2}$, and $\e{z_1,
  z_2}$ are joined to the outputs of other widgets using \nz{wire} widgets
  $L_1$, $L_2$, and $L_3$, respectively, and let $\vec{N}$ be an HPO of $N$ that
  makes the input of \nz{$C$} {\true}.  For $j \in \{1, 2, 3\}$, let $v_j'$ and
  $v_j\dprime$ be the two interior vertices of $L_j$, and \nz{let $\e{i_j',
  i_j\dprime}$ be its inputs (see \cref{fig:clause-widget}).}  Then there exists
  an index $j \in \{1, 2, 3\}$ such that $\a{v_j', i_j'}$ and $\a{v_j\dprime,
  i_j\dprime}$ \nz{are arcs of} $\vec{N}$.  If $\a{v_3', i_3'}$ and
  $\a{v_3\dprime, i_3\dprime}$ \nz{are arcs of} $\vec{N}$, then there exist
  semi-directed paths from $o_1$ and $o_2$ to $i_3'$ and $i_3\dprime$
  \nz{in~$\vec{N}$}.  If \nz{at least one of} $\a{v_3', i_3'}$ and
  $\a{v_3\dprime, i_3\dprime}$ \nz{is not an arc of} $\vec{N}$, then there exist
  semi-directed paths from $o_1$ and $o_2$ to $i_j'$ and $i_j\dprime$ in
  $\vec{N}$, for any index $j \in \{1, 2\}$ such that $\a{v_j', i_j'}$ and
  $\a{v_j\dprime, i_j\dprime}$ \nz{are arcs of} $\vec{N}$.
\end{lem}

\begin{proof}
  By applying \cref{lem:or-widget} to the top or-widget of $C$, we conclude that
  $\a{v_1, w_1}$ and $\a{v_2, w_2}$ \nz{are arcs of} $\vec{N}$ or $\a{v_3',
  i_3'}$ and $\a{v_3\dprime, i_3\dprime}$ \nz{are arcs of} $\vec{N}$.  In the
  latter case, we have found an index $j$ such that $\a{v_j', i_j'}$ and
  $\a{v_j\dprime, i_j\dprime}$ \nz{are arcs of} $\vec{N}$.

  In the former case, observe that the proof of \cref{lem:or-widget} starts by
  concluding that if $\a{o_1', o_1}$ and $\a{o_2', o_2}$ \nz{are arcs of}
  $\vec{N}$, then $\a{o_1, u_1}$ and $\a{o_2, v_1}$ \nz{are arcs of} $\vec{N}$,
  referring to the labels in \cref{fig:or-widget}.  The lemma then followed from
  this observation.
  
  The vertices $o_1$, $o_2$, $u_1$, and $v_1$ in \cref{fig:or-widget} are the
  vertices $v_1$, $v_2$, $w_1$, and $w_2$ of the bottom or-widget in
  \cref{fig:clause-widget}.  Thus, if $\a{v_1, w_1}$ and $\a{v_2, w_2}$ \nz{are
  arcs of} $\vec{N}$, the proof of \cref{lem:or-widget} shows that $\a{v_1',
  i_1'}$ and $\a{v_1\dprime, i_1\dprime}$ \nz{are arcs of} $\vec{N}$ or
  $\a{v_2', i_2'}$ and $\a{v_2\dprime, i_2\dprime}$ \nz{are arcs of} $\vec{N}$.
  Thus, we have an index $j$ \nz{such that} $\a{v_j', i_j'}$ and $\a{v_j\dprime,
  i_j\dprime}$ \nz{are arcs of} $\vec{N}$ in this case as well.

  Now assume that $\a{o_1', o_1}$, $\a{o_2', o_2}$, $\a{v_j', i_j'}$, and
  $\a{v_j\dprime, i_j\dprime}$ \nz{are arcs of} $\vec{N}$, \yuki{for some~$j\in\{1,2,3\}$}.  If $j = 3$, then
  \cref{lem:or-widget} applied to the top or-widget of $C$ shows that there
  exist semi-directed paths from $o_1$ and $o_2$ to $i_3'$ and $i_3\dprime$ in
  $\vec{N}$.

  If $j \ne 3$ and $\a{v_3', i_3'}$ or $\a{v_3\dprime, i_3\dprime}$ \nz{is not
  an arc of} $\vec{N}$, then \cref{lem:or-widget} applied to the top or-widget
  of $C$ shows that $\a{v_1, w_1}$ and $\a{v_2, w_2}$ \nz{are arcs of} $\vec{N}$
  and that there exist semi-directed paths from $o_1$ and $o_2$ to $w_1$ and
  $w_2$.  Thus, by \cref{prop:in-degree}, \nz{neither $\a{w_1', w_1}$ nor
  $\a{w_2', w_2}$ is an arc of $\vec{N}$,} which implies that there also exist
  semi-directed paths from $o_1$ and $o_2$ to $w_1'$ and $w_2'$ in $\vec{N}$.

  By the same argument as in the proof of \cref{lem:or-widget}, we also have
  semi-directed paths from $w_1'$ or $w_2'$ to $i_j'$ and $i_j\dprime$.
  Together with the semi-directed paths from $o_1$ and $o_2$ to $w_1'$ and
  $w_2'$, we obtain semi-directed paths from $o_1$ and $o_2$ to $i_j'$ and
  $i_j\dprime$ in $\vec{N}$.
\end{proof}

There are in total 7 BPOs for a clause widget, corresponding to the subset of
{\true} literals in the clause.  Accordingly, we call these labellings the
\emph{1-BPO}, \emph{2-BPO}, \emph{3-BPO}, \emph{1-2-BPO}, \emph{1-3-BPO},
\emph{2-3-BPO}, and \emph{1-2-3-BPO}, referring to the indices of the literals
in the clause that are {\true}.  The 1-BPO and 2-BPO are symmetric, as are the
1-3-BPO and 2-3-BPO.  \cref{fig:clause-bpos} shows the 1-, 3-, 1-2-, 1-3-, and
1-2-3-BPOs, respectively.  The following observation is easily verified by
inspection of \cref{fig:clause-bpos}.

\begin{obs}
  \label{obs:clause-widget-boolean}
  Let $N$ be an unrooted binary network that contains a clause widget $C$, and
  let $\vec{N}$ be a PO of $N$ whose restriction to $C$ is one of the BPOs of
  $C$.  Assume further that $\vec{N}$ make the output of $C$ {\true} and that it
  makes each input of $C$ {\true} or {\false} in accordance with the chosen BPO
  of $C$.  Then all vertices in $C$ satisfy
  \cref{prop:in-degree,prop:out-degree}, any semi-directed cycle in $\vec{N}$
  must contain a vertex not in $C$, and the root of $\vec{N}$ does not belong to
  $C$.
\end{obs}

\subsection{The Reduction}

\label{sec:reduction}

Let $F$ be a \nz{Boolean formula in} 3-CNF with $m$ clauses $C_1, \ldots, C_m$
and $n$ variables $x_1, \ldots, x_n$.  We denote the literals in clause $C_i$ as
$\lambda_{i,1}, \lambda_{i,2}, \lambda_{i,3}$.  For each literal $\lambda \in
\{x_1, \bar x_1, \ldots, x_n, \bar x_n\}$, let $k_\lambda$ be the number of
clauses that contain $\lambda$.

\begin{figure}[p]
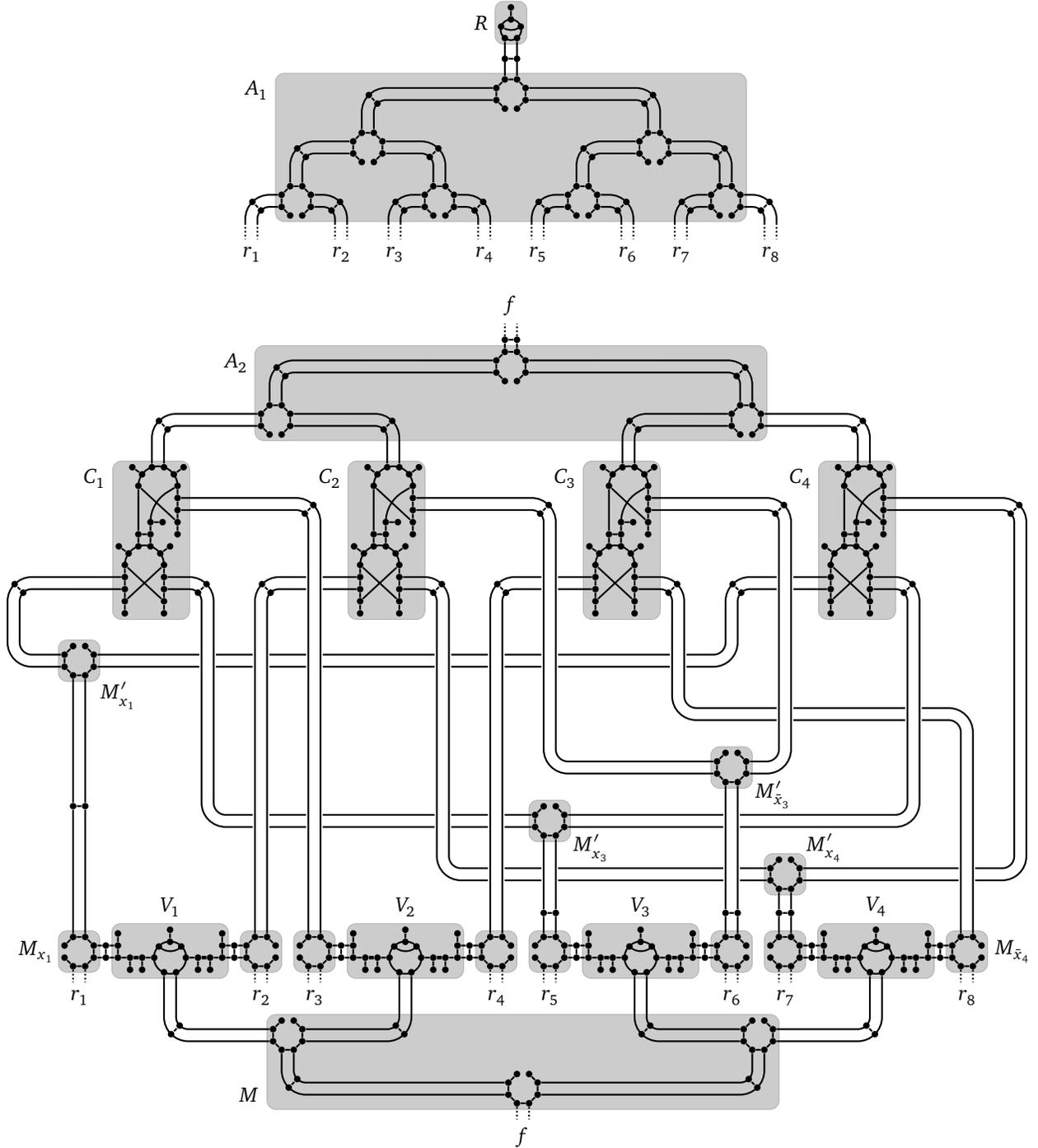

  \centering
  $F = (x_1 \vee x_3 \vee x_2) \wedge (\bar x_1 \vee x_4 \vee \bar x_3) \wedge
  (\bar x_2 \vee \bar x_4 \vee \bar x_3) \wedge (x_1 \vee x_3 \vee x_4)$
  
  \bigskip


  \caption{An illustration of the reduction from 3-SAT to \nz{\textsc{Orchard
   Orientation}.} For the sake of visual clarity, the connections between
   \nz{$A_1$ and the replicator widgets $M_\lambda$, for $\lambda \in \{x_1,
   \bar x_1, \ldots, x_4, \bar x_4\}$, are not drawn.} Instead, the outputs of
   \nz{these replicator widgets} are labelled $r_1$ through $r_8$ and should be
   understood to be connected to the inputs \nz{of $A_1$} with the same labels.
   Similarly, the output of \nz{$A_2$} is labelled $f$ and should be understood
   to be connected to the input \nz{of $M$, which has the same label.}}
  \label{fig:reduction}
\end{figure}

Given such a formula $F$, we construct an unrooted binary network $N = N(F)$
such that $N$ \nz{has an HPO} if and only if $F$ is satisfiable.  We construct
$N$ from the following widgets (see \cref{fig:reduction}):
\begin{itemize}[noitemsep]
  \item A root widget $R$,
  \item For every variable $x_j$, $1 \le j \le n$, a choice widget $V_j$,
  \item For every clause $C_i$, $1 \le i \le m$, a clause widget $C_i$,
  \item An and-widget $A_1$ with $2n$ inputs,
  \item An and-widget $A_2$ with $m$ inputs,
  \item A replicator widget $M$ with $n$ outputs,
  \item For every literal $\lambda \in \{x_1, \bar x_1, \ldots, x_n, \bar x_n\}$
  with $k_\lambda \ge 1$, a basic replicator widget $M_\lambda$, and
  \item For every literal $\lambda \in \{x_1, \bar x_1, \ldots, x_n, \bar x_n\}$
  with $k_\lambda \ge 2$, a replicator widget $M_\lambda'$ with $k_\lambda$
  outputs.
\end{itemize}
We join these widgets as follows (using \nz{wire} widgets), making sure that
every input is joined to exactly one output and vice versa:
\begin{itemize}[noitemsep]
  \item The output of $A_1$ to the input of $R$,
  \item For all $1 \le j \le n$ and $\lambda \in \{x_j, \bar x_j\}$,
  \begin{itemize}[noitemsep]
    \item If $k_\lambda = 0$, the output of $V_j$ corresponding to $\lambda$ to
    an input of $A_1$,
    \item If $k_\lambda \ge 1$, the output of $V_j$ corresponding to $\lambda$
    to the input of $M_\lambda$ and an output of $M_\lambda$ to an input of
    $A_1$, and
    \item If $k_\lambda \ge 2$, an output of $M_\lambda$ to the input of
    $M_\lambda'$,
  \end{itemize}
  \item For all $1 \le j \le n$, an output of $M$ to the input of $V_j$,
  \item The output of $A_2$ to the input of $M$,
  \item For all $1 \le i \le m$, the output of $C_i$ to an input of $A_2$,
  \item For all $1 \le i \le m$, $1 \le j \le 3$,
  \begin{itemize}[noitemsep]
    \item If $k_{\lambda_{i,j}} = 1$, then an output of $M_{\lambda_{i,j}}$ to
    the $j$th input of $C_i$, and
    \item If $k_{\lambda_{i,j}} \ge 2$, then an output of $M_{\lambda_{i,j}}'$
    to the $j$th input of $C_i$.
  \end{itemize}
\end{itemize}

\begin{lem}
  \label{lem:SAT-to-orchard}
  If a \nz{Boolean formula $F$ in $3$-CNF} is satisfiable, then the
  corresponding network \nz{$N = N(F)$ has an HPO}.
\end{lem}

\begin{proof}
  Assume that $\phi$ is a truth assignment to the variables $x_1, \ldots, x_n$
  in $F$ that satisfies $F$.  \nz{We construct an HPO $\vec{N}$ of $N$ as follows:
  We choose}
  \begin{itemize}[noitemsep]
    \item The BPO of $R$,
    \item The BPO of $A_1$,
    \item The BPO of $A_2$,
    \item The {\true}-BPO of $M$,
    \item For every variable $x_j$, $1 \le j \le n$,
    \begin{itemize}[noitemsep]
      \item If $\phi(x_j) = \true$, the {\true}-BPO of $V_j$,
      \item If $\phi(x_j) = \false$, the {\false}-BPO of $V_j$,
    \end{itemize}
    \item For every clause $C_i$, the BPO of $C_i$ corresponding to the subset
    of literals $\lambda_{i,j}$ in $C_i$ such that $\phi(\lambda_{i,j}) =
    \true$,
    \item For every literal $\lambda \in \{x_1, \bar x_1, \ldots, x_n, \bar
    x_n\}$ such that $\phi(\lambda) = \true$,
    \begin{itemize}[noitemsep]
      \item If $k_\lambda \ge 1$, the {\true}-BPO of $M_\lambda$,
      \item If $k_\lambda \ge 2$, the {\true}-BPO of $M_{\lambda'}$,
    \end{itemize}
    \item For every literal $\lambda \in \{x_1, \bar x_1, \ldots, x_n, \bar
    x_n\}$ such that $\phi(\lambda) = \false$,
    \begin{itemize}[noitemsep]
      \item If $k_\lambda \ge 1$, the escape BPO of $M_\lambda$,
      \item If $k_\lambda \ge 2$, the {\false}-BPO of $M_{\lambda'}$,
    \end{itemize}
    \item The {\true}-BPO of any \nz{wire} widget joining $A_1$ to another
    widget (input or output),
    \item The {\true}-BPO of any \nz{wire} widget joining $A_2$ to another
    widget (input or output),
    \item The {\true}-BPO of any \nz{wire} widget joining a choice widget $V_j$
    to another widget (input or output),
    \item For every literal $\lambda \in \{x_1, \bar x_1, \ldots, x_n, \bar
    x_n\}$ such that $\phi(\lambda) = \true$, the {\true}-BPO of any \nz{wire}
    widget joining $M_\lambda$ or $M_\lambda'$ to another widget (input or
    output),
    \item For every literal $\lambda \in \{x_1, \bar x_1, \ldots, x_n, \bar
    x_n\}$ such that $\phi(\lambda) = \false$,
    \begin{itemize}[noitemsep]
      \item The {\false}-BPO of any \nz{wire} widget joining $M_\lambda$ to
      $M_\lambda'$ or to a clause widget $C_i$, and
      \item The {\false}-BPO of any \nz{wire} widget joining $M_\lambda'$ to a
      clause widget $C_i$.
    \end{itemize}
  \end{itemize}

  It is easily verified that $\vec{N}$ satisfies all of the conditions of
  \cref{obs:root-widget-boolean,obs:connector-boolean,obs:choice-widget-boolean,obs:and-widget-boolean,obs:replicator-widget-escape,obs:replicator-widget-true,obs:replicator-widget-false,obs:clause-widget-boolean}.
  We use this fact to prove that $\vec{N}$ satisfies
  \crefrange{prop:unique-root}{prop:cycle-free} and, thus, is an HPO of $N$.
  \yuki{By \Cref{obs:HGT-consistent} and \Cref{lem:HGT-consistent},} this proves
  that $N$ is an orchard.

  \paragraph{\labelcref{prop:unique-root}}

  By
  \cref{obs:connector-boolean,obs:choice-widget-boolean,obs:and-widget-boolean,obs:replicator-widget-escape,obs:replicator-widget-true,obs:replicator-widget-false,obs:clause-widget-boolean},
  only vertices in $R$ can be roots of $\vec{N}$.  Since the vertex $r$ in $R$
  has in-degree $0$, out-degree $1$, and no incident edges, and it is the only
  vertex in $R$ \nz{without in-arcs}, \cref{prop:unique-root} holds.

  \paragraph{\labelcref{prop:in-degree,prop:out-degree}}

  By
  \cref{obs:root-widget-boolean,obs:connector-boolean,obs:choice-widget-boolean,obs:and-widget-boolean,obs:replicator-widget-escape,obs:replicator-widget-true,obs:replicator-widget-false,obs:clause-widget-boolean},
  all vertices of $\vec{N}$ satisfy \cref{prop:in-degree,prop:out-degree}.

  \paragraph{\labelcref{prop:cycle-free}}

  By
  \cref{obs:root-widget-boolean,obs:connector-boolean,obs:choice-widget-boolean,obs:and-widget-boolean,obs:replicator-widget-escape,obs:replicator-widget-true,obs:replicator-widget-false,obs:clause-widget-boolean},
  any semi-directed cycle $H$ in $\vec{N}$ involves vertices from more than one
  widget.
  
  It is easy to see that any \nz{semi-directed} cycle that enters a \nz{wire}
  widget at one end (its input or output) must leave the \nz{wire} widget at the
  other end, due to the orientations of the edges in the BPOs of \nz{wire}
  widgets.  We can therefore conclude that any \nz{semi-directed} cycle that
  leaves some non-\nz{wire} widget via one of its \nz{inputs or outputs} must
  also enter the non-\nz{wire} widget joined to this input or output via a
  \nz{wire} widget.

  $R$ has no outputs and one input.  This input is joined to $A_1$ via the
  {\true}-BPO of a \nz{wire} widget.  Thus, for every vertex $x \in R$ with
  external neighbour $x'$, $\vec{N}$ contains the \nz{arc} $(x,x')$.  This
  implies that there is no edge \nz{or arc} via which $H$ could enter $R$, that
  is, $H$ contains no vertex in $R$.

  Similarly, all \nz{wire} widgets used to join $A_1$ to other widges are
  directed according to their {\true}-BPOs.  Thus, if~$H$ were to visit a vertex
  in $A_1$, it would have to enter $A_1$ via its output, and thus would also
  have to visit $R$.  Since we just proved that $H$ does not visit $R$, it does
  not visit $A_1$ either.

  Now assume that $H$ contains a vertex in some clause widget $C_i$.  Then $H$
  must leave $C_i$ via \nz{its output or via one of its inputs.} Since $\vec{N}$
  contains the {\true}-BPO \nz{of} the \nz{wire} widget joining $C_i$ to $A_2$,
  \nz{$H$ cannot leave $C_i$ via its output and, therefore, must do so via an
  input corresponding to a {\true} literal $\lambda_{i,j}$.}  Thus, $H$ also
  visits $M_{\lambda_{i,j}}$ or~$M_{\lambda_{i,j}}'$.  Since
  $\phi(\lambda_{i,j}) = \true$ $\vec{N}$ contains the {\true}-BPOs of all
  \nz{wire} widgets joining $M_{\lambda_{i,j}}'$ to other widgets.  Thus, if $H$
  visits \nz{$M_{\lambda_{i,j}}'$,} it can leave $M_{\lambda_{i,j}}'$ only via
  its input, and thus must also visit $M_{\lambda_{i,j}}$.  In other words,
  whether $M_{\lambda_{i,j}}'$ exists or not, $H$~must visit
  $M_{\lambda_{i,j}}$.

  Since $\phi(\lambda_{i,j}) = \true$, $\vec{N}$ contains the {\true}-BPOs of
  all \nz{wire} widgets joining $M_{\lambda_{i,j}}$ to other widgets.  Thus, the
  only \nz{edges or arcs} through which $H$ can leave $M_{\lambda_{i,j}}$ are
  the \nz{ones incident to the input} of $M_{\lambda_{i,j}}$.
  
  Assume that $\lambda_{i,j} = x_h$, for some $1 \le h \le n$.  (The case when
  $\lambda_{i,j} = \bar x_h$ is analogous.) Then since $\phi(\lambda_{i,j}) =
  \true$, $\vec{N}$ includes the {\true}-BPO of $V_h$.  Now observe that the two
  output vertices of $V_h$ corresponding to $x_h$ can \nz{reach only} leaves of
  $V_h$ in the {\true}-BPO of $V_h$.  Thus, $H$ cannot get back to the vertex in
  $C_i$ that we assumed it contains.  This proves that there is no cycle $H$ in
  $\vec{N}$ that includes a vertex in a clause widget.

  $\vec{N}$ contains the {\true}-BPOs of all \nz{wire} widgets joining $M$ and
  $A_2$ to other widgets.  Thus, if $H$ contains a vertex in $M$ or $A_2$, it
  must leave the subgraph composed of $M$, $A_2$, and the \nz{wire} widget
  joining them via \nz{an edge incident to one of the inputs} of $A_2$.  It must
  therefore visit a clause widget, which we just proved it cannot do.  Thus, $H$
  does not visit $M$ or $A_2$ either.

  For every replicator widget $M_\lambda'$, $\vec{N}$ contains the {\true}-BPOs
  of all \nz{wire} widgets joining $M_\lambda'$ to other widgets, or the
  {\false}-BPOs of all such \nz{wire} widgets.  Thus, $H$ can only enter
  $M_\lambda'$ via one of its outputs and leave it via its input, or enter it
  via its input and leave it via one of its outputs.  In either case, $H$ must
  visit a clause widget, which we proved it cannot do.  Thus, $H$ does not visit
  any replicator widget $M_\lambda'$.
  
  If $H$ visits a widget $M_\lambda$, then it must also visit at least two of
  the widgets to which $M_\lambda$ is joined using \nz{wire} widgets.  These
  neighbouring widgets are $V_j$, $M_\lambda'$ or a clause widget, and $A_1$.
  We already proved that $H$ cannot visit $A_1$, $M_\lambda'$ or a clause
  widget.  Thus, $H$ cannot visit $M_\lambda$ either.

  We conclude that $H$ must be confined to some choice widget $V_j$ because
  leaving $V_j$ would force it to visit some widget other than $V_j$.  Since
  there is no cycle in $\vec{N}$ containing only vertices in a choice widget, by
  \cref{obs:choice-widget-boolean}, we conclude that there cannot exist any
  semi-directed cycle $H$ in $\vec{N}$.
\end{proof}

\begin{lem}
  \label{lem:orchard-to-SAT}
  If the network \nz{$N = N(F)$} corresponding to a \nz{Boolean formula $F$ in
  $3$-CNF} is an orchard, then $F$ is satisfiable.
\end{lem}

\begin{proof}
  Let $\vec{N}$ be an HPO of $N$.  Since $N$ contains a root widget,
  \cref{lem:choice-widget} shows that every output of $M$ is {\true}.  By
  \cref{lem:replicator-widget}, this implies that the output of $A_2$ is
  {\true}.  By \cref{lem:and-widget}, this implies that the output of every
  clause widget $C_i$ is {\true}.

  By \cref{lem:clause-widget}, this implies that every clause $C_i$ contains a
  literal $\lambda_{i,j_i}$ such that the output of $M_{\lambda_{i,j_i}}$ or
  $M_{\lambda_{i,j_i}}'$ joined to the $j_i$th input of $C_i$ is {\true}. If
  $\lambda_{i,3}$ is such a literal, we choose $j_i = 3$; otherwise, we choose
  $j_i$ arbitrarily.  Let $h_i$ be the variable such that $\lambda_{i,j_i} \in
  \{x_{h_i}, \bar x_{h_i}\}$.

  If there are no two indices $1 \le a < b \le m$ such that $\lambda_{a,j_a} =
  \bar \lambda_{b,j_b}$, then we can use the literals $\lambda_{1,j_1},
  \ldots, \lambda_{m,j_m}$ to define a truth assignment $\phi$ that satisfies
  $F$.  In particular, for all $1 \le i \le m$, we set $\phi(x_{h_i}) = \true$
  if $\lambda_{i,j_i} = x_{h_i}$, and $\phi(x_{h_i}) = \false$ if
  $\lambda_{i,j_i} = \bar x_{h_i}$.  For any $h$ such that $h \notin \{h_1,
  \ldots, h_m\}$, we set $\phi(x_h) = \false$.

  This assignment $\phi$ is well-defined because if there are two indices $1 \le
  a < b \le m$ such that $h_a = h_b$, we must have $\lambda_{a,j_a} =
  \lambda_{b,j_b}$ because $\lambda_{a,j_a}, \lambda_{b,j_b} \in \{x_{h_a}, \bar
  x_{h_a}\}$ but $\lambda_{a,j_a} \ne \bar \lambda_{b,j_b}$.

  The assignment $\phi$ satisfies $F$ because its definition explicitly ensures
  that for each clause $C_i$, the literal $\lambda_{i,j_i} \in C_i$ satisfies
  $\phi(\lambda_{i,j_i}) = \true$.

  To finish the proof, we show that if there exist two indices $1 \le a < b \le
  m$ such that $\lambda_{a,j_a} = \bar \lambda_{b,j_b}$, then $\vec{N}$ contains
  a semi-directed cycle, contradicting \cref{prop:cycle-free}.  Thus, this case
  cannot arise.

  Since the outputs of $M$ are all {\true}, there exists a semi-directed path
  from each output vertex of $M$ to every output vertex of $A_2$, by
  \cref{lem:replicator-widget}.

  Since the output of $A_2$ is {\true}, there exists a semi-directed path from
  each output vertex of $A_2$ to each output vertex of each clause widget $C_i$,
  by \cref{lem:and-widget}.

  Since the output of each clause widget $C_i$ is true, and by the choice of
  $\lambda_{i,j_i}$, there exist semi-directed paths from the output vertices of
  $C_i$ to the output vertices of $M_{\lambda_{i,j_i}}$ or
  $M_{\lambda_{i,j_i}}'$ joined to the $j_i$th input of $C_i$, by
  \cref{lem:clause-widget}.

  If the $j_i$th input of $C_i$ is joined to an output of
  $M_{\lambda_{i,j_i}}'$, then this output of $M_{\lambda_{i,j_i}}'$ is {\true}.
  Thus, there exist semi-directed paths from these two output vertices to the
  two output vertices of $M_{\lambda_{i,j_i}}$, and the output of
  $M_{\lambda_{i,j_i}}$ joined to the input of $M_{\lambda_{i,j_i}}'$ is
  {\true}, by \cref{lem:replicator-widget}.

  Finally, since the output of $M_{\lambda_{i,j_i}}$ joined to $C_i$ or
  $M_{\lambda_{i,j_i}}'$ is {\true}, there exist semi-directed paths from these
  two output vertices of $M_{\lambda_{i,j_i}}$ to the two output vertices of
  $V_{h_i}$ that correspond to $\lambda_{i,j_i}$, by
  \cref{lem:replicator-widget} again.

  By composing the semi-directed paths we have identified so far, we obtain a
  semi-directed path from each output vertex of $M$ to each output vertex of
  $V_{h_i}$ corresponding to $\lambda_{i,j_i}$, for each $1 \le i \le m$.  In
  particular, since $\lambda_{a,j_a} = \bar \lambda_{b,j_b}$, we have
  semi-directed paths from each output vertex of $M$ to both $c_{x_{h_a}}$ and
  $c_{\bar x_{h_a}}$.

  By \cref{lem:choice-widget}, we have semi-directed paths from $c_{x_{h_a}}$
  or $c_{\bar x_{h_a}}$ to the two output vertices of $M$ joined to the input of
  $V_{h_a}$.  Together with the semi-directed paths from these two output
  vertices of $M$ to $c_{x_{h_a}}$ and $c_{\bar x_{h_a}}$, we obtain a
  semi-directed cycle in $\vec{N}$, which is the desired contradiction.
\end{proof}

\begin{thm}
  \label{thm:main}
  \textsc{Orchard Recognition} and \textsc{Orchard Orientation} are NP-hard.
\end{thm}

\begin{proof}
  By \cref{lem:SAT-to-orchard,lem:orchard-to-SAT}, the network $N(F)$ \nz{has an
  HPO} if and only if $F$ is satisfiable.  \nz{By
  \cref{thm:drawing,obs:HGT-consistent}, this shows that $N(F)$ has an orchard
  orientation if and only if $F$ is satisfiable.  Thus,} since the construction of
  $N(F)$ from $F$ is easily seen to take polynomial time, it constitutes a
  polynomial-time reduction from \textsc{3-SAT} to \nz{\textsc{Orchard Orientation}}.
  Since \textsc{3-SAT} is NP-hard, so is \nz{\textsc{Orchard Orientation}}.

  \nz{By \cref{thm:orchard-orientation}, $N(F)$ is an orchard if and only if it
  has an orchard orientation.  Thus, since \textsc{Orchard Orientation} is
  NP-hard, so is \textsc{Orchard Recognition}.}
\end{proof}

\section{Discussion}

\yuki{We have shown that the membership problem \textsc{Orchard Recognition} and
the orientation problem \textsc{Orchard Orientation} are equivalent and both
NP-hard, via reduction from \textsc{3-SAT}.  \textsc{Orchard Orientation}
belongs to a broader family of problems, called
\textsc{$\mathcal{C}$-Orientation}, where~$\mathcal{C}$ is some class of
phylogenetic networks.  Given an undirected graph, it asks whether this graph
can be oriented to produce a rooted network that belongs to $\C$.  As mentioned
in \Cref{sec:intro}, \textsc{$\C$-orientation} is NP-hard when $\mathcal{C}$ is
chosen to be the class of tree-based networks~\cite{huber2024orienting}, and our
paper shows that the same is true when $\C$ is the class of orchards.

Another important class of networks is the class of tree-child networks, which
are networks in which every non-leaf vertex has a child that is a tree vertex or
a leaf.  \textsc{Tree-Child Orientation} is known to be NP-hard when
\leo{the input unrooted network is allowed degree-$d$ nodes with~$d=1,3,5$}~\cite[Theorem 3.8]{docker2024existence}
(the original idea for the reduction allowed for degree-$d$ nodes with $1\le d \le 5$ and restricting the root node at a specified position~\cite[Corollary
5]{bulteau2023turning}. This reduction was shown to be incorrect~\cite{docker2024existence}).  
Certain graph families are known not to admit tree-child
orientations, and necessary conditions for tree-child orientability have been
studied~\cite{maeda2023orienting}.  FPT algorithms based on the level of a
network~\cite{huber2024orienting}, the number of reticulations, and the maximal
size of so-called ``minimal basic cycles''~\cite{urata2024orientability} have
been proposed.  The computational complexity of \textsc{Tree-Child Orientation}
\leo{for binary networks is still} open.
Perhaps a modification of our reduction to \textsc{Orchard Orientation} could
provide new insights here.}

\nz{It is also interesting to ask what makes unrooted orchards harder to
recognize than rooted ones.  On the surface, the definitions of rooted and
unrooted orchards are virtually identical:  Cherries are defined analogously for
rooted and unrooted networks.  Ignoring the edge directions in a rooted network
turns every reticulated cherry into a 2-chain.  Reducing cherries and
reticulated cherries or 2-chains applies essentially the same changes to the
network whether it is rooted or unrooted.  Thus, it is rather surprising that
unrooted \textsc{Orchard Recognition} is NP-hard while the problem can be solved
in polynomial time for rooted
networks~\cite{erdHos2019class,janssen2021cherry}.

The primary reason for this difference in difficulty seems to be that for every
reducible pair $\p{x, y}$ of a rooted network $N$, any reducible pair $\p{u, v}
\ne \p{x, y}$ of $N$ is also a reducible pair of $\red{N}{\p{x, y}}$ (reducing
reducible pairs does not destroy other reducible pairs).  Thus, we can
effectively reduce reducible pairs in any order and still succeed in reducing
$N$ if $N$ is indeed an orchard.  This is at the heart of the polynomial-time
recognition of rooted orchards.

Unrooted networks do not enjoy the same ``stability'' of reducible pairs.
Consider a network $N$ obtained from a cycle $C = \s{v_0, v_1, \ldots, v_n =
v_0}$ by attaching a leaf $u_i$ to each vertex $v_i$, for all $i \in [n]$.  Then
$\p{u_{i-1}, u_i}$ is a 2-chain of~$N$, for all $i \in [n]$, but only
$\p{u_{i-2}, u_{i-1}}$ and $\p{u_i, u_{i+1}}$ are reducible pairs of
$\red{N}{\p{u_{i-1}, u_i}}$.  This suggests that the choice as to which reducible pair
to reduce can have non-local effects on the remainder of the reduction that are
difficult to predict at the time when this choice is made.  The example in
\cref{fig:real-sheep}---the original ``sheep'' that gave this paper its name
and which contains the root widget used in this paper as a subgraph---may
illustrate this issue more clearly.}

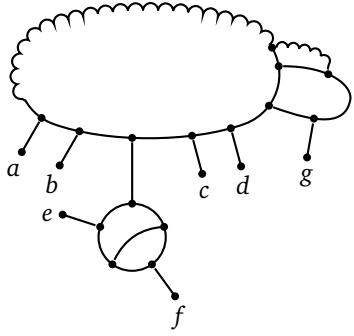
\begin{figure}[t]
  \centering
  \begin{tikzpicture}
    \pic (sheep) {real sheep};
    \path
    (sheep-rear-leg-3) node [anchor=north,xshift=-3pt,yshift=-1pt] {$a$}
    (sheep-rear-leg-4) node [anchor=north,xshift=-3pt]             {$b$}
    (sheep-in-3)       node [anchor=north,xshift=1pt,yshift=-1pt]  {$c$}
    (sheep-in-4)       node [anchor=north,xshift=1pt]              {$d$}
    (sheep-p-7)        node [anchor=east]                          {$e$}
    (sheep-p-8)        node [anchor=north,xshift=2pt]              {$f$}
    (sheep-beard-2)    node [anchor=north,yshift=-1pt]             {$g$};
  \end{tikzpicture}
  \caption{The actual sheep $S$.  This network has two reducible pairs $\p{a,
  b}$ and $\p{c, d}$.  Reducing $\p{a, b}$ renders $\p{c, d}$ irreducible and
  does not create any new reducible pairs.  Nevertheless, $S$ is an orchard, as
  the sequence $\s{\p{c, d}, \p{d, g}, \p{d, g}, \p{d, g}, \p{g, a}, \p{a, b},
  \p{b, c}, \p{c, e}, \p{c, e}, \p{c, e}, \p{e, f}}$ reduces it.}
  \label{fig:real-sheep}
\end{figure}

\nz{Given the NP-hardness of \textsc{Orchard Orientation}, an important question
to ask is whether it is fixed-parameter tractable.} \textsc{Orchard Orientation}
is known to be FPT with respect to the level of the network. Is it also FPT with
respect to other popular parameters such as treewidth?  \nz{Are there classes of
graphs for which \textsc{Orchard Orientation} can be solved in polynomial time?
A related question is whether there exists a characterization of
orchard-orientable graphs based on forbidden minors.  Notably, even for rooted
orchards, no characterization based on forbidden structures is known to date, so
any findings in this area would be beneficial.}


\yuki{Finally, we can consider relaxing one of the key assumptions for defining
unrooted orchards.  We defined 2-chains as a pair of leaves whose neighbours are
connected by an edge that is not a cut edge.  Suppose we remove this cut edge
condition. Then reducing a 2-chain can potentially disconnect the network into
two connected components.  We can then pose the following optimization problem.}

\yuki{%
\begin{problem}{Minimum Component Orchard Sequence (MCOS)}
  Find a sequence~$S$ of reductions for a given network~$N$, that minimizes the
  number of connected components in~$\red{N}{S}$ such that each connected component contains a
  single vertex.
\end{problem}

The proof of \cref{thm:orchard-orientation} shows that $S$ reduces an
unrooted orchard $N_u$ if and only if there exists an orchard orientation $N_d$
of $N_u$.  Thus, it follows immediately that an unrooted network~$N$ has a
MCOS~$S$ for which $\red{N}{S}$ has only one connected component if and only if~$N$ is an
orchard.}  In general, an optimal solution to \textsc{MCOS} can give an
indication of how far an unrooted network is from being an orchard.
The number of connected components represents the number of roots we would need to introduce in a generalized orchard orientation that allows for multiple roots.

\end{document}